\newif\ifdraft\draftfalse
\newif\iffull\fulltrue

\documentclass[conference]{IEEEtran}

\usepackage{amsmath,amssymb}

\usepackage{amsthm}

\theoremstyle{plain}
\newtheorem{theorem}{Theorem}[section]
\newtheorem{lemma}[theorem]{Lemma}
\newtheorem{proposition}[theorem]{Proposition}

\theoremstyle{definition}
\newtheorem{definition}{Definition}[section]
\newtheorem{example}{Example}[section]

\usepackage{txfonts}

\usepackage{booktabs}   
\usepackage{subcaption} 
\usepackage{bussproofs, mathpartir} 
\usepackage{graphicx}
\graphicspath{ {../../images/} }
\usepackage{thmtools} 
\usepackage{thm-restate} 
\usepackage{xspace}
\usepackage{multicol}
\usepackage{array}
\usepackage{pgfplots}
\usepackage{tikz,tikz-cd}

\usepackage[numbers]{natbib}
\usepackage{hyperref}
\usepackage[capitalise]{cleveref} 

\usetikzlibrary{shapes.geometric}
\usetikzlibrary{positioning}
\usetikzlibrary{arrows.meta}
\usetikzlibrary{patterns}

\newcommand{\R}[1]{\mathbf{Mem}[#1]}
\newcommand{\m}[1]{\mathbf{Mem}[#1]}
\newcommand{\empmem}{\langle\rangle}

\newcommand{\DR}[1]{\mathcal{D}(\R{#1})}
\newcommand{\PR}[1]{\mathcal{P}(\R{#1})}
\newcommand{\TR}[1]{\mathcal{T}(\R{#1})}
\newcommand{\T}{\mathcal{T}}

\newcommand{\dm}[1]{D_{#1}}
\newcommand{\rg}[1]{R_{#1}}
\newcommand{\dom}{\ensuremath{\mathbf{dom}}}
\newcommand{\range}{\ensuremath{\mathbf{range}}}

\newcommand{\DD}{\mathcal{D}}
\newcommand{\PP}{\mathcal{P}}

\newcommand{\defeq}{\mathrel{\mathop:}=}

\newcommand{\wh}[1]{\widehat{#1}}

\usepackage{wrapfig}
\def\ci{\perp\!\!\!\perp}
\newcommand{\pearlindep}[3]{\ensuremath{#1 \ci #2 \mid #3}}
\newcommand{\I}[3]{\ensuremath{I(#1, #2, #3)}}

\newcommand{\jia}[1]{\ifdraft{\textcolor{brown}{[Jialu: {#1}]}}\fi}

\newcommand*{\myalign}[2]{\multicolumn{1}{#1}{#2}} 

\usepackage{mathtools}
\makeatletter
\def\@acmplainindent{0pt}
\def\@acmdefinitionindent{0pt}
\def\@proofindent{\noindent}
\makeatother
\newcommand{\Var}{\ensuremath{\mathsf{Var}}}
\newcommand{\Exp}{\ensuremath{\mathsf{Exp}}}
\newcommand{\Com}{\ensuremath{\mathsf{Com}}}

\newcommand{\Val}{\ensuremath{\mathbf{Val}}}

\newcommand{\denot}[1]{\ensuremath{\llbracket #1 \rrbracket}}
\newcommand{\dbind}{\bind}
\newcommand{\dunit}{\unit}
\newcommand{\dcond}[2]{\ensuremath{{#1} \mid {#2}}}
\newcommand{\dconv}[3]{\ensuremath{{#2} \mathbin{\oplus_{#1}} {#3}}}
\newcommand{\kconv}[3]{\ensuremath{{#2} \mathbin{\overline{\oplus}_{#1}} {#3}}}

\newcommand{\ktt}{\ensuremath{\mathit{tt}}}
\newcommand{\kff}{\ensuremath{\mathit{ff}}}
\newcommand{\FV}{\ensuremath{\text{FV}}}
\newcommand{\RV}{\ensuremath{\text{RV}}}
\newcommand{\WV}{\ensuremath{\text{WV}}}
\newcommand{\MV}{\ensuremath{\text{MV}}}
\newcommand{\FFV}{\ensuremath{\text{FV}_{\text{D}}}}
\newcommand{\SFV}{\ensuremath{\text{FV}_{\text{R}}}}

\newcommand{\Skip}{\mathbf{skip}}
\newcommand{\Seq}[2]{{#1} \mathrel{;} {#2}}
\newcommand{\Assn}[2]{\ensuremath{{#1} \leftarrow {#2}}}
\newcommand{\Rand}[2]{{#1} \stackrel{\raisebox{-.25ex}[.25ex]%
{\tiny $\mathdollar$}}{\raisebox{-.2ex}[.2ex]{$\leftarrow$}} {#2}}

\newcommand{\Cond}[3]{\mathbf{if}\ #1\ \mathbf{then}\ #2\ \mathbf{else}\ #3}
\newcommand{\Condt}[2]{\mathbf{if}\ #1\ \mathbf{then}\ #2}

\newcommand\Field{\textsf{Field}}
\newcommand\Researcher{\textsf{Researcher}}
\newcommand\Conference{\textsf{Conference}}

\newcommand{\LOGIC}{DIBI\xspace}
\newcommand{\RLOGIC}{R\LOGIC}
\newcommand{\SYSTEM}{CPSL\xspace}
\newcommand{\ExONE}{\textsc{CommonCause}\xspace}
\newcommand{\ExTWO}{\textsc{CondSamples}\xspace}
\newcommand\sepid{\ensuremath{I}}

\usepackage{enumitem}
\usepackage{bbding}
\newcommand\sepand{\mathrel{\ast}}
\newcommand\sepimp{\mathrel{-\mkern-6mu*}}

\usepackage{color, colortbl}
\usepackage{stmaryrd}
\newcommand\depand{\fatsemi}
\newcommand\depimpr{\mathrel{\multimap}}
\newcommand\depimpl{\mathrel{\multimapinv}}

\newcommand{\MD}{\ensuremath{M^{D}}}
\newcommand{\MP}{\ensuremath{M^{P}}}

\newcommand{\M}{M}

\newcommand{\rseq}{\mathbin{-\mkern-2mu\triangleright}} 
\newcommand{\lseq}{\mathbin{\triangleright\mkern-2mu-}}
\newcommand{\PrimePredicate}[2]{P(#1) = \begin{cases} 1 & \text{if } #2 \\ 
				0 & \text{otherwise} \end{cases}}

\newcommand{\Bern}{\ensuremath{\mathbf{B}}}
\newcommand{\Dist}{\ensuremath{}}
\newcommand{\Exact}[1]{#1}
\newcommand{\kp}[2]{\ensuremath{({#1} \mathrel{\triangleright} {#2})}}
\newcommand{\pair}[2]{\kp{\Exact{#1}}{\Dist{[#2]}}}
\newcommand{\empdom}[1]{\kp{\Exact{\emptyset}}{#1}}
\newcommand{\indep}[3]{ [#1] \depand ([#2] \sepand [#3])}

\newcommand{\psl}[3]{\{ #1 \}\ #2\ \{ #3 \}}
\newcommand\bind{\textsf{bind}}
\newcommand\unit{\textsf{unit}}
\newcommand\Kl{\mathcal{K}\ell}
\newcommand{\Frestrict}{\mathrm{Form_{\RLOGIC}}}

\begin{document}

\title{A Bunched Logic for Conditional Independence}


\IEEEoverridecommandlockouts
\IEEEpubid{\makebox[\columnwidth]{978-1-6654-4895-6/21/\$31.00~
\copyright2021 IEEE \hfill} \hspace{\columnsep}\makebox[\columnwidth]{ }}

\author{
\IEEEauthorblockN{Jialu Bao}
\IEEEauthorblockA{University of Wisconsin--Madison} \\
\IEEEauthorblockN{Justin Hsu}
\IEEEauthorblockA{University of Wisconsin--Madison} \and
\IEEEauthorblockN{Simon Docherty}
\IEEEauthorblockA{University College London} \\
\IEEEauthorblockN{Alexandra Silva}
\IEEEauthorblockA{University College London}
}

\maketitle
\thispagestyle{plain}
\pagestyle{plain}

\begin{abstract}
  \emph{Independence} and \emph{conditional independence} are fundamental
  concepts for reasoning about groups of random variables in probabilistic
  programs. Verification methods for independence are still nascent, and
  existing methods cannot handle conditional independence. We extend the logic
  of bunched implications (BI) with a non-commutative conjunction and provide a
  model based on Markov kernels; conditional independence can be directly
  captured as a logical formula in this model. Noting that Markov kernels are
  Kleisli arrows for the distribution monad, we then introduce a second model
  based on the powerset monad and show how it can capture \emph{join
  dependency}, a non-probabilistic analogue of conditional independence from
  database theory. Finally, we develop a program logic for verifying
  conditional independence in probabilistic programs.
\end{abstract}

\section{Introduction} \label{sec:intro}
The study of probabilistic programming languages and their semantics dates back
to the 1980s, starting from the seminal work of \citet{kozen81}. The last decade
has seen a surge of richer probabilistic
languages~\cite{tabular,church,anglican}, motivated by applications in machine
learning, and accompanying research into their semantics
\cite{tasson,kammar,dahlqvist}. This burst of activity has also created new
opportunities and challenges for formal verification.

Independence and conditional independence are two fundamental properties that
are poorly handled by existing verification methods. Intuitively, two random
variables are {\em probabilistically independent} if information about one gives
no information about the other (for example, the results of two coin flips).
{\em Conditional independence} is more subtle: two random variables $X$ and $Y$
are independent conditioned on a third variable $Z$ if for every fixed value of
$Z$, information about one of $X$ and $Y$ gives no information about the other.

Both forms of independence are useful for modelling and verification.
{\em Probabilistic independence} enables compositional reasoning about groups of
random variables: if a group of random variables are independent, then their
joint distribution is precisely described by the distribution of each variable
in isolation. It also captures the semantics of random sampling constructs in
probabilistic languages, which generate a fresh random quantity that is
independent of the program state.
{\em Conditional independence} often arises in programs with probabilistic
control flow, as conditioning models probabilistic branching. Bayesian networks
encode conditional independence statements in complex distributions, and
conditional independence captures useful properties in many applications.  For
instance, criteria ensuring that algorithms do not discriminate based on
sensitive characteristics (e.g., gender or race) can be formulated using
conditional independence~\citep{barocas-hardt-narayanan}.

Aiming to prove independence in probabilistic programs,
\citet{barthe2019probabilistic} recently introduced Probabilistic Separation
Logic (PSL) and applied it to formalize security for several well-known
constructions from cryptography. The key ingredient of PSL is a new model of the
logic of bunched implications (BI), in which separation is interpreted as
probabilistic independence. While PSL enables formal reasoning about
independence, it does not support conditional independence. The core issue is
that the model of BI underlying PSL provides no means to describe the
distribution of one set of variables obtained by fixing (\emph{conditioning})
another set of variables to take specific values. Accordingly, one cannot
capture the basic statement of conditional independence---$X$ and $Y$ are
independent random variables conditioned on any value of $Z$.

In this paper, we develop a logical framework for \emph{formal reasoning about
notions of dependence and independence}. Our approach is inspired by PSL but the
framework is more sophisticated: to express conditional independence, we develop
a novel assertion logic extending BI with new connectives---$\depand$ and its
adjoints. The key intuition is that conditional independence can be expressed as
independence plus composition of \emph{Markov kernels}; as our leading example,
we give a kernels model of our logic.

Then, we show how to adapt the probabilistic model to other settings. As is
well-known in category theory, Markov kernels are the arrows in the Kleisli
category of the distribution monad. By varying the monad, our logic smoothly
extends to analogues of conditional independence in other domains. To
demonstrate, we show how replacing the distribution monad by the {\em powerset}
monad gives a model where we can capture \emph{join}/\emph{multivalued
dependencies} in relational algebra and database theory. We also show that the
\emph{semi-graphoid laws}, introduced by \citet{pearl1985graphoids} in their
work axiomatizing conditional independence, can be translated into formulas that
are valid in both of our models.

The rest of the paper is organized as follows. We give a bird's-eye view in
\Cref{sec:overview}, providing intuitions on our design choices and highlighting
differences with existing work. \Cref{sec:logic} presents the main contribution:
the design of \LOGIC, a new bunched logic to reason about dependence and
independence. We show that the proof system of \LOGIC is sound and complete with
respect to its Kripke semantics. Then, we present two concrete models in
\Cref{sec:models}, based on probability distributions and relations. In
\Cref{sec:CI}, we consider how to express dependencies in \LOGIC: we show that
the same logical formula captures conditional independence and join dependency
in our two models, and our models validate the semi-graphoid laws. Finally, in
\Cref{sec:cpsl}, we design a program logic with \LOGIC assertions, and use it to
verify conditional independence in two probabilistic programs.

\section{Overview of the contributions}\label{sec:overview}

\paragraph*{The logic \LOGIC}
The starting point of our work is the logic of bunched implications
(BI)~\cite{o1999logic}. BI extends intuitionistic propositional logic with
substructural connectives to facilitate reasoning about sharing and separation
of resources, an idea most prominently realized in Separation Logic's handling
of heap-manipulating programs~\citep{DBLP:conf/csl/OHearnRY01}. The novel
connectives are a \emph{separating conjunction} $P \sepand Q$, intuitively
stating that $P$ and $Q$ hold in separate resources, and its adjoint $\sepimp$,
called \emph{magic wand}. We will extend BI with a non-commutative conjunction,
written $P \depand Q$.  Intuitively, $\depand$ expresses a {\em possible
dependency} of $Q$ on $P$. The end result is a logic with two conjunctive
connectives---$\sepand$ and $\depand$---capturing notions of {\em independence}
and {\em dependence}. We call the logic {\em Dependence and Independence Bunched
Implications (DIBI)}.

To give a semantics to our logic, we start from the semantics of BI. The
simplest BI models are {\em partial resource monoids}: Kripke structures $(M,
\sqsubseteq, \circ, e)$ in which $\circ$ is an order-preserving, partial,
commutative monoid operation with unit $e$. The operation $\circ$ allows
interpreting the separating conjunction $P \sepand Q$ and magic wand $P \sepimp
Q$. For example, the probabilistic model of BI underlying
PSL~\citep{barthe2019probabilistic} is a partial resource monoid: by taking $M$
to be the set of distributions over program memories and $\circ$ to be the
independent product of distributions over memories with disjoint variables, the
interpretation of $P\sepand Q$ gives the desired notion of {\em probabilistic
independence}.

This is the first point where we fundamentally differ from PSL. To capture both
dependence and independence, we change the structure in which formulas are
interpreted. In \Cref{sec:logic}, we will introduce a structure $\mathcal{X} =
(X, \sqsubseteq, \oplus, \odot, E)$, a \emph{\LOGIC frame}, with two operations
$\oplus \colon X^2 \to \mathcal P(X)$ and $\odot \colon X^2 \to \mathcal P(X)$,
and a set of units $E\subseteq X$. Three remarks are in order. First, the
preorder $\sqsubseteq$ makes \LOGIC an \emph{intuitionistic} logic. There are
many design trade-offs between intuitionistic and classical, but the most
important consideration is that intuitionistic formulas can describe proper
subsets of states (e.g., random variables), leaving the rest of the state
implicit. Second, \LOGIC frames contain an additional monoidal operation $\odot$
for interpreting $\depand$ ($\oplus$ will be used in interpreting $\sepand$).
Third, as the completeness of BI for its simple PCM models is an open
problem~\cite{Galmiche2019}, our models are examples of a broader notion of BI
model with non-deterministic operations
(following~\cite{docherty2019bunched,Galmiche2006}). These models subsume
partial resource monoids, and enable our completeness proof of \LOGIC. While the
conditions that \LOGIC frames must satisfy are somewhat cryptic at first sight,
they can be naturally understood as axioms defining monoidal operations in a
partial, non-deterministic setting. E.g., we will require:
	\[\small
	\begin{array}{ll}
		\text{($\oplus$ Comm.)} & z \in x \oplus y \rightarrow z \in y \oplus x; \\
		\text{($\oplus$ Assoc.)} & w \in t \oplus z \land t \in x \oplus y \rightarrow
		\exists s(s \in y \oplus z \land w \in x \oplus s); \\
		\text{($\odot$ Unit Exist.\textsubscript{L})} & \exists {e \in E}.\ (x \in e \odot x)
	\end{array} \]
where unbound variables are universally quantified. Crucially, the operation
$\odot$ need \emph{not} be commutative: this operation interprets the dependence
conjunction $\depand$, where commutativity is undesirable. In a \LOGIC frame,
$\sepand$ and $\depand$ are interpreted as:
{\small
\begin{align*}
  x &\models P \sepand Q &&\text{  iff exists $x', y, z$ s.t. $x \sqsupseteq x' \in y \oplus z$, $y \models P$, and  $z \models Q$} \\
  x &\models P \depand Q &&\text{  iff exists $y, z$ s.t. $x \in y \odot z$, $y \models P$, and  $z \models Q$}
\end{align*}
}%
In \LOGIC, $\sepand$ has a similar reading as in PSL: it states that two parts
of a distribution can be combined because they are {\em independent}. In
contrast, the new conjunction $P \depand Q$ asserts that the $Q$ part of a
distribution {\em may depend} on the $P$ part. Combined with the separating
conjunction $\sepand$, the new conjunction $\depand$ can express more complex
dependencies: e.g.  $P \depand (Q \sepand R)$ asserts that $Q$ and $R$ both may
depend on $P$, and are independent conditioned on $P$.

\paragraph*{A sound and complete proof system for \LOGIC}
To reason about \LOGIC validity, in \Cref{sec:logic} we also provide a
Hilbert-style proof system for \LOGIC, and prove soundness and completeness.
The proof system extends BI with rules for the new connective $\depand$, e.g.
$\depand$ \textsc{Conj}, and for the interaction between $\depand$ and
$\sepand$, e.g., \textsc{RevEx}:
\begin{mathpar}\footnotesize
	\inferrule* [right=$\depand$ \textsc{Conj}]
	{ P \vdash R \\ Q \vdash S }
	{ P \depand Q \vdash R \depand S}
	\and
	\inferrule* [right= \textsc{RevEx}]
	{~}
	{ (P \depand Q) \sepand (R \depand S) \vdash (P \sepand R) \depand (Q \sepand S) }
\end{mathpar}
\textsc{RevEx}---{\em reverse-exchange}---captures the fundamental interaction
between the two conjunctions. Computations $T=P\depand Q$ and $U=R\depand S$ are
built from dependent components, yet $T$ and $U$ are independent and hence can
be combined with $\sepand$. We can then infer that the building blocks of $T$
and $U$ must also be pair-wise independent and can be combined, yielding
formulas $P \sepand R$ and $Q \sepand S$. These can then be combined with
$\depand$ as they retain the dependency of the original building blocks.

\paragraph*{Models and applications of \LOGIC}
Separation logics are based on a concrete BI model over program states, together
with a choice of atomic assertions. Before explaining the models of \LOGIC, we
recall two prior models of BI.

In the heap model, states are heaps: partial maps from memory addresses to
values. Atomic assertions of the form $x \mapsto v$ indicate that the location
to which $x$ points has value $v$. Then, $x \mapsto v \sepand y \mapsto u$
states that $x$ points to $v$ and $y$ points to $u$, and $x$ and $y$ do not
alias---they must point to different locations. In general, $P \sepand Q$ holds
when a heap can be split into two subheaps with disjoint domains, satisfying $P$
and $Q$ respectively.


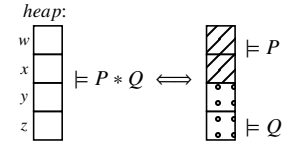
\tikzset{
pattern size/.store in=\mcSize,
pattern size = 5pt,
pattern thickness/.store in=\mcThickness,
pattern thickness = 0.3pt,
pattern radius/.store in=\mcRadius,
pattern radius = 1pt}
\makeatletter
\pgfutil@ifundefined{pgf@pattern@name@_w9n4bpx7p}{
\pgfdeclarepatternformonly[\mcThickness,\mcSize]{_w9n4bpx7p}
{\pgfqpoint{0pt}{0pt}}
{\pgfpoint{\mcSize+\mcThickness}{\mcSize+\mcThickness}}
{\pgfpoint{\mcSize}{\mcSize}}
{
\pgfsetcolor{\tikz@pattern@color}
\pgfsetlinewidth{\mcThickness}
\pgfpathmoveto{\pgfqpoint{0pt}{0pt}}
\pgfpathlineto{\pgfpoint{\mcSize+\mcThickness}{\mcSize+\mcThickness}}
\pgfusepath{stroke}
}}
\makeatother


\tikzset{
pattern size/.store in=\mcSize,
pattern size = 5pt,
pattern thickness/.store in=\mcThickness,
pattern thickness = 0.3pt,
pattern radius/.store in=\mcRadius,
pattern radius = 1pt}
\makeatletter
\pgfutil@ifundefined{pgf@pattern@name@_p7pvh2bcc}{
\makeatletter
\pgfdeclarepatternformonly[\mcRadius,\mcThickness,\mcSize]{_p7pvh2bcc}
{\pgfpoint{-0.5*\mcSize}{-0.5*\mcSize}}
{\pgfpoint{0.5*\mcSize}{0.5*\mcSize}}
{\pgfpoint{\mcSize}{\mcSize}}
{
\pgfsetcolor{\tikz@pattern@color}
\pgfsetlinewidth{\mcThickness}
\pgfpathcircle\pgfpointorigin{\mcRadius}
\pgfusepath{stroke}
}}
\makeatother


\tikzset{
pattern size/.store in=\mcSize,
pattern size = 5pt,
pattern thickness/.store in=\mcThickness,
pattern thickness = 0.3pt,
pattern radius/.store in=\mcRadius,
pattern radius = 1pt}
\makeatletter
\pgfutil@ifundefined{pgf@pattern@name@_yaoys05py}{
\makeatletter
\pgfdeclarepatternformonly[\mcRadius,\mcThickness,\mcSize]{_yaoys05py}
{\pgfpoint{-0.5*\mcSize}{-0.5*\mcSize}}
{\pgfpoint{0.5*\mcSize}{0.5*\mcSize}}
{\pgfpoint{\mcSize}{\mcSize}}
{
\pgfsetcolor{\tikz@pattern@color}
\pgfsetlinewidth{\mcThickness}
\pgfpathcircle\pgfpointorigin{\mcRadius}
\pgfusepath{stroke}
}}
\makeatother

  \begin{center}
\tikzset{
pattern size/.store in=\mcSize,
pattern size = 5pt,
pattern thickness/.store in=\mcThickness,
pattern thickness = 0.3pt,
pattern radius/.store in=\mcRadius,
pattern radius = 1pt}
\makeatletter
\pgfutil@ifundefined{pgf@pattern@name@_aqts6sq02}{
\pgfdeclarepatternformonly[\mcThickness,\mcSize]{_aqts6sq02}
{\pgfqpoint{0pt}{0pt}}
{\pgfpoint{\mcSize+\mcThickness}{\mcSize+\mcThickness}}
{\pgfpoint{\mcSize}{\mcSize}}
{
\pgfsetcolor{\tikz@pattern@color}
\pgfsetlinewidth{\mcThickness}
\pgfpathmoveto{\pgfqpoint{0pt}{0pt}}
\pgfpathlineto{\pgfpoint{\mcSize+\mcThickness}{\mcSize+\mcThickness}}
\pgfusepath{stroke}
}}
\makeatother
\tikzset{every picture/.style={line width=0.75pt}} 
 \scalebox{0.8}{$\!\!\!\!\!\!\!\!$
\begin{tikzpicture}[x=0.75pt,y=0.75pt,yscale=-1,xscale=1]

\draw   (201,84) -- (219,84) -- (219,102) -- (201,102) -- cycle ;
\draw   (201,120) -- (219,120) -- (219,138) -- (201,138) -- cycle ;
\draw   (201,138) -- (219,138) -- (219,156) -- (201,156) -- cycle ;
\draw   (201,102) -- (219,102) -- (219,120) -- (201,120) -- cycle ;
\draw  [pattern=_w9n4bpx7p,pattern size=6pt,pattern thickness=0.75pt,pattern radius=0pt, pattern color={rgb, 255:red, 0; green, 0; blue, 0}] (310,84) -- (328,84) -- (328,102) -- (310,102) -- cycle ;
\draw  [pattern=_p7pvh2bcc,pattern size=6pt,pattern thickness=0.75pt,pattern radius=0.75pt, pattern color={rgb, 255:red, 0; green, 0; blue, 0}] (310,138) -- (328,138) -- (328,156) -- (310,156) -- cycle ;
\draw  [pattern=_yaoys05py,pattern size=6pt,pattern thickness=0.75pt,pattern radius=0.75pt, pattern color={rgb, 255:red, 0; green, 0; blue, 0}] (310,120) -- (328,120) -- (328,138) -- (310,138) -- cycle ;
\draw  [pattern=_aqts6sq02,pattern size=6pt,pattern thickness=0.75pt,pattern radius=0pt, pattern color={rgb, 255:red, 0; green, 0; blue, 0}](310,102) -- (328,102) -- (328,120) -- (310,120) -- cycle ;

\node at (195,94) {\footnotesize$w$};
\node at (195,112) {\footnotesize$x$};
\node at (195,130) {\footnotesize$y$};
\node at (195,148) {\footnotesize$z$};

\draw (193,69) node [anchor=north west][inner sep=0.75pt]    {\small $heap$:};
\draw (225,111) node [anchor=north west][inner sep=0.75pt]    {$\models P*Q  \iff $};
\draw (333,91) node [anchor=north west][inner sep=0.75pt]    {$\models P$};
\draw (333,141) node [anchor=north west][inner sep=0.75pt]    {$\models Q$};
\end{tikzpicture}}
 \end{center}

In PSL, states are distributions over program memories, basic assertions
$\mathbf D[x]$ indicate that $x$ is a random variable, and $P \sepand Q$ states
that a distribution $\mu$ can be factored into two \emph{independent}
distributions $\mu_1$ and $\mu_2$ satisfying $P$ and $Q$, respectively. Consider
the following simple program:
\begin{equation}\label{eq:simple}
	\Rand x {\Bern_{1/2}}; \Rand y {\Bern_{1/2}} ; z \gets x \vee y
\end{equation}
Here, $x$ and $y$ are Boolean variables storing the result of two fair coin
flips and $z$ stores the result of $x\vee y$. The output distribution $\mu$ is a
distribution over a memory with variables $x$, $y$ and $z$ (depicted below on
the right). In $\mu$, the variables $x$ and $y$ are independent and $\mathbf
D[x] \sepand \mathbf D[y]$ holds, since the marginal distribution of $\mu$ is a
product of $\mu_1$ and $\mu_2$, which satisfy $\mathbf D[x]$ and $\mathbf D[y]$
respectively:
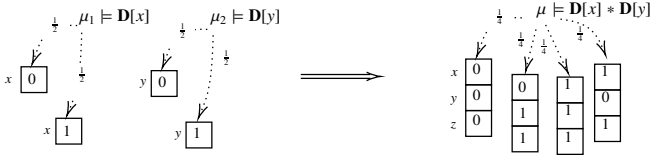
\begin{center}
\tikzset{every picture/.style={line width=0.75pt}} 
 \scalebox{0.65}{
\begin{tikzpicture}[x=0.75pt,y=0.75pt,yscale=-1,xscale=1]

\draw   (489,418) -- (509,418) -- (509,438) -- (489,438) -- cycle ;
\draw   (489,458) -- (509,458) -- (509,478) -- (489,478) -- cycle ;
\draw   (489,438) -- (509,438) -- (509,458) -- (489,458) -- cycle ;
\draw  [dash pattern={on 0.84pt off 2.51pt}]  (532.5,386) .. controls (505.62,384.08) and (503.62,399.67) .. (498.2,414.19) ;
\draw [shift={(497.5,416)}, rotate = 291.8] [color={rgb, 255:red, 0; green, 0; blue, 0 }  ][line width=0.75]    (10.93,-3.29) .. controls (6.95,-1.4) and (3.31,-0.3) .. (0,0) .. controls (3.31,0.3) and (6.95,1.4) .. (10.93,3.29)   ;
\draw  [dash pattern={on 0.84pt off 2.51pt}]  (562.5,387) .. controls (582.87,389.91) and (589.13,393.76) .. (596.78,413.16) ;
\draw [shift={(597.5,415)}, rotate = 249.15] [color={rgb, 255:red, 0; green, 0; blue, 0 }  ][line width=0.75]    (10.93,-3.29) .. controls (6.95,-1.4) and (3.31,-0.3) .. (0,0) .. controls (3.31,0.3) and (6.95,1.4) .. (10.93,3.29)   ;
\draw  [dash pattern={on 0.84pt off 2.51pt}]  (543.5,395) -- (535.3,427.05) ;
\draw [shift={(535.5,427.5)}, rotate = 282.8] [color={rgb, 255:red, 0; green, 0; blue, 0 }  ][line width=0.75]    (10.93,-3.29) .. controls (6.95,-1.4) and (3.31,-0.3) .. (0,0) .. controls (3.31,0.3) and (6.95,1.4) .. (10.93,3.29)   ;
\draw   (525,430) -- (545,430) -- (545,450) -- (525,450) -- cycle ;
\draw   (525,450) -- (545,450) -- (545,470) -- (525,470) -- cycle ;
\draw   (525,470) -- (545,470) -- (545,490) -- (525,490) -- cycle ;
\draw   (589,422) -- (609,422) -- (609,442) -- (589,442) -- cycle ;
\draw   (589,462) -- (609,462) -- (609,482) -- (589,482) -- cycle ;
\draw   (589,442) -- (609,442) -- (609,462) -- (589,462) -- cycle ;
\draw  [dash pattern={on 0.84pt off 2.51pt}]  (551,391.5) -- (567.65,427.19) ;
\draw [shift={(568.5,429)}, rotate = 244.98000000000002] [color={rgb, 255:red, 0; green, 0; blue, 0 }  ][line width=0.75]    (10.93,-3.29) .. controls (6.95,-1.4) and (3.31,-0.3) .. (0,0) .. controls (3.31,0.3) and (6.95,1.4) .. (10.93,3.29)   ;
\draw   (560,432) -- (580,432) -- (580,452) -- (560,452) -- cycle ;
\draw   (560,472) -- (580,472) -- (580,492) -- (560,492) -- cycle ;
\draw   (560,452) -- (580,452) -- (580,472) -- (560,472) -- cycle ;
\draw   (144,424) -- (164,424) -- (164,444) -- (144,444) -- cycle ;
\draw  [dash pattern={on 0.84pt off 2.51pt}]  (187.5,392) .. controls (160.62,390.08) and (158.62,405.67) .. (153.2,420.19) ;
\draw [shift={(152.5,422)}, rotate = 291.8] [color={rgb, 255:red, 0; green, 0; blue, 0 }  ][line width=0.75]    (10.93,-3.29) .. controls (6.95,-1.4) and (3.31,-0.3) .. (0,0) .. controls (3.31,0.3) and (6.95,1.4) .. (10.93,3.29)   ;
\draw   (171,464) -- (191,464) -- (191,484) -- (171,484) -- cycle ;
\draw  [dash pattern={on 0.84pt off 2.51pt}]  (191,397.5) .. controls (190.04,439.74) and (185.85,446.5) .. (180.21,460.24) ;
\draw [shift={(179.5,462)}, rotate = 291.8] [color={rgb, 255:red, 0; green, 0; blue, 0 }  ][line width=0.75]    (10.93,-3.29) .. controls (6.95,-1.4) and (3.31,-0.3) .. (0,0) .. controls (3.31,0.3) and (6.95,1.4) .. (10.93,3.29)   ;
\draw   (245,427) -- (265,427) -- (265,447) -- (245,447) -- cycle ;
\draw  [dash pattern={on 0.84pt off 2.51pt}]  (288.5,395) .. controls (261.62,393.08) and (259.62,408.67) .. (254.2,423.19) ;
\draw [shift={(253.5,425)}, rotate = 291.8] [color={rgb, 255:red, 0; green, 0; blue, 0 }  ][line width=0.75]    (10.93,-3.29) .. controls (6.95,-1.4) and (3.31,-0.3) .. (0,0) .. controls (3.31,0.3) and (6.95,1.4) .. (10.93,3.29)   ;
\draw   (272,467) -- (292,467) -- (292,487) -- (272,487) -- cycle ;
\draw  [dash pattern={on 0.84pt off 2.51pt}]  (293,399.5) .. controls (292.04,441.74) and (287.85,448.5) .. (282.21,462.24) ;
\draw [shift={(281.5,464)}, rotate = 291.8] [color={rgb, 255:red, 0; green, 0; blue, 0 }  ][line width=0.75]    (10.93,-3.29) .. controls (6.95,-1.4) and (3.31,-0.3) .. (0,0) .. controls (3.31,0.3) and (6.95,1.4) .. (10.93,3.29)   ;
\draw    (361,430) -- (415,430)(361,433) -- (415,433) ;
\draw [shift={(423,431.5)}, rotate = 180] [color={rgb, 255:red, 0; green, 0; blue, 0 }  ][line width=0.75]    (10.93,-3.29) .. controls (6.95,-1.4) and (3.31,-0.3) .. (0,0) .. controls (3.31,0.3) and (6.95,1.4) .. (10.93,3.29)   ;

\draw (543,376) node [anchor=north west][inner sep=0.75pt]    {$\mu$};
\draw  [draw opacity=0][fill={rgb, 255:red, 255; green, 255; blue, 255 }  ,fill opacity=1 ]  (506.22,377) -- (524.22,377) -- (524.22,402) -- (506.22,402) -- cycle  ;
\draw (515.22,378) node [anchor=north] [inner sep=0.75pt]  [font=\scriptsize]  {$\frac{1}{4}$};
\draw (527,396) node [anchor=north west][inner sep=0.75pt]  [font=\scriptsize]  {$\frac{1}{4}$};
\draw (573,391) node [anchor=north west][inner sep=0.75pt]  [font=\scriptsize]  {$\frac{1}{4}$};
\draw (555,372) node [anchor=north west][inner sep=0.75pt]    {$\models  \mathbf D[x] \sepand  \mathbf D[y]$};
\draw (493,420) node [anchor=north west][inner sep=0.75pt]  [font=\small]  {$0$};
\draw (493,440) node [anchor=north west][inner sep=0.75pt]  [font=\small]  {$0$};
\draw (493,460) node [anchor=north west][inner sep=0.75pt]  [font=\small]  {$0$};
\draw (528.5,434) node [anchor=north west][inner sep=0.75pt]  [font=\small]  {$0$};
\draw (529,454) node [anchor=north west][inner sep=0.75pt]  [font=\small]  {$1$};
\draw (529,474) node [anchor=north west][inner sep=0.75pt]  [font=\small]  {$1$};
\draw (592.5,422) node [anchor=north west][inner sep=0.75pt]  [font=\small]  {$1$};
\draw (593,442) node [anchor=north west][inner sep=0.75pt]  [font=\small]  {$0$};
\draw (593,463) node [anchor=north west][inner sep=0.75pt]  [font=\small]  {$1$};
\draw (545,401) node [anchor=north west][inner sep=0.75pt]  [font=\scriptsize]  {$\frac{1}{4}$};
\draw (563.5,432) node [anchor=north west][inner sep=0.75pt]  [font=\small]  {$1$};
\draw (564,452) node [anchor=north west][inner sep=0.75pt]  [font=\small]  {$1$};
\draw (564,473) node [anchor=north west][inner sep=0.75pt]  [font=\small]  {$1$};
\draw  [draw opacity=0][fill={rgb, 255:red, 255; green, 255; blue, 255 }  ,fill opacity=1 ]  (161.22,383) -- (179.22,383) -- (179.22,408) -- (161.22,408) -- cycle  ;
\draw (170.22,384) node [anchor=north] [inner sep=0.75pt]  [font=\scriptsize]  {$\frac{1}{2}$};
\draw (148,428) node [anchor=north west][inner sep=0.75pt]  [font=\small]  {$0$};
\draw (476,425) node [anchor=north west][inner sep=0.75pt]  [font=\footnotesize]  {$x$};
\draw (476,445) node [anchor=north west][inner sep=0.75pt]  [font=\footnotesize]  {$y$};
\draw (476,465) node [anchor=north west][inner sep=0.75pt]  [font=\footnotesize]  {$z$};
\draw (130,430) node [anchor=north west][inner sep=0.75pt]  [font=\footnotesize]  {$x$};
\draw  [draw opacity=0][fill={rgb, 255:red, 255; green, 255; blue, 255 }  ,fill opacity=1 ]  (182.22,422) -- (200.22,422) -- (200.22,447) -- (182.22,447) -- cycle  ;
\draw (191.22,423) node [anchor=north] [inner sep=0.75pt]  [font=\scriptsize]  {$\frac{1}{2}$};
\draw (175,468) node [anchor=north west][inner sep=0.75pt]  [font=\small]  {$1$};
\draw (160,470) node [anchor=north west][inner sep=0.75pt]  [font=\footnotesize]  {$x$};
\draw  [draw opacity=0][fill={rgb, 255:red, 255; green, 255; blue, 255 }  ,fill opacity=1 ]  (260.22,388) -- (278.22,388) -- (278.22,413) -- (260.22,413) -- cycle  ;
\draw (269.22,389) node [anchor=north] [inner sep=0.75pt]  [font=\scriptsize]  {$\frac{1}{2}$};
\draw (249,430) node [anchor=north west][inner sep=0.75pt]  [font=\small]  {$0$};
\draw (235,432) node [anchor=north west][inner sep=0.75pt]  [font=\footnotesize]  {$y$};
\draw (276,470) node [anchor=north west][inner sep=0.75pt]  [font=\small]  {$1$};
\draw (262,472) node [anchor=north west][inner sep=0.75pt]  [font=\footnotesize]  {$y$};
\draw (188,378) node [anchor=north west][inner sep=0.75pt]    {$\mu _{1} \models \mathbf D[x]$};
\draw (290,378) node [anchor=north west][inner sep=0.75pt]    {$\mu _{2}\models \mathbf D[y]$};

\draw  [draw opacity=0][fill={rgb, 255:red, 255; green, 255; blue, 255 }  ,fill opacity=1 ]  (294.22,409) -- (312.22,409) -- (312.22,434) -- (294.22,434) -- cycle  ;
\draw (303.22,410) node [anchor=north] [inner sep=0.75pt]  [font=\scriptsize]  {$\frac{1}{2}$};
\end{tikzpicture}}
\vspace{-.2cm}
\end{center}

In \Cref{sec:models}, we develop two concrete models for \LOGIC: one based on
probability distributions, and one based on relations. Here we outline the
probabilistic model, as it generalizes the model of PSL. Let $\Val$ be a finite
set of values and $S$ a finite set of memory locations. We use $\R S$ to denote
functions $S\to \Val$, representing program memories. The states in the \LOGIC
probabilistic model, over which the formulas will be interpreted, are Markov
kernels on program memories. More precisely, given sets of memory locations $S
\subseteq U$, these are functions $f\colon \R{S} \rightarrow \DR{U}$ that
preserve their input. Regular distributions can be lifted to Markov kernels: the
distribution $\mu \colon \DR U$ corresponds to the kernel $f_\mu \colon
\R{\emptyset} \to \DR U$ that assigns $\mu$ to the only element in
$\R{\emptyset}$.
\begin{wrapfigure}{r}{1.8cm}
\scalebox{.6}{$\!\!\!\!\!\!\!\!\!$	\begin{tikzpicture}[scale=0.40]
			\coordinate (A1) at (-8,0) ;
			\coordinate (B1) at (-8,3) ;
			\coordinate (C1) at (-6,3.5) ;
			\coordinate (D1) at (-6,-1) ;
			\draw [thick,-,fill=blue,pattern=north west lines, pattern color=green] (A1) -- (B1)node[left, midway] {\small$\dom(f)$}   -- (C1) -- node[right, midway] {\small$\range(f)$}  (D1) -- cycle ;
		\end{tikzpicture}}
		\end{wrapfigure}
We depict input-preserving Markov kernels as trapezoids, where the smaller side
represents the domain and the larger side the range; our basic assertions will
track $\dom(f)$ and $\range(f)$, justifying this simplistic depiction.

Separating and dependent conjunction will be interpreted via $\oplus$ and
$\odot$ on Markov kernels. Intuitively, $\oplus$ is a parallel composition that
takes union on both domains and ranges, whereas $\odot$ composes the kernels
using Kleisli composition.
	\begin{multicols}{2}
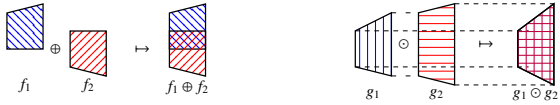
		\begin{center}
		\scalebox{.6}{
			\begin{tikzpicture}[scale=0.40]
				\coordinate (A) at (0,0) ;
				\coordinate (B) at (0,2) ;
				\coordinate (C) at (2,2.5) ;
				\coordinate (D) at (2,0) ;
				\draw [thick,-,fill=blue,pattern=north west lines, pattern color=blue] (A) -- (B)  -- (C) -- (D) -- cycle ;
				\node[] at (2.7,0)  (otimes) {$\oplus$};
				\coordinate (E) at (3.5,-1) ;
				\coordinate (F) at (3.5,1.0) ;
				\coordinate (G) at (5.5,1.0) ;
				\coordinate (H) at (5.5,-1.5) ;
				\draw [thick,-, fill=black, pattern=north east lines, pattern color=red] (E) -- (F)  -- (G) -- (H)  -- cycle;
				\node[] at (7.5,0) {$\mapsto$};
				\coordinate (J) at (9,2) ;
				\coordinate (L) at (9,0) ;
				\coordinate (M) at (11,2.5) ;
				\coordinate (O) at (11, 0) ;
				\draw [thick,-,pattern=north west lines, pattern color=blue] (L) -- (J)   -- (M) -- (O)  -- cycle;
				\coordinate (J1) at (9,-1) ;
				\coordinate (M1) at (11,-1.5) ;
				\coordinate (O1) at (11,1) ;
				\coordinate (L1) at (9,1) ;
				\draw [thick,-,pattern=north east lines, pattern color=red] (L1) -- (J1)   -- (M1) -- (O1)  -- cycle;
				\node[] at (1,-2)  (a) {$f_1$};
				\node[] at (4.5,-2)  (b) {$f_2$};
				\node[] at (10,-2.1)  (b) {$f_1 \oplus f_2$};
			\end{tikzpicture}}
		\end{center}
		\columnbreak
		\begin{center}
				\scalebox{.6}{
			\begin{tikzpicture}[scale=0.4]
				\coordinate (A) at (0,2.5) ;
				\coordinate (B) at (0,4.5) ;
				\coordinate (C) at (2,5.5) ;
				\coordinate (D) at (2,2) ;
				\draw [thick,pattern=vertical lines,pattern color=blue] (A) -- (B) -- (C) -- (D) -- cycle ;
				\node[] at (2.7,3.75)  (circ) {$\odot$};
				\coordinate (E) at (3.5,2) ;
				\coordinate (F) at (3.5,5.5) ;
				\coordinate (G) at (5.5, 6) ;
				\coordinate (H) at (5.5, 1.5) ;
				\draw [thick,pattern=horizontal lines,pattern color=red] (E) -- (F)  -- (G) -- (H)  -- cycle;
				\node[] at (7.25,3.75) {$\mapsto$};
				\draw [dashed] (0,2.5) -- (9.6,2.5);
				\draw [dashed] (0,4.5) -- (9.6,4.5);
				\draw [dashed] (5.5,6) -- (11, 6);
				\draw [dashed] (5.5, 1.5) -- (11, 1.5);
				\draw [loosely dashed] (2,5.5) -- (3.5, 5.5);
				\draw [loosely dashed] (2, 2) -- (3.5, 2);
				\coordinate (I) at (9, 2.5) ;
				\coordinate (J) at (9, 4.5) ;
				\coordinate (M) at (11, 6) ;
				\coordinate (N) at (11, 1.5) ;
				\draw [thick,pattern=horizontal lines,pattern color=purple] (I) -- (J)  -- (M) -- (N) -- cycle;
				\draw [thick,pattern=vertical lines,pattern color=purple] (I) -- (J)  -- (M) -- (N) -- cycle;
				\node[] at (1,1)  (a) {$g_1$};
				\node[] at (4.5,1)  (b) {$g_2$};
				\node[] at (10,1)  (c) {$g_1 \odot g_2$};
			\end{tikzpicture}}
		\end{center}
	\end{multicols}
To demonstrate, recall the simple program \eqref{eq:simple}. In the output
distribution $\mu$, $z$ depends on $x$ and $y$ since $z$ stores $x \lor y$, and
$x$ and $y$ are independent. In our setting, this dependency structure can be
seen when decomposing $f_\mu = (f_{\mu_1} \oplus f_{\mu_2}) \odot f_z$, where
kernel $f_z \colon \R{\{x,y\}} \to \DR{\{x,y,z\}}$ captures how the value of $z$
depends on the values of $\{ x, y \}$:
\begin{center}
\tikzset{every picture/.style={line width=0.75pt}} 
\scalebox{.75}{
\begin{tikzpicture}[x=0.75pt,y=0.75pt,yscale=-1,xscale=1]
\begin{scope}[shift={(100,-210)}]
\draw   (210,620) -- (225,620) -- (225,635) -- (210,635) -- cycle ;
\draw   (210,635) -- (225,635) -- (225,650) -- (210,650) -- cycle ;
\draw   (320,620) -- (335,620) -- (335,635) -- (320,635) -- cycle ;
\draw   (320,635) -- (335,635) -- (335,650) -- (320,650) -- cycle ;
\draw   (320,650) -- (348,650) -- (348,665) -- (320,665) -- cycle ;

\draw (212,623) node [anchor=north west][inner sep=0.75pt]  [font=\small]  {$a$};
\draw (212,637) node [anchor=north west][inner sep=0.75pt]  [font=\small]  {$b$};
\draw (202,622) node [anchor=north west][inner sep=0.75pt]  [font=\scriptsize,xslant=0.13]  {$x$};
\draw (202,636) node [anchor=north west][inner sep=0.75pt]  [font=\scriptsize]  {$y$};
\draw (230,623) node [anchor=north west][inner sep=0.75pt]  [font=\large]  {$\xmapsto{\quad f_z\quad} $};
\draw (294,620) node [anchor=north west][inner sep=0.75pt]  [font=\large]  {$\delta \Bigg($};
\draw (350,620) node [anchor=north west][inner sep=0.75pt]  [font=\large]  {$\Bigg)$ };
\draw (311,655) node [anchor=north west][inner sep=0.75pt]  [font=\scriptsize]  {$z$};
\draw (320,650) node [anchor=north west][inner sep=0.75pt]  [font=\small]  {$a\lor b$};
\draw (311,637) node [anchor=north west][inner sep=0.75pt]  [font=\scriptsize]  {$y$};
\draw (311,620) node [anchor=north west][inner sep=0.75pt]  [font=\scriptsize,xslant=0.13]  {$x$};
\draw (323,635) node [anchor=north west][inner sep=0.75pt]  [font=\small]  {$b$};
\draw (323,622) node [anchor=north west][inner sep=0.75pt]  [font=\small]  {$a$};
\draw (390,625) node [anchor=north west][font=\large] {$\delta\colon X\to \mathcal D(X)$ is the Dirac distribution};
\draw (390,645) node [anchor=north west][font=\large] {$\delta(v)(w) = 1$ if $v = w$, 0 otherwise.};
\end{scope}
\end{tikzpicture}}
\end{center}
We can then prove:
\begin{equation}\label{eq:dep_simple}
	f_{\mu_1} \oplus f_{\mu_2} \models P_{x\sepand y}\quad \text{and} \quad f_z \models Q_z
\quad\text{implies}\quad f_\mu \models  P_{x\sepand y} \depand Q_z
\end{equation}
When analyzing composition of Markov kernels, the domains and ranges provide key
information: the domain determines which variables a kernel may depend on, and
the range determines which variables a kernel describes. Accordingly, we use
basic assertions of the form $\pair{A}{B}$, where $A$ and $B$ are sets of memory
locations. A Markov kernel $f \colon \R S \to \DR T$ satisfies \pair{A}{B} if
there exists a $f'\sqsubseteq f$ with $\dom(f') = A$ and $\range(f') \supseteq
B$ (we will define $f' \sqsubseteq f$ formally later and for now read it as $f$
extends $f'$). For instance, the kernel $f_z$ above satisfies
$\pair{\{x,y\}}{x,y}$, $\pair{\{x,y\}}{x,y,z}$, and $\pair{\{x,y\}}{\emptyset
}$. One choice for $P_{x\sepand y}$ and $Q_z$ in \eqref{eq:dep_simple} can be:
\(
P_{x\sepand y} = \pair{\emptyset}{x} \sepand \pair{\emptyset}{y} \text{ and }
Q_z = \pair{\{x,y\}}{x,y,z})
\)
\begin{figure*}[!t]
  \begin{small}
	\[\begin{array}{llcl}
		\text{($\oplus$ Down-Closed)} &  z \in x \oplus y \land x \sqsupseteq x' \land y \sqsupseteq y'& \rightarrow
        & \exists z' (z \sqsupseteq z' \land z' \in x' \oplus y'); \\
		\text{($\odot$ Up-Closed)} &  z \in x \odot y \land z' \sqsupseteq z &\rightarrow
		    & \exists x', y' (x' \sqsupseteq x \land y' \sqsupseteq y \land z' \in x' \odot y') \\
		\text{($\oplus$ Commutativity)} & z \in x \oplus y &\rightarrow& z \in y \oplus x; \\
		\text{($\oplus$ Associativity)} & w \in t \oplus z \land t \in x \oplus y &\rightarrow
		    & \exists s(s \in y \oplus z \land w \in x \oplus s); \\
		\text{($\oplus$ Unit Existence)} & \exists e \in E(x \in e \oplus x); \\
		\text{($\oplus$ Unit Coherence)} & e \in E \land x \in y \oplus e &\rightarrow& x \sqsupseteq y; \\
		\text{($\odot$ Associativity)} & \exists t(w \in t \odot z \land t \in x \odot y) &\leftrightarrow
		    & \exists s (s \in y \odot z \land w  \in x \odot s); \\
		\text{($\odot$ Unit Existence\textsubscript{L})} & \exists e \in E(x \in e \odot x); \\
		\text{($\odot$ Unit Existence\textsubscript{R})} & \exists e \in E(x \in x \odot e); \\
		\text{($\odot$ Coherence\textsubscript{R})} & e \in E \land x \in y \odot e &\rightarrow &x \sqsupseteq y; \\
		\text{(Unit Closure)} & e \in E \land e' \sqsupseteq e& \rightarrow &e' \in E; \\
		\text{(Reverse Exchange)} & x\in y \oplus z \land y \in y_1 \odot y_2 \land z \in z_1 \odot z_2 &\rightarrow
				& \exists u, v(u \in y_1 \oplus z_1 \land v  \in y_2 \oplus z_2 \land x
 \in u \odot v ).
	\end{array} \] 
  \end{small}
  \caption{DIBI frame requirements (with outermost universal quantification omitted for readability).}\label{fig:dibi_rules}
\end{figure*}

\paragraph*{Formalizing conditional independence}
The reader might wonder how to use such simple atomic propositions, which only
talk about the domain/range of a kernel and do not describe numeric
probabilities, to assert conditional independence. The key insight is that
conditional independence can be formulated using sequential ($\odot$) and
parallel ($\oplus$) composition of kernels. In \Cref{sec:CI}, we show that given
$\mu \in \DR{\Var}$, for any $X, Y, Z \subseteq \Var$, the satisfaction of
\begin{equation}\label{eq:ci}
  f_\mu \models \pair{\emptyset}{Z} \depand \pair{Z}{X} \sepand \pair{Z}{Y}
\end{equation}
captures conditional independence of $X,Y$ given $Z$ in $\mu$.

Moreover, the formula in \eqref{eq:ci} smoothly generalizes to other models. In
the relational model of \LOGIC---obtained by switching the distribution monad to
the powerset monad---the exact same formula encodes \emph{join dependency}, a
notion of conditional independence from the databases and relational algebra
literature. More generally, we also show that the semi-graphoid axioms
of~\citet{pearl1985graphoids} are valid in these two models, and two of the
axioms can be derived in the \LOGIC proof system.

\section{The Logic \LOGIC} \label{sec:logic}
\subsection{Syntax and semantics}\label{sec:logic-synt}

The syntax of \LOGIC extends the logic of bunched implications (BI)~\cite{o1999logic} with a non-commutative conjunctive connective~$\depand$ and its associated implications. Let $\mathcal{AP}$ be a set of propositional atoms. The set of \LOGIC formulas, $\mathrm{Form_{\LOGIC}}$, is generated by the following grammar:
\begin{align*}
	P, Q &::= p \in \mathcal{AP}
	\mid \top
	\mid \sepid
	\mid \bot
	\mid P \land Q
	\mid P \lor Q
	\mid P \rightarrow Q \\
  &\mid P \sepand Q
	\mid P \sepimp Q
	\mid P \depand Q
	\mid P \depimpr Q
	\mid P \depimpl Q.
\end{align*}
\LOGIC is interpreted on \LOGIC frames, which extend BI frames.

\begin{definition}[\LOGIC Frame]
  A \emph{\LOGIC frame} is a structure $\mathcal{X} = (X, \sqsubseteq, \oplus,
  \odot, E)$ such that $\sqsubseteq$ is a preorder, $E \subseteq X$, and
  $\oplus\colon X^2 \rightarrow \mathcal{P}(X)$ and $\odot\colon X^2 \rightarrow
  \mathcal{P}(X)$ are binary operations, satisfying the rules in \Cref{fig:dibi_rules}.
  \end{definition}

Intuitively, $X$ is a set of states, the preorder $\sqsubseteq$ describes when
a smaller state can be extended to a larger state, the binary operators $\odot$,
$\oplus$ offer two ways of combining states, and $E$ is the set of states that
act like units with respect to these operations. The binary operators return a
\emph{set} of states instead of a single state, and thus can be either
deterministic (at most one state returned) or non-deterministic, either partial
(empty set returned) or total. The operators in the concrete models below will
be deterministic, but the proof of completeness relies on the frame's admission
of non-deterministic models, as is standard for bunched
logics~\citep{docherty2019bunched}.

The frame conditions define properties that must hold for all models of \LOGIC.
Most of these properties can be viewed as generalizations of familiar algebraic
properties to non-deterministic operations, suitably interacting with the
preorder. The ``Closed'' properties give coherence conditions between the order
and the composition operators. It is known that having the Associativity frame
condition together with either the Up- or Down-Closed property for an operator
is sufficient to obtain the soundness of associativity for the conjunction
associated with the operator~\cite{Cao2017, docherty2019bunched}. The choices
of Closed conditions match the desired interpretations of $\oplus$ as
independence and $\odot$ as dependence: independence should drop down to
substates (which must necessarily be independent if the superstates were), while
dependence should be inherited by superstates (the source of dependence will
still be present in any extensions). Having $\odot$ non-commutative also splits
the $\odot$ analogues of $\oplus$ axioms into pairs of axioms, although we note
that we exclude the left version of ($\odot$ Coherence) for reasons we explain
in Section \ref{subsec:proofsystem}. Finally, the (Reverse Exchange) condition
defines the interaction between $\oplus$ and $\odot$.

We will give a Kripke-style semantics for \LOGIC, much like the semantics for
BI~\citep{Pym2004}. Given a \LOGIC frame, the semantics defines which states in
the frame \emph{satisfy} each formula. Since the definition is inductive on
formulas, we must specify which states satisfy the atomic propositions.

\begin{definition}[Valuation and model]
A \emph{persistent valuation} is an assignment $\mathcal{V}\colon \mathcal{AP} \rightarrow \mathcal{P}(X)$ of atomic propositions to subsets of states of a \LOGIC frame satisfying: if $x \in \mathcal{V}(p)$ and $y \sqsupseteq x$ then $y \in \mathcal{V}(p)$. A \emph{\LOGIC model} $(\mathcal{X}, \mathcal{V})$ is a \LOGIC frame $\mathcal{X}$ together with a persistent valuation $\mathcal{V}$.
\end{definition}
Since \LOGIC is an intuitionistic logic, persistence is necessary for soundness. We can now give a semantics to \LOGIC formulas in a \LOGIC model.

\begin{figure*}[!ht]\small
	\centering
	\begin{tabular}{c c c c l r c c c c r r c}
		$x$ & $\models_{\mathcal{V}}$ & $\top$ & & \multicolumn{2}{l}{always $\qquad \qquad \qquad \quad \, x \quad \models_{\mathcal{V}} \quad \bot \qquad$ never}  \\
		$x$ & $\models_{\mathcal{V}}$ & $\sepid$ & iff & \multicolumn{2}{r}{$x \in E\qquad \qquad\qquad \qquad x \quad \models_{\mathcal{V}} \quad \mathrm{p}\quad $ iff $x \in \mathcal{V}(\mathrm{p})$} \\
		$x$ & $\models_{\mathcal{V}}$ & $P \land Q$ & iff & \myalign{l}{$x \models_{\mathcal{V}} P$ and $x \models_{\mathcal{V}} Q$} \\
		$x$ & $\models_{\mathcal{V}}$ & $P \lor Q$ & iff & $x \models_{\mathcal{V}} P$ or $x \models_{\mathcal{V}} Q$ \\
		$x$ & $\models_{\mathcal{V}}$ & $P \rightarrow Q$ & iff & \myalign{l}{for all $y \sqsupseteq x$, $y \models_{\mathcal{V}} P$ implies $y$  $\models_{\mathcal{V}}$ $Q$} \\
		$x$ & $\models_{\mathcal{V}}$ & $P \sepand Q$ & iff & \multicolumn{8}{l}{there exist $x', y, z$ s.t. $x \sqsupseteq x' \in y \oplus z$, $y \models_{\mathcal{V}} P$ and  $z \models_{\mathcal{V}} Q$} \\
		$x$ & $\models_{\mathcal{V}}$ & $P \depand Q$ & iff & \multicolumn{8}{l}{there exist $y, z$ s.t. $x \in y \odot z$, $y \models_{\mathcal{V}} P$ and  $z \models_{\mathcal{V}} Q$} \\
		$x$ & $\models_{\mathcal{V}}$ & $P \sepimp Q$ & iff & \multicolumn{8}{l}{for all $y, z$ s.t. $z \in x \oplus y$: $y \models_{\mathcal{V}} P$ implies $z \models_{\mathcal{V}} Q$} \\
		$x$ & $\models_{\mathcal{V}}$ & $P \depimpr Q$ & iff & \multicolumn{8}{l}{for all $x', y, z$ s.t. $x' \sqsupseteq x$ and $z \in x' \odot y$: $y \models_{\mathcal{V}} P$ implies $z \models_{\mathcal{V}} Q$} \\
		$x$ & $\models_{\mathcal{V}}$ & $P \depimpl Q$ & iff & \multicolumn{8}{l}{for all $x', y, z$ s.t. $x' \sqsupseteq x$ and $z \in y \odot x'$:  $y \models_{\mathcal{V}} P$ implies $z \models_{\mathcal{V}} Q$}
	\end{tabular}
	\caption{Satisfaction for \LOGIC}
	\label{fig:sat-dibi}
\end{figure*}

\begin{definition}[\LOGIC Satisfaction and Validity]
	Satisfaction at a state $x$ in a model is inductively defined by the clauses in \Cref{fig:sat-dibi}. $P$ is \emph{valid in a model}, $\mathcal{X} \models_{\mathcal{V}} P$, iff $x \models_{\mathcal{V}} P$ for all $x \in \mathcal{X}$. $P$ is \emph{valid}, $\models P$, iff $P$ is valid in all models. $P \models Q$ iff, for all models, $\mathcal{X} \models_{\mathcal{V}} P$ implies $\mathcal{X} \models_{\mathcal{V}} Q$.
\end{definition}

Where the context is clear, we omit the subscript $\mathcal{V}$ on the
satisfaction relation. With the semantics in \Cref{fig:sat-dibi},
persistence on propositional atoms extends to all formulas:

\begin{lemma}[Persistence Lemma]
	For all $P \in \mathrm{Form_{\LOGIC}}$, if $x \models P$ and $y \sqsupseteq x$ then $y \models P$. 
\end{lemma}

The reader may note the difference between the semantic clauses for $\depand$ and $\sepand$, and $\sepimp$ and $\depimpr$: the satisfaction of the Up-Closed (Down-Closed) frame axiom for $\odot$ ($\oplus$) leads to the persistence and thus the soundness of the simpler clause for $\depand$ ($\sepimp$) \cite{Cao2017}. Without the other Closed property, we must use a satisfaction clause which accounts for the order, as in BI. 

\subsection{Proof system} \label{subsec:proofsystem}

\begin{figure*}[!ht]
	\centering
\scalebox{.85}{
	\begin{tabular}{cccc}
		$\inferrule* [right = \textsc{Ax}]
		{~}
		{P \vdash P}$ &
		$\inferrule* [right = $\top$]
		{~}
		{P \vdash \top} $ &
		$\inferrule* [right = $\bot$]
		{~}
		{\bot \vdash P}$ &
		$\inferrule* [right=$\lor 1$]
		{P \vdash R \\ Q \vdash R}
		{P \lor Q \vdash R}$
		\quad
		$\inferrule* [right=$\lor 2$]
		{P \vdash Q_i}
		{P \vdash Q_1 \lor Q_2}$  \\

		$\inferrule* [right=$\land 1$]
		{P \vdash Q \\ P \vdash R}
		{P \vdash Q \land R}$ &
		$\inferrule* [right=$\land 2$]
		{Q \vdash R}
		{P \land Q \vdash R}$ &

		$\inferrule* [right=$\land 3/ \land4$]
		{P \vdash Q_1 \land Q_2}
		{P \vdash Q_i}$ &

		$\inferrule* [right=$\rightarrow$]
		{P \land Q \vdash R}
		{P \vdash Q \rightarrow R}$ \\

		$\inferrule* [right=\textsc{MP}]
		{P \vdash Q \rightarrow R \\ P \vdash Q}
		{P \vdash R}$ &

		$\inferrule* [right=$\sepimp$]
		{P \sepand Q \vdash R}
		{P \vdash Q \sepimp R}$ &

		$\inferrule* [right=$\sepimp$ \textsc{MP}]
		{ P \vdash Q \sepimp R \and S \vdash Q }
		{ P \sepand S \vdash R }$ &

		$\inferrule* [right=$\depimpr$]
		{ P \depand Q \vdash R }
		{ P \vdash Q \depimpr R }$ \\

		$\inferrule* [right=$\depimpr$ \textsc{MP}]
		{ P \vdash Q \depimpr R \and S \vdash Q }
		{ P \depand S \vdash R }$ &

		$\inferrule* [right=$\depimpl$]
		{ P \depand Q \vdash R }
		{ Q \vdash P \depimpl R }$ &

		$\inferrule* [right=$\depimpl$ \textsc{MP}]
		{ P \vdash Q \depimpl R \and S \vdash Q }
		{ S \depand P \vdash R }$ &

				$\inferrule* [right=$\sepand$-\textsc{Unit}]
				{~}
				{P \dashv \vdash P \sepand \sepid}$ \\

				$\inferrule* [right=$\sepand$-\textsc{Conj}]
				{P \vdash R \\ Q \vdash S}
				{P \sepand Q \vdash R \sepand S}$ &

				$\inferrule* [right=$\sepand$-\textsc{Comm}]
				{~}
				{P \sepand Q \vdash Q \sepand P}$ &
			$\inferrule* [right= $\sepand$-\textsc{Assoc}]
			{~}
	{(P \sepand Q) \sepand R \dashv\vdash P \sepand (Q \sepand R)}$
				&
				$\inferrule* [right=$\depand$-\textsc{Left Unit}]
				{~}
				{P \vdash \sepid \depand P }$ \\

				$\inferrule* [right=$\depand$-\textsc{Conj}]
				{ P \vdash R \\ Q \vdash S }
				{ P \depand Q \vdash R \depand S}$	&

				$\inferrule* [right=$\depand$-\textsc{Right Unit}]
				{~}
				{ P \dashv \vdash P \depand \sepid }$
				&
				$\inferrule* [right= $\depand$-\textsc{Assoc}]
		{~}
	{(P \depand Q) \depand R \dashv\vdash P \depand (Q \depand R)}$ &
	$\inferrule* [right= \textsc{RevEx}]
		{~}
	{(P \depand Q) \sepand (R \depand S) \vdash (P \sepand R) \depand (Q \sepand S)}$
\end{tabular}
}

\caption{Hilbert system for \LOGIC}
\label{fig:hilbert-dibi}
\end{figure*}

A Hilbert-style proof system for \LOGIC is given in \Cref{fig:hilbert-dibi}.
This calculus extends a system for BI with additional rules governing the new
connectives $\depand$, $\depimpr$ and $\depimpl$: in \Cref{sec:sc} we will prove
this calculus is sound and complete.  We briefly comment on two important
details in this proof system.

\paragraph*{Reverse exchange}
The proof system of \LOGIC shares many similarities with Concurrent Kleene
Bunched Logic (CKBI)~\cite{docherty2019bunched}, which also extends BI with a
non-commutative conjunction. Inspired by concurrent Kleene algebra
(CKA)~\cite{Hoare2011}, CKBI supports the following exchange axiom, derived from
CKA's exchange law:
\[ (P \sepand R) \depand (Q \sepand S) \vdash_{\text{CKBI}}  (P \depand Q) \sepand (R \depand S) \]
%
%
In models of CKBI, $\sepand$ describes interleaving concurrent composition,
while $\depand$ describes sequential composition. The exchange rule states that
the process on the left has \emph{fewer} behaviors than the process on the
right---e.g., $P \depand Q$ allows fewer behaviors than $P \sepand Q$, so $P
\depand Q \vdash_{\text{CKBI}} P \sepand Q$ is derivable.

In our models, $\sepand$ has a different reading: it states that two
computations can be combined because they are \emph{independent} (i.e.,
non-interfering). Accordingly, \LOGIC replaces \textsc{Exch} by the
\emph{reversed} version \textsc{RevEx}---the fact that the process on the left
is \emph{safe} to combine implies that the process on the right is also safe. $P
\sepand Q$ is now \emph{stronger} than $P \depand Q$, and $P \sepand Q \vdash P
\depand Q$ is derivable
\iffull
(\Cref{sep2dep}).
\else
(see the extended version~\cite{extended}).
\fi

\paragraph*{Left unit}
While $\depand$ has a \emph{right} unit in our logic, it does not have a proper
\emph{left} unit. Semantically, this corresponds to the lack of a frame
condition for $\odot$-Coherence\textsubscript{L} in our definition of \LOGIC
frames. This difference can also be seen in our proof rules: while
$\depand$-\textsc{Right Unit} gives entailment in both directions,
$\depand$-\textsc{Left Unit} only shows entailment in one direction---there is
no axiom stating $\sepid \depand P \vdash P$.

We make this relaxation to support our intended models, which we will see in
\Cref{sec:models}. In a nutshell, states in our models are Kleisli arrows that
\emph{preserve their input through to their output}---intuitively, in
conditional distributions, the variables that have we conditioned on will remain
fixed. Our models take $\odot$ to be Kleisli composition, which exhibits an
important asymmetry for such arrows: $f$ can always be recovered from $f \odot
g$, but not from $g \odot f$. As a result, the set of all arrows naturally
serves as the set of right units, but these arrows cannot all serve as left
units.



\subsection{Soundness and Completeness of \LOGIC} \label{sec:sc}
A methodology for proving the soundness and completeness of bunched logics is
given by \citet{docherty2019bunched}, inspired by the duality-theoretic approach
to modal logic \cite{Goldblatt1989}.  First, \LOGIC is proved sound and complete
with respect to an algebraic semantics obtained by interpreting the rules of the
proof system as algebraic axioms.  We then establish a representation theorem:
every \LOGIC algebra $\mathbb{A}$ embeds into a \LOGIC algebra generated by a
DIBI frame, that is in turn generated by $\mathbb{A}$. Soundness and
completeness of the algebraic semantics can then be transferred to the Kripke
semantics. Omitted details can be found in
\iffull
  \Cref{app:completeness}.
\else
  \cite{extended}.
\fi

\begin{definition}[\LOGIC Algebra]
	A \emph{\LOGIC algebra} is an algebra $\mathbb{A} = (A, \land, \lor, \rightarrow, \top, \bot, \sepand, \sepimp, \depand, \depimpr, \depimpl, \sepid)$ such that, for all $a, b, c, d \in A$:
	\begin{itemize}
		\item $(A, \land, \lor, \rightarrow, \top, \bot)$ is a Heyting algebra;
		\item $(A, \sepand, \sepid)$ is a commutative monoid;
		\item $(A, \depand, \sepid)$ is a \emph{weak monoid}: $\depand$ is an associative operation with right unit $\sepid$ and  $a \leq \sepid \depand a$;
		\item $a \sepand b \leq c$ iff $a \leq b \sepimp c$;
		\item $a \depand b \leq c$ iff $a \leq b \depimpr c$ iff $b \leq a \depimpl c$;
		\item $(a \depand b) \sepand (c \depand d) \leq (a \sepand c) \depand (b \sepand d)$.
	\end{itemize}
\end{definition}

An \emph{algebraic interpretation} of \LOGIC is specified by an assignment
$\llbracket - \rrbracket: \mathcal{AP} \rightarrow A$. The interpretation is
obtained as the unique homomorphic extension of this assignment, and so we use
the notation $\llbracket - \rrbracket$ interchangeably for both assignment and
interpretation. Soundness and completeness can be established by constructing a
term \LOGIC algebra by quotienting formulas by equiderivability.

\begin{restatable}{theorem}{algebraicsc}					      	    					\label{thm:algebraicsc}
$P \vdash Q$ is derivable iff
					  	$\llbracket P \rrbracket \leq \llbracket Q \rrbracket$
for all algebraic interpretations $\llbracket - \rrbracket$. 
\end{restatable}

We now connect these algebras to \LOGIC frames. A \emph{filter} on a bounded
distributive lattice $\mathbb{A}$ is a non-empty set $F \subseteq A$ such that,
for all $x, y \in A$, (1) $x \in F$ and $x \leq y$ implies $y \in F$; and (2)
$x, y \in F$ implies $x \land y \in F$. It is a \emph{proper} filter if it
additionally satisfies (3) $\bot \not\in F$, and a \emph{prime} filter if it
also satisfies (4) $x \lor y \in F$ implies $x \in F$ or $y \in F$. We denote
the set of prime filters of $\mathbb{A}$ by $\mathbb{PF}_{\mathbb{A}}$.

\begin{definition}[Prime Filter Frame]
	Given a \LOGIC algebra $\mathbb{A}$, the \emph{prime filter frame} of $\mathbb{A}$ is defined as $Pr(\mathbb{A}) = (\mathbb{PF}_{\mathbb{A}}, \subseteq, \oplus_{\mathbb{A}}, \odot_{\mathbb{A}}, E_{\mathbb{A}})$, where
$F \oplus_{\mathbb{A}} G := \{ H \in \mathbb{PF}_{\mathbb{A}}  \mid \forall a \in F, b \in G(a \sepand b \in H)\}$,
$F \odot_{\mathbb{A}} G := \{ H \in \mathbb{PF}_{\mathbb{A}} \mid \forall a \in F, b \in G(a \depand b \in H)\}$ and
$E_{\mathbb{A}} := \{ F \in \mathbb{PF}_{\mathbb{A}} \mid \sepid \in F \}$.
%
\end{definition}

\begin{restatable}{proposition}{primeisdibi}
For any \LOGIC algebra $\mathbb{A}$, the prime filter frame $Pr(\mathbb{A})$ is a \LOGIC frame. 
\end{restatable}

In the other direction, \LOGIC frames generate \LOGIC algebras.

\begin{definition}[Complex Algebra]
Given a \LOGIC frame $\mathcal{X} = (X, \sqsubseteq, \oplus, \odot, E)$, the
\emph{complex algebra} of $\mathcal{X}$ is defined to be
$Com(\mathcal{X}) = (\mathcal{P}_{\sqsubseteq}(X), \cap, \cup,
\Rightarrow_{\mathcal{X}}, X, \emptyset, \bullet_{\mathcal{X}},
\multimapdot_{\mathcal{X}}, \triangleright_{\mathcal{X}}, \rseq_{\mathcal{X}},
\lseq_{\mathcal{X}}, E)$:

\vspace{-2.5ex}
\begin{footnotesize}
	\[ \begin{array}{cl}
	\mathcal{P}_{\sqsubseteq}(X) &= \{ A \subseteq X \mid \text{ if } a \in A \text{ and } a \sqsubseteq b \text{ then } b \in A \} \\
	A \Rightarrow_{\mathcal{X}} B &= \{ a \mid \text{for all } b, \text{ if } b \sqsupseteq a \text{ and } b \in A \text{ then } b \in B \} \\
	A \bullet_{\mathcal{X}} B &= \{ x \mid \text{there exist } x', a, b \text{ s.t } x \sqsupseteq x' \in a \oplus b, a \in A \text{ and } b \in B \} \\
	A \multimapdot_{\mathcal{X}} B &= \{ x \mid \text{for all } a, b, \text{ if }  b \in x \oplus a \text{ and } a \in A \text{ then } b \in B \}  \\
	A \triangleright_{\mathcal{X}} B &= \{ x \mid \text{there exist }  a, b \text{ s.t } x \in a \odot b, a \in A \text{ and } b \in B \} \\
	A \rseq_{\mathcal{X}} B &= \{ x \mid \text{for all } x', a, b, \text{ if } x \sqsubseteq x', b \in x' \odot a \text{ and } a \in A \text{ then } b \in B \}  \\
	A \lseq_{\mathcal{X}} B &= \{ x \mid \text{for all } x', a, b, \text{ if } x \sqsubseteq x', b \in a \odot x' \text{ and } a \in A \text{ then } b \in B \}.
	\end{array} \]
\end{footnotesize}
\end{definition}

\begin{restatable}{proposition}{complexisdibi}
For any \LOGIC frame $\mathcal{X}$, the complex algebra $Com(\mathcal{X})$ is a \LOGIC algebra. 
\end{restatable}


The following main result facilitates transference of soundness and completeness.

\begin{restatable}[Representation of \LOGIC algebras]{theorem}{representationthm}
Every \LOGIC algebra is isomorphic to a subalgebra of a complex algebra: given a
\LOGIC algebra $\mathbb{A}$, the map $\theta_{\mathbb{A}}: \mathbb{A}
\rightarrow Com(Pr(\mathbb{A}))$ defined by $\theta_{\mathbb{A}}(a) = \{ F \in
\mathbb{PF}_{\mathbb{A}} \mid a \in F \}$ is an embedding. 
\end{restatable}

Given the previous correspondence between \LOGIC algebras and frames, we only
need to show that $\theta$ is a monomorphism: the necessary argument is
identical to that for similar bunched logics \cite[Theorems 6.11,
6.25]{docherty2019bunched}. Given $\llbracket - \rrbracket$ on $\mathbb{A}$, the
representation theorem establishes that $\mathcal{V}_{\llbracket -
\rrbracket}(p) := \theta_{\mathbb{A}}(\llbracket p \rrbracket)$ is a persistent
valuation on $Pr(\mathbb{A})$ such that $F \models_{\mathcal{V}_{\llbracket -
\rrbracket}} P$ iff $\llbracket P \rrbracket \in F$, from which our main theorem
can be proved.

\begin{restatable}[Soundness and Completeness]{theorem}{kripkesc}
$P \vdash Q$ is derivable iff $P \models Q$. 
\end{restatable}

\section{Models of \LOGIC} \label{sec:models}
In this section, we introduce two concrete models of \LOGIC to facilitate
logical reasoning about (in)dependence in probability distributions and
relational databases. In both models the operations $\odot$ and $\oplus$ will be
{\em deterministic} partial functions; we write $h = f \bullet g$ instead of $\{
h \} = f \bullet g$, for $\bullet \in \{ \odot, \oplus\}$. We start with some
preliminaries on memories and distributions.

\subsection{Memories, distributions, and Markov kernels}\label{sec:memories}

\subsubsection*{Operations on Memories}
Let $\Val$ be a fixed set of values (e.g., the Booleans), $S$ be a set of
variable names, and let $\R{S}$ denote the set of functions of type $m\colon S
\rightarrow \Val$.  We call such functions \emph{memories} because we can think
of $m$ as assigning a value to each variable in $S$; we will refer to $S$
as the \emph{domain} of $m$.  The only element in $\R{\emptyset}$ is the empty
memory, which we write as $\empmem$.

We need two operations on memories. First, a memory $m$ with domain $S$
can be projected to a memory $m^T$ with domain $T$ if $T \subseteq S$, defined
as $m^T (x) = m(x)$ for all $x \in T$. Second, two memories can be combined if
they agree on the intersection of their domains:
%
  given memories $m_1 \in \R{S}$, $m_2 \in \R{T}$ such that $m_1^{S \cap T} =
  m_2^{S \cap T}$, we define $m_1 \otimes m_2\colon S \cup T \rightarrow \Val$ by
  \begin{equation}\label{def:singleton_bowtie_prob}\small
    m_1 \otimes m_2 (x) \ \defeq \begin{cases}
      m_1(x) &\text{if } x \in S \\ 
      m_2(x) &\text{if } x \in T  
    \end{cases}
  \end{equation}

\subsubsection*{Probability distributions and Markov kernels}
We use the distribution monad to model distributions over memories. Given a set
$X$, let $\mathcal{D}(X)$ denote the set of finite distributions over $X$, i.e.,
the set containing all finite support functions $\mu\colon X \rightarrow [0,1]$
satisfying $\sum_{x \in X} \mu(x) = 1$. This operation on sets can be lifted to
functions $f\colon X\to Y$, resulting in a map of distributions $\DD(f)\colon
\DD(X) \to \DD(Y)$ given by $ \DD(f)(\mu)(y) \defeq \sum_{f(x) =y} \mu(x)$
(intuitively, $\DD(f)$ takes the sum of the probabilities of all elements in the
pre-image of $y$). These operations turn $\DD$ into a functor on sets and,
further, $\DD$ is also a \emph{monad}~\cite{giry1982categorical, Moggi91}.
\begin{definition}[Distribution Monad]
  Define $\unit \colon X \to \DD(X)$ as $\unit_X(x)\defeq\delta_x$ where
  $\delta_x$ denotes the \emph{Dirac distribution} on $x$: for any $y \in X$, we
  have $\delta_x(y) = 1$ if $y = x$, otherwise $\delta_x(y) = 0$. Further,
  define $\bind\colon \DD (X) \rightarrow (X \to \DD (Y)) \to \DD(Y)$ by
  $\bind(\mu)(f)(y)\defeq \sum_{p \in \DD(Y)} \DD(f)(\mu)(p) \cdot p(y)$.
\end{definition}
Intuitively, $\unit$ embeds a set into distributions over the set, and $\bind$
enables the sequential combination of probabilistic computations. Both maps are
natural transformations and satisfy the following interaction laws, establishing
that $\langle \DD, \unit, \bind \rangle$ is a monad:
\begin{equation}\label{eq:monad}
  \begin{gathered}
    \bind(\unit(x))(f) = f(x),
    \qquad\bind(\mu)(\unit) = \mu,  \\
    \bind(\bind(\mu)(f))(g) = \bind(\mu)(\lambda x. \bind(f(x))(g)).
  \end{gathered}
\end{equation}
The distribution monad has an equivalent presentation in which $\bind$ is
replaced with a multiplication operation $\DD\DD(X) \to \DD(X)$, which flattens
distributions by averaging.

The monad $\DD$ gives rise to the {\em Kleisli category} of $\DD$, denoted
$\Kl(\DD)$, with sets as objects and arrows of the form $f\colon X\to \DD(Y)$,
also known as \emph{Markov kernels}~\cite{Panangaden2009}. Arrow composition in
$\Kl(\DD)$ is defined using $\bind$: given $f\colon X\to \DD(Y)$, $g\colon Y\to
\DD(Z)$, the composition $f\odot g \colon X\to \DD(Z)$ is:
\begin{equation}\label{eq:odot}
(f\odot g) (x) \defeq \bind (f(x))(g)
\end{equation}
Markov kernels generalize distributions: we can lift a distribution $\mu\colon
\DD(X)$ to the kernel $f_\mu \colon 1 \to \DD(X)$ assigning $\mu$ to the single
element of $1$. Kernels can also encode conditional distributions, which play a
key role in conditional independence.

\definecolor{orange}{rgb}{1,0.7,0}
\begin{figure*}[!t]
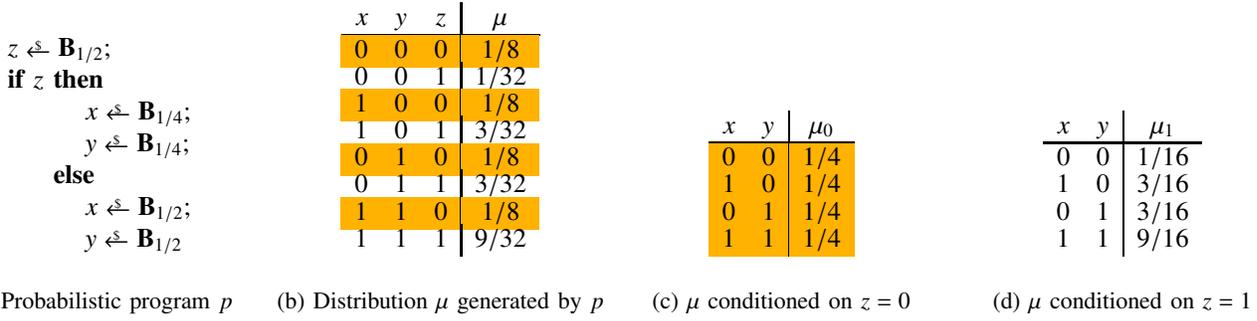

  \begin{subfigure}[b]{0.25\textwidth}
    \begin{center}
\[
\begin{array}{l@{\!\!\!\!\!\!}l}
 \Rand z {\Bern_{1/2}}; \\
\textbf{if}\ z \textbf{ then} \\
&\Rand x {\Bern_{1/4}}; \\
&\Rand y {\Bern_{1/4}}; \\
\phantom{\textbf{if}\ z }\textbf{ else} \\
&\Rand x {\Bern_{1/2}}; \\
&\Rand y {\Bern_{1/2}}
\end{array}\]
\end{center}
\caption{Probabilistic program $p$}
\label{kernels2coins:prog}
\end{subfigure}%
\begin{subfigure}[b]{0.25\textwidth}
  \begin{center}
  \[
\begin{array}{ccc|c}
x&y&z &\mu\\
\hline
\rowcolor{orange}
0&0&0& 1/8\\[-.4ex]
0&0&1& 1/32\\[-.4ex]
\rowcolor{orange}
1&0&0& 1/8\\[-.4ex]
1&0&1& 3/32\\[-.4ex]
\rowcolor{orange}
0&1&0& 1/8\\[-.4ex]
0&1&1& 3/32\\[-.4ex]
\rowcolor{orange}
1&1&0& 1/8\\[-.4ex]
1&1&1& 9/32
\end{array}
\]
\end{center}
\caption{Distribution $\mu$ generated by $p$}
\label{kernels2coins:mu}
\end{subfigure}%
\begin{subfigure}[b]{0.25\textwidth}
  \begin{center}
  \[
\begin{array}{cc|c}
x&y&\mu_0\\
\hline
\rowcolor{orange}
0&0& 1/4\\[-.4ex]
\rowcolor{orange}
1&0& 1/4\\[-.4ex]
\rowcolor{orange}
0&1& 1/4\\[-.4ex]
\rowcolor{orange}
1&1& 1/4
\end{array}
\]
\end{center}
\caption{$\mu$ conditioned on $z = 0$}
\label{kernels2coins:cond0}
\end{subfigure}%
\begin{subfigure}[b]{0.25\textwidth}
  \begin{center}
  \[
\begin{array}{cc|c}
x&y &\mu_1\\
\hline
0&0& 1/16\\[-.4ex]
1&0& 3/16\\[-.4ex]
0&1& 3/16\\[-.4ex]
1&1& 9/16
\end{array}
\]
\end{center}
\caption{$\mu$ conditioned on $z = 1$}
\label{kernels2coins:cond1}
\end{subfigure}
	\caption{From probabilistic programs to kernels}
	\label{kernels2coins}
\end{figure*}

\begin{example}\label{example:ci}
Consider the program $p$ in \Cref{kernels2coins:prog}, where $x,y$, and $z$ are
Boolean variables. First, flip a fair coin and store the result in $z$.  If $z =
0$, flip a fair coin twice, and store the results in $x$ and $y$, respectively.
If $z = 1$, flip a coin with bias $1/4$ twice, and store the results in $x$ and
$y$. This program produces a distribution $\mu$, shown in
\Cref{kernels2coins:mu}.

If we condition $\mu$ on $z = 0$, then the resulting distribution $\mu_0$ models
two independent fair coin flips: $1/4$ probability for each possible pair of
outcomes (\Cref{kernels2coins:cond0}). If we condition on $z = 1$, however, then
the distribution $\mu_1$ will be skewed---there will be a much higher
probability that we observe $(1,1)$ than $(0, 0)$, but $x$ and $y$ are still
independent (\Cref{kernels2coins:cond1}).

To connect $\mu_0$ and $\mu_1$ to the original distribution $\mu$, we package
$\mu_0$ and $\mu_1$ into a Markov kernel $k \colon \R{z} \to \DR{\{x,y,z\}}$
given by $k(i)(d)=\mu_i (d^{\{x,y\}})$.
Then, the relation between the conditional and original distributions is $f_\mu
= f_{\mu_z}\odot k$, where $\mu_z$ is the projection of $\mu$ on $\{z\}$.
\end{example}

Finite distributions of memories over $U$, denoted $\DR{U}$, will play a central
role in our models. We will refer to maps $f\colon \R{S} \rightarrow \DR{U}$ as
\emph{(Markov) kernels}, and define $\dom(f) = S$ and $\range(f) = U$.

We can marginalize/project kernels to a smaller range.
\begin{definition}[Marginalizing kernels]
  \label{def:marginalize_prob}
  For a Markov kernel $f \colon \R{S} \rightarrow \DR{U}$ and $V \subseteq U$,
  the {\em marginalization of $f$ by V} is the map $\pi_V f \colon \R{S}
  \rightarrow \DR{V}$: \((\pi_V f) (d)(r) \defeq \sum_{m \in \R{U \setminus V}}
  f(d)(r \otimes m)\) for $d \in \R{S}, r \in \R{V}$; undefined terms do not
  contribute to the sum.
\end{definition}

We say a kernel $f \colon \R{S} \rightarrow \DR{U}$ \emph{preserves its input
to its output} if $S \subseteq U$ and $\pi_S f = \unit_{\R{S}}$.  Intuitively,
such kernels are suitable for encoding conditional distributions: once a
variable has been conditioned on, its value should not change. We can compose
these kernels in two ways.

\begin{definition}[Composing Markov kernels on memories]
  Given $f\colon \R{S} \rightarrow \DR{T}$ and $g\colon \R{U} \rightarrow
  \DR{V}$ that preserve their inputs, we define their \emph{parallel
  composition}, whenever $S\cap U=T\cap V$, as the map $ f\oplus g \colon
  \R{S\cup U} \to \DR{T\cup V}$ given by \[(f \oplus g) (d) (m) \defeq
  f(d^S)(m^T)  \cdot g(d^U) (m^V). \] If $T=U$, the \emph{sequential
  composition} $f \odot g \colon \R{S} \to \DR{V}$ is just Kleisli composition
  (\cref{eq:odot}).
\end{definition}

\subsection{A concrete probabilistic model of \LOGIC}\label{sec:probmodel}

We now have all the ingredients to define a first concrete model: states are
Markov kernels that preserve their input; $\oplus$ (resp. $\odot$) will be
parallel (resp. sequential) composition. The use of $\oplus$ to model
independence generalizes the approach in \citet{barthe2019probabilistic}.
Combining both compositions---sequential and parallel---enables capturing
conditional independence.

\begin{definition}[Probabilistic frame]
  \label{def:probmodel}
  We define the frame $(\MD, \sqsubseteq, \oplus, \odot, \MD)$ as follows:
  \begin{itemize}
    \item Let $\MD$ consist of Markov kernels that preserve their input to their
      output;
    \item  $\oplus$, $\odot$ are parallel and sequential composition of
      kernels;
    \item Given $f, g \in \MD$, $f \sqsubseteq g$ if there exist $R \subseteq
      \Val$, $h \in \MD$ such that $g = (f \oplus \unit_{\R{R}}) \odot h$.
  \end{itemize}
\end{definition}

We make two remarks. First, $f\sqsubseteq g$ holds when $g$ can be obtained from
extending $f$: compose $f$ in parallel with $\unit_{\R{R}}$, then extend the
range via composition with $h$. We can recover $f$ from $g$ by marginalizing $g$
to $\range(f) \cup R$, then ignoring the $R$ portion. Second, the definition of
$f \odot g$ on $\MD$ can be simplified. Given $f\colon \R{S} \to \DR{T}$ and
$g\colon \R{T} \to \DR{V}$, \cref{eq:odot} yields the formula:
\[
  (f \odot g)(d)(m) \defeq \sum_{m' \in \R{T}} f(d)(m') \cdot g(m')(m) .
\]
Since $f,g \in \MD$ preserve input to output, this reduces to
\begin{equation}\label{eq:simplemu}
  (f \odot g)(d)(m) = f(d)(m^{T}) \cdot g(m^T)(m^V) .
\end{equation}
We show that our probabilistic frame is indeed a \LOGIC frame.
\begin{restatable}{theorem}{MdisLOGIC} \label{MdisLOGIC}
  $(\MD, \sqsubseteq, \oplus, \odot, \MD)$ is a \LOGIC frame.
\end{restatable}
\begin{proof}[Proof sketch]
  First, we show that $\MD$ is closed under $\oplus$ and $\odot$, and
  $\sqsubseteq$ is transitive and reflexive. The frame axioms are mostly
  straightforward, but some conditions rely on a property of our model we call
  \emph{Exchange Equality}: if both $(f_1 \oplus f_2) \odot (f_3 \oplus f_4)$
  and $(f_1 \odot f_3) \oplus (f_2 \odot f_4)$ are defined, then they are equal,
		and if the second is defined, then so is the first.
  For example:
  \begin{description}
    \item [($\oplus$ Unit Coherence):]
      The unit set in this frame is the entire state space $\MD$: we must show
      that for any $f_1, f_2 \in \MD$, if $f_1 \oplus f_2$ is defined, then $f_1
      \sqsubseteq f_1 \oplus f_2$:
      \begin{align*}
        f_1 \oplus f_2 &= (f_1 \odot \unit_{\range(f_1)}) \oplus (\unit_{\dom(f_2)} \odot f_2) \\
                       &= (f_1 \oplus \unit_{\dom(f_2)}) \odot (\unit_{\range(f_1)} \oplus f_2) \tag{\scriptsize Exch. Eq.}\\
                       &= (f_1 \oplus \unit_{\dom(f_2)}) \odot (f_2 \oplus \unit_{\range(f_1)}) \tag{\scriptsize $\oplus$ Comm.}
      \end{align*}

      %
  \end{description}
  We present the complete proof in
  \iffull
    \Cref{app:probabilistic}.
  \else
    \cite{extended}.
  \fi
\end{proof}

\begin{example}[Kernel decomposition]
  Recall the distribution $\mu$ on $\R{\{x, y ,z \}} $ from \Cref{example:ci}. Let
  $k_x \colon \R{z} \to \DR{\{ x, z\}}$ encode the conditional distribution of $x$
  given $z$, and let $k_y \colon \R{z} \to \DR{ \{y, z\} }$ encode the conditional
  distribution of $y$ given $z$. Explicitly, for $v = x \text{ or } y$,
  \begin{align*}
    k_v(z = 0)(v = 1, z = 0) &= 1/2 & k_v(z = 0)(v = 0, z = 0) &= 1/2 \\
    k_v(z = 1)(v = 1, z = 1) &= 1/4 & k_v(z = 1)(v = 0, z = 1) &= 3/4.
  \end{align*}
  Since $k_x,k_y$ include $z$ in their range, $k_x \oplus k_y$ is defined.  A
  small calculation shows that $k_x \oplus k_y = k$, where $k\colon \R{z} \to
  \DR{ \{x, y, z\} }$ is the conditional distribution of $(x, y, z)$ given $z$.
  This decomposition shows that $x$ and $y$ are independent conditioned on $z$
  (we shall formally prove this later in~\cref{section:condindep}).
\end{example}

\subsection{Relations, join dependency, and powerset kernels}
We developed the probabilistic model in the previous section using operations
from the distribution monad $\DD$. Instantiating our definitions with operations
from other monads gives rise to other interesting models of \LOGIC. In this
section, we develop a \emph{relational} model based on the powerset monad $\PP$,
and show how our logic can be used to reason about join dependency properties of
tables from database theory. Before we present our relational model, we
introduce some notations and basic definitions on relations.

Tables are often viewed as \emph{relations}---sets of tuples where each
component of the tuple corresponds to an \emph{attribute}. Formally, a relation
$R$ over a set of attributes $S$ is a set of \emph{tuples} indexed by $S$. Each
tuple maps an attribute in $S$ to a value in $\Val$, and hence can be seen as a
memory in $\R{S}$, as defined in \Cref{sec:memories}.  The projection and
$\otimes$ operations on $\R{S}$ from \Cref{def:singleton_bowtie_prob} can be
lifted to relations.
\begin{definition}[Projection and Join]
  The \emph{projection} of a relation $R$ over attributes $X$ to $Y \subseteq X$
  is given by $R^{Y} \defeq \{ r^Y \mid r \in R\} $.  The \emph{natural join} of
  relations $R_1$ and $R_2$ over attributes $X_1$ and $X_2$, respectively, is
  the relation $R_1 \bowtie R_2 \defeq \{ m_1 \otimes m_2 \mid m_1 \in R_1
  \text{ and } m_2 \in R_2 \}$ over attributes $X_1 \cup X_2$.
\end{definition}

Since tables can often be very large, finding compact representations for them
is useful. These representations can leverage additional structure common in
real-world databases; for instance, the value of one attribute might determine
the value of another, a so-called \emph{functional dependency}. Other dependency
structures can enable a large relation to be factored as a combination of
smaller ones. A classical example is on \emph{join dependency}, a relational
analogue of conditional independence.

\begin{definition}[Join dependency \citep{DBLP:journals/tods/Fagin77,abiteboul1995foundations}]
  A relation $R$ over attribute set $X_1 \cup X_2$ satisfies the \emph{join
  dependency} $X_1 \bowtie X_2$ if $R = (R^{X_1}) \bowtie (R^{X_2})$.
\end{definition}

\begin{example}[Decomposition]\label{example:join}
  Consider the relation $R$ in \Cref{ex:joindep}, with three attributes:
  \Researcher, \Field, and \Conference. $R$ contains triple $(a,b,c)$ if and
  only if researcher $a$ works in field $b$ and attends conference $c$. If we
  know that researchers in the same field all have a shared set of conferences
  they attend, then we can recover $R$ by joining two relations: one associating
  researchers to their fields, and another associating fields to conferences.
  As shown below, $R$ satisfies the join dependency $\{\Researcher, \Field\}
  \bowtie \{\Conference, \Field\}$. While the factored form is only a bit
  smaller (12 entries instead of 15), savings can be significant for larger
  relations.

  \begin{figure*}[!h]
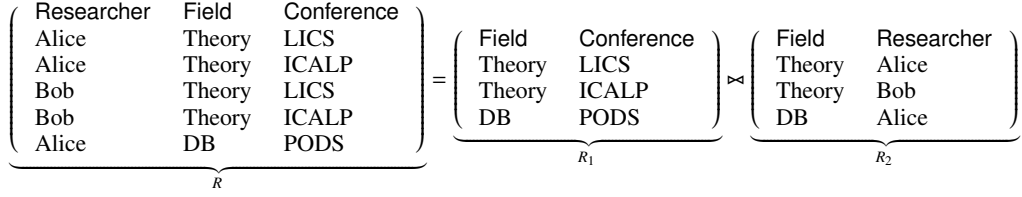
\small
    \begin{align*}\label{ex:joindep}
      \underbrace{%
        \left(
          \begin{tabular}{lll}
            \Researcher & \Field & \Conference \\
            Alice &  Theory & LICS \\
            Alice &  Theory & ICALP \\
            Bob &  Theory & LICS \\
            Bob &  Theory & ICALP \\
            Alice &  DB & PODS
          \end{tabular}
      \right)}_{R}
      =
      \underbrace{%
        \left(
          \begin{tabular}{ll}
            \Field & \Conference \\
            Theory & LICS \\
            Theory & ICALP \\
            DB & PODS
          \end{tabular}
      \right)}_{R_1}
      \bowtie
      \underbrace{%
        \left(
          \begin{tabular}{ll}
            \Field & \Researcher \\
            Theory & Alice \\
            Theory & Bob \\
            DB & Alice
          \end{tabular}
      \right)}_{R_2}
    \end{align*}
    \caption{Factoring a relation}
    \label{ex:joindep}
  \end{figure*}
\end{example}

\subsubsection*{Powerset monad and kernels}
Much like how we decomposed distributions as Markov kernels---Kleisli arrows for
the distribution monad---we will decompose relations using Kleisli arrows for
the powerset monad, $\Kl(\PP)$.
\begin{definition}[Powerset monad]
  Let $\PP$ be the endofunctor $\mathbf{Set} \rightarrow \mathbf{Set}$ mapping
  every set to the set of its subsets $\PP(X) = \{U \mid U\subseteq X\}$. We
  define $\unit_X\colon X \rightarrow \PP(X)$ mapping each $x \in X$ to the
  singleton $\{x\}$, and $\bind \colon \PP(X) \to (X \to \PP(Y)) \to \PP(Y)$
  by $\bind (U)(f) \defeq \cup \{ y \mid \exists x\in U.  f(x) = y\}$.
\end{definition}
The triple $\left<\PP, \unit, \bind\right>$ forms a monad, and obeys the laws in
\Cref{eq:monad}. We overload the use of $\unit$ and $\bind$ as it will be clear
from the context which monad, powerset or distribution, we are considering. The
Kleisli category $\Kl(\PP)$ is defined analogously as for $\DD$, with sets as
objects and arrows $X\to \PP(Y)$, and composition given as in \Cref{eq:odot}.

Like before, we consider maps $\R{S} \rightarrow \PR{T}$, which we call {\em
powerset kernels} in analogy to Markov kernels, or simply kernels when the monad
is clear from the context. Powerset kernels can also be projected to a smaller
range.
\begin{definition}[Marginalization]
  \label{def:Mp_marginalization}
  Suppose that $T \subseteq U$.  A map $f$ of type $\R{S} \rightarrow \PR{U}$
  can be marginalized to $\pi_T f \colon \R{S} \to \PR{T}$ by defining:
  $(\pi_T f) (s) \defeq f(s)^{T}$
\end{definition}

We need two composition operations on powerset kernels. We say that powerset
kernel $f\colon  \R{S} \rightarrow \PR{S \cup T}$ \emph{preserves input to
output} if $\pi_S f = \unit_{\R{S}}$.

\begin{definition}[Composition of powerset kernels]
  Given kernels $f\colon \R{S} \rightarrow \PR{T}$ and $g \colon \R{U}
  \rightarrow \PR{V}$ that preserve input to output, we define their {\em
  parallel composition} whenever $T\cap V= S\cap U$ as the map $f \oplus g
  \colon \R{S\cup U}\to \PR{T \cup V}$ given by $(f \oplus g) (d) \defeq f(d^S)
  \bowtie g(d^U)$.  Whenever $T=U$ we define the {\em sequential composition} $f
  \odot g\colon \R{S} \to \PR{V}$ using Kleisli composition. Explicitly: $(f
  \odot g)(s)  =  \{ v \mid u \in f(s) \text{ and } v \in g(u)\}$.
\end{definition}
%

\subsection{A concrete relational model of \LOGIC}\label{sec:relmodel}
We can now define the second concrete model of \LOGIC: states will be powerset
kernels, and we will use the parallel and sequential composition in a
construction similar to $\MD$.

\begin{definition} [Relational frame]
  We define the frame $(\MP, \sqsubseteq,
  \oplus, \odot, \MP)$ as follows:
  \begin{itemize}[leftmargin=*]
    \item $\MP$ consists of powerset kernels preserving input to output;
    \item $\oplus$, $\odot$ are parallel and sequential composition of powerset kernels;
    \item Given $f, g \in \MP$, $f \sqsubseteq g$ if there exist $R \subseteq \Val$, $h \in \MP$
      such that $g = (f \oplus \unit_{\R{R}}) \odot h$.
  \end{itemize}
\end{definition}
Like in $\MD$, $f \sqsubseteq g$ iff $g$ can be obtained from $f$ by adding
attributes that are preserved from domain to range, and then mapping tuples in
the range to relations over a larger set of attributes. We can recover $f$ from
$g$ by marginalizing to $\range(f) \cup R$, and then ignoring the attributes in
$R$.

$\MP$ is also a \LOGIC frame.
\begin{restatable}{theorem}{thmrelframe}
  \label{thm:relframe}
  $(\MP, \sqsubseteq, \oplus, \odot, \MP)$ is a \LOGIC frame.
\end{restatable}
\begin{proof}[Proof sketch.]
  The proof follows \cref{MdisLOGIC} quite closely, since $\MP$ also satisfies
  Exchange equality. We present the full proof in
  \iffull
    \cref{app:relational}.
  \else
    \cite{extended}.
  \fi
\end{proof}

\section{Application: Modeling Conditional and Join Dependencies} \label{sec:CI}
In our concrete models, distributions and relations can be factored into simpler
parts. Here, we show how \LOGIC formulas capture conditional independence and
join dependency.

\subsection{Conditional independence}
\label{section:condindep}

Conditional independence (CI) is a well-studied notion in probability theory
and statistics~\citep{dawid1979}. While there are many interpretations of
CI, a natural reading is in terms of \emph{irrelevance}: $X$ and $Y$ are
independent conditioned on $Z$ if knowing the value of $Z$ renders $X$
irrelevant to $Y$---observing one gives no further information about the other.

Before defining CI, we introduce some notations. Let $\mu \in
\DR{\Var}$ be a distribution. For any subset $S \subseteq \Var$ and
assignment $s \in \R{S}$, we write:
\[
  \mu(S = s) \defeq \sum_{m \in \R{\Var}} \mu(s \otimes m) .
\]
Terms with undefined $s \otimes m$ contribute zero to the sum. We can now define
conditional probabilities:
\[
  \mu(S = s \mid S' = s') \defeq \frac{\mu(S = s, S' = s')}{\mu(S' = s')} ,
\]
where $\mu(S = s, S' = s') \defeq \mu( S \cup S' = s \otimes s')$.
Intuitively, this ratio is the probability of $S = s$ given $S' = s'$, and
it is only defined when the denominator is non-zero and $s, s'$ are
consistent (i.e., $s \otimes s'$ is defined). CI can be defined as follows.

\begin{definition}[Conditional independence]
  Let $X, Y, Z \subseteq \Var$. $X$ and $Y$ are \emph{independent
  conditioned on $Z$}, written $\pearlindep{X}{Y}{Z}$, if for all $x \in
  \R{X}$, $y \in \R{Y}$, and $z \in \R{Z}$:
  \[
    \mu(X = x \mid Z = z) \cdot \mu(Y = y \mid Z = z) = \mu(X = x, Y = y \mid Z = z) .
  \]
  When $Z = \emptyset$, we say $X$ and $Y$ are \emph{independent}, written $X \ci Y$.
\end{definition}

\begin{example}
  We give two simple examples of CI.
  \paragraph*{Chocolate and Nobel laureates}
  Researchers found a strong positive correlation between a nation's per capita
  Nobel laureates number and chocolate consumption.  But the correlation may be
  due to other factors, e.g., a nation's economic status. A simple check is to
  see if the two are conditionally independent fixing the third factor.

  \paragraph*{Algorithmic fairness}
  To prevent algorithms from discriminating based on sensitive features (e.g.,
  race and gender), researchers formalized notions of fairness using conditional
  independence~\cite{barocas-hardt-narayanan}. For instance, let $A$ be the
  sensitive features, $Y$ be the target label, and $\wh{Y}$ be the algorithm's
  prediction for $Y$. Considering the joint distribution of $(A, Y, \wh{Y})$, an
  algorithm satisfies \emph{equalized odds} if $\pearlindep{\wh{Y}}{A}{Y}$;
  \emph{calibration} if $\pearlindep{Y}{A}{\wh{Y}}$.
\end{example}

We will define a \LOGIC formula $P$ such that a distribution $\mu$ satisfies
$\pearlindep{X}{Y}{Z}$ if and only if its lifted kernel $f_\mu \defeq \empmem \mapsto f$
satisfies $P$. For this, we will need a basic atomic
proposition which describes the domain and range of kernels.

\begin{definition}[Basic atomic proposition]
  For sets of variables $A, B \subseteq \Var$, a basic atomic proposition
  has the form $\pair{A}{B}$. We give the following semantics to these
  formulas:
  \begin{align*}
    f \models \pair{A}{B} \text{ iff } &\text{there exists } f' \sqsubseteq f \\
                                       &\text{such that } \dom(f') = A \text{ and } \range(f') \supseteq B .
  \end{align*}
\end{definition}

For example, $f \colon \R{y} \to \DR{y, z}$ defined by $f(y \mapsto v) \defeq
\dunit(y \mapsto v, z \mapsto v)$ satisfies $\pair{y}{y}$, $\pair{y}{z}$,
$\pair{y}{\emptyset}$, $\pair{y}{y,z}$, $\pair{\emptyset}{\emptyset}$, and no
other atomic propositions.

\begin{restatable}{theorem}{theoprob}
  \label{theo:prob}
  Given distribution $\mu \in \DR{\Var}$, then for any $X, Y, Z \subseteq \Var$,
  \begin{equation}\label{eq:ci_form}
    f_\mu \models \pair{\emptyset}{Z} \depand \pair{Z}{X} \sepand \pair{Z}{Y}
  \end{equation}
  if and only if $\pearlindep{X}{Y}{Z}$ and $X \cap Y \subseteq Z$ are both satisfied.
\end{restatable}
The restriction $X \cap Y \subseteq Z$ is harmless: when $\pearlindep{X}{Y}{Z}$
but $X \cap Y \not\subseteq Z$, $X\cap Y$ must be deterministic given $Z$
\iffull
  (see~\cref{concrete:XcapYinS}),
\else
  (see~\cite{extended}),
\fi
and it suffices to check $\pearlindep{X}{Y}{Z \cup (X \cap Y)}$.  For
simplicity, we abbreviate the formula $\pair{\emptyset}{Z} \depand (\pair{Z}{X}
\sepand \pair{Z}{Y})$ as $\indep{Z}{X}{Y}$.

\begin{proof}[Proof sketch]
  For the forward direction, suppose $f_\mu$ satisfies \ref{eq:ci_form}.
  Then by
  \iffull
    \cref{findexact},
  \else
    \citep[Lemma A.38]{extended},
  \fi
  there exist $f$, $g$, and $h$ in $\MD$ with $f \odot (g \oplus h) \sqsubseteq
  f_\mu$, where $f\colon \R{\emptyset} \rightarrow \DR{Z}$, $g \colon \R{Z}
  \rightarrow \DR{Z \cup X}$, and $h \colon \R{Z} \rightarrow \DR{Z \cup Y}$; we
  also have $X \cap Y \subseteq Z$ as $f \odot (g \oplus h)$ is defined. Since
  $\dom(f_\mu) = \R{\emptyset}$, $f \odot (g \oplus h) \sqsubseteq f_\mu$
  implies:
  \begin{align*}
    f \odot (g \oplus h) & = \pi_{Z \cup X \cup Y} f_\mu
    \quad \text{and} \quad
    \qquad f = \pi_{Z} f_\mu.
  \end{align*}
  Further, we can show that  $f \odot (g \oplus h) = f \odot g
  \odot (\unit_{X} \oplus h) =  f \odot h \odot (\unit_{Y} \oplus
  g)$, and thus:
  \begin{align*}
    f \odot g  & = \pi_{Z \cup X} f_\mu 
    \quad \text{and} \quad
    f \odot h = \pi_{Z \cup Y} f_\mu. 
  \end{align*}
  These imply that $g$ ($h$ resp.) encodes the conditional distributions of $X$
  ($Y$ resp.) given $Z$, and $g \oplus h$ encodes the conditional distribution
  of $(X, Y)$ given $Z$. Hence, the conditional distribution of $(X,Y)$ given
  $Z$ is equal to the product distribution of $X$ given $Z$ and $Y$ given $Z$,
  and so $\pearlindep{X}{Y}{Z}$ holds in $\mu$.

  For the reverse direction, suppose that (a) $\pearlindep{X}{Y}{Z}$ holds in
  $\mu$ and (b) $X\cap Y\subseteq Z$.  Now, consider  $\pi_{X\cup Y \cup Z}
  f_\mu$, the marginal distribution on $(X, Y, Z)$ encoded as a kernel, and
  observe that $\pi_{X,Y,Z} f_\mu = f \odot f'$, where $f$ encodes the marginal
  distribution of $Z$, and $f'$ is the conditional distribution of $(X,Y)$ given
  values of $Z$. From (a), the conditional distribution of $(X,Y)$ given $Z$ is
  the product of the conditional distributions of $X$ given $Z$, and $Y$ given
  $Z$, that is $f'= g \oplus h$, where $g$ (resp. $h$) encode the conditional
  distribution of $X$ (resp. $Y$) given $Z$.  Then by (b), $f \odot (g \oplus
  h)$ is defined and $f \odot (g \oplus h) = \pi_{X\cup Y \cup Z} f_\mu
  \sqsubseteq f_\mu$.  It is straightforward to see that $f \odot (g \oplus h)$
  satisfies $\indep{Z}{X}{Y}$. Hence,  persistence shows that $f_\mu$ also
  satisfies $\indep{Z}{X}{Y}$.

  See
  \iffull
    \cref{Md:probhelper}
  \else
    \citep[Theorem A.11]{extended}
  \fi
  for details.
\end{proof}

\subsection{Join dependency}\label{sec:relational}

Recall that a relation $R$ over attributes $X \cup Y$ satisfies the \emph{Join
Dependency} (JD) $X \bowtie Y$ if $R = R^X \bowtie R^Y$. As we illustrated
through the \Researcher-\Field-\Conference\ example in \cref{sec:models}, join
dependencies can enable a relation to be represented more compactly.  By
interpreting the atomic propositions in the relational model, JD is captured by
the same formula we used for CI.

\begin{restatable}{theorem}{theorel}
  \label{theo:rel}
  Let $R \in \PR{\Var}$ and $X, Y$ be sets of attributes such that $X \cup Y =
  \Var$. The lifted relation $f_R = \empmem \mapsto R$ satisfies $f_R \models
  \indep{X \cap Y}{X}{Y}$ iff $R$ satisfies the join dependency $X \bowtie Y$.
\end{restatable}

JD is a special case of Embedded Multivalued Dependency (EMVD), where the
relation $R$ may have more attributes than $X \cup Y$. It is straightforward
to encode EMVD in our logic, but for simplicity we stick with JD.

\begin{proof}[Proof sketch]
  For the forward direction, by
  \iffull
    \cref{findexact},
  \else
  \citep[Lemma A.38]{extended},
  \fi
  there exist $f$, $g$, and $h \in \MP$ such that $f\colon \R{\emptyset}
  \rightarrow \PR{X \cap Y}$, $g \colon \R{X \cap Y} \rightarrow \PR{X}$, $h
  \colon \R{X \cap Y} \rightarrow \PR{Y}$, and $f \odot (g \oplus h) \sqsubseteq
  f_R$. Since by assumption $X \cup Y = \Var$, we must have $f \odot (g \oplus
  h) = f_R$.

  Unfolding $\oplus$ and $\odot$ and using the fact that $\range(f) =
  \dom(g) = \dom(h)$, we can show:
  \begin{align*}
    f \odot (g \oplus h) (\empmem)
        &= \{ u \bowtie (v_1 \bowtie v_2) \mid u \in f(\empmem), v_1 \in g(u ), v_2 \in h(u ) \} .
  \end{align*}
  Since $\bowtie$ is commutative, associative and idempotent, we have:
  \begin{align*}
    f \odot (g \oplus h) (\empmem)
        &= \{ (u \bowtie v_1) \bowtie ( u \bowtie v_2) \mid u \in f(\empmem), v_1 \in g(u), v_2 \in h(u) \} \\
        &= f \odot g (\empmem) \bowtie f \odot h (\empmem).
  \end{align*}
  We can also convert the parallel composition of $g, h$ into sequential
  composition by padding  to make the respective domain and range match:
  \( f \odot (g \oplus h) =  f \odot g \odot (\unit_{X} \oplus h)
  =  f \odot h \odot (\unit_{Y} \oplus g) \).
  Hence $f \odot g = \pi_X f_R $ and $f \odot h = \pi_Y f_R$,
  which implies $f \odot g (\empmem) = R^{X}$ and $f \odot h (\empmem) =
  R^{Y}$. Thus:
  \[
    R = f \odot (g \oplus h)(\empmem)
    = f \odot g (\empmem) \bowtie f \odot h (\empmem)
    = R^{X} \bowtie R^{Y} ,
  \]
  so $R$ satisfies the join dependency $X \bowtie Y$. The reverse direction
  is analogous to~\Cref{theo:prob}.
  See
  \iffull
    \cref{Mp:mainhelper}
  \else
  \citep[Theorem A.14]{extended}
  \fi
  for details.
\end{proof}

\subsection{Proving and validating the semi-graphoid axioms} \label{sec:graphoid}
Conditional independence and join dependency are closely related in our
models. Indeed, there is a long line of research on generalizing these
properties to other independence-like notions, and identifying suitable
axioms. \emph{Graphoids} are perhaps the most well-known
approach~\citep{pearl1985graphoids}; \citet{dawid2001} has a similar notion
called \emph{separoids}.

\begin{definition}[Graphoids and semi-graphoids]
  Suppose that $\I{X}{Z}{Y}$ is a ternary relation on subsets of $\Var$
  (i.e., $X, Z, Y \subseteq \Var$). Then $I$ is a \emph{graphoid} if it
  satisfies:
  {\small
    \begin{align}
        &\I{X}{Z}{Y} \Leftrightarrow \I{Y}{Z}{X}
        \tag{\textsc{Symmetry}} \label{graphoid:sym} \\
        &\I{X}{Z}{Y \cup W} \Rightarrow \I{X}{Z}{Y} \land \I{X}{Z}{W}
        \tag{\textsc{Decomposition}} \label{graphoid:decomp}\\
        &\I{X}{Z}{Y \cup W} \Rightarrow \I{X}{Z \cup W}{Y}
        \tag{\textsc{Weak Union}} \label{graphoid:weakunion} \\
        &\I{X}{Z}{Y} \land \I{X}{Z \cup Y}{W} \Leftrightarrow \I{X}{Z}{Y \cup W}
        \tag{\textsc{Contraction}} \label{graphoid:contr} \\
        &\I{X}{Z \cup W}{Y} \land \I{X}{Z \cup Y}{W} \Rightarrow \I{X}{Z}{Y \cup W}
        \tag{\textsc{Intersection}} \label{graphoid:intersect}
    \end{align}
  }
  \noindent If $I$ satisfies the first four properties, then it is a
  \emph{semi-graphoid}.
\end{definition}

Intuitively, $\I{X}{Z}{Y}$ states that knowing $Z$ renders $X$ irrelevant to
$Y$. If we fix a distribution over $\mu \in \DR{\Var}$, then taking
$\I{X}{Z}{Y}$ to be the set of triples such that $\pearlindep{X}{Y}{Z}$ holds
(in $\mu$) defines a semi-graphoid. Likewise, if we fix a relation $R \in
\PR{\Var}$, then the triples of sets of attributes such that $R$ satisfies an
Embedded Multivalue Dependency (EMVD) forms a
semi-graphoid~\citep{DBLP:journals/tods/Fagin77,DBLP:conf/aaai/PearlV87}.

Previously, we showed that the \LOGIC formula $\indep{Z}{X}{Y}$ asserts
conditional independence of $X$ and $Y$ given $Z$ in $\MD$, and join dependency
$X \bowtie Y$ in \MP when $Z = X \cap Y$. Here, we show that the semi-graphoid
axioms can be naturally translated into valid formulas in our concrete models.

\begin{theorem}
  Given a model $M$, define $\I{X}{Z}{Y}$ iff $M \models \indep{Z}{X}{Y}$.
  Then,~\ref{graphoid:sym},~\ref{graphoid:decomp},~\ref{graphoid:weakunion}, and
  ~\ref{graphoid:contr} are valid when $M$ is the probabilistic or the
  relational model. Furthermore, \ref{graphoid:sym} is derivable in the proof
  system, and \ref{graphoid:decomp} is derivable given the following axiom,
  valid in both models:
  \begin{equation}
    \pair{Z}{Y \cup W} \leftrightarrow \pair{Z}{Y} \land \pair{Z}{W} \tag{\textsc{Split}} \label{ax:split}
  \end{equation}
\end{theorem}
\begin{proof}[Proof sketch]
  We comment on the derivable axioms. To derive \ref{graphoid:sym}, we use the
  $\sepand$-\textsc{Comm} proof rule to commute the separating conjunction.
  The proof of \ref{graphoid:decomp} uses the axiom \ref{ax:split} to split up
  $Y \cup W$, and then uses proof rules $\land 3$ and $\land 4$ to prove the
  two conjuncts.  We show derivations
  \iffull
    (\cref{graphoid:symmetry,graphoid:decomposition})
    and prove validity
    (\cref{graphoid:weaku,graphoid:contraction})
    in \cref{app:separoid}.
  \else
		(\cite[Theorems A.15 and A.16]{extended})
    and prove validity
				(\cite[Theorems A.17 and A.18]{extended}).
  \fi
\end{proof}

\section{Application: Conditional Probabilistic Separation Logic} \label{sec:cpsl}
As our final application, we design a separation logic for probabilistic
programs. We work with a simplified probabilistic imperative language with
assignments, sampling, sequencing, and conditionals; our goal is to show how a
\LOGIC-based program logic could work in the simplest setting. For lack of
space, we only show a few proof rules and example programs here; we defer
the full presentation of the separation logic, the metatheory, and the
examples to
\iffull
  \Cref{app:cpsl-full}.
\else
  \cite{extended}.
\fi

\paragraph*{Proof rules}
\SYSTEM includes novel proof rules for randomized conditionals and inherits
the frame rule from PSL~\citep{barthe2019probabilistic}. Here, we show two of
the rules and explain how to use them in the simple program from
\cref{eq:simple}, reproduced here:
\[
  \textsc{Simple} \defeq \Rand x {\Bern_{1/2}}; \Rand y {\Bern_{1/2}} ; z \gets x \vee y
\]
\SYSTEM has Hoare-style rules for sampling and assignments:
\begin{mathpar}	\small
  \inferrule*[Left=Samp]
  { x \not\in \FV(d) \cup \FV(P) }
  { \vdash \psl{P}{\Rand{x}{d}}{P \depand \kp{\FV(d)}{ [x] }} }

  \inferrule*[Left=Assn]
  { x \not\in \FV(e) \cup \FV(P) }
  { \vdash \psl{P}{\Assn{x}{e}}{P \depand \kp{\FV(e)}{ [x] }} }
\end{mathpar}
Using \textsc{Samp} and the fact that the coin-flip distribution $\Bern_{1/2}$
has no free variables, we can infer:
\[
  \vdash \psl{\top}{\Rand{x}{\Bern_{1/2}}}{\empdom{\Dist[x]}} \qquad \vdash
  \psl{\top}{\Rand{y}{\Bern_{1/2}}}{\empdom{\Dist[y]}}
\]
Applying a variant of the frame rule, we are able to derive:
\[
  \vdash \psl{\top}{\Rand{x}{\Bern_{1/2}};\Rand{y}{\Bern_{1/2}}}{\empdom{\Dist[x]}\sepand \empdom{\Dist[y]}}
\]
Using \textsc{Assn} on $P=\empdom{\Dist[x]}\sepand \empdom{\Dist[y]}$ and the
fact that $z$ is not a free variable in either $P$ or $x\lor y$:
\[
  \vdash \psl{P}{\Assn{z}{x\vee y}}{P \depand \kp{\{x,y\}}{ [z] }}
\]
Putting it all together, we get the validity of triple:
\[
  \vdash \psl{\top}
  {\textsc{Simple}}
  {(\empdom{\Dist[x]} \sepand \empdom{\Dist[y]}) \depand \kp{\{x,y\}}{ [z] }}
\]
stating that $z$ depends on $x$ and $y$, which are independent.

\paragraph*{Example programs}
\Cref{fig:examples} introduces two example programs. \ExONE (\Cref{fig:prog-1})
models a distribution where two random observations share a common cause.
Specifically, we consider $z$, $x$, and $y$ to be independent random samples,
and $a$ and $b$ to be values computed from $(x, z)$ and $(y, z)$, respectively.
Intuitively, $z$, $x$, $y$ could represent independent noisy measurements, while
$a$ and $b$ could represent quantities derived from these measurements. Since
$a$ and $b$ share a common source of randomness $z$, they are not independent.
However, $a$ and $b$ are independent conditioned on the value of $z$---this is a
textbook example of conditional independence. Our program logic can establish
the following judgment capturing this fact:
\[
  \vdash \psl{\top}
  {\ExONE}
  {\kp{\Exact{\emptyset}}{\Dist[z]} \depand (\kp{\Exact{z}}{\Dist[a]} \sepand \kp{\Exact{z}}{\Dist[b]})}
\]

The program \ExTWO (\Cref{fig:prog-2}) is a bit more complex: it branches on
a random value $z$, and then assigns $x$ and $y$ with two independent
samples from $\Bern_p$ in the true branch, and $\Bern_q$ in the false
branch. While we might think that $x$ and $y$ are independent at the end of
the program since they are independent at the end of each branch, this is
not true because their distributions are different in the two branches. For
example, suppose that $p = 1$ and $q = 0$. Then at the end of the first
branch $(x, y) = (\ktt, \ktt)$ with probability $1$, while at the end of the
second branch $(x, y) = (\kff, \kff)$ with probability $1$. Thus observing
whether $x = \ktt$ or $x = \kff$ determines the value of $y$---clearly, $x$
and $y$ can't be independent. However, $x$ and $y$ \emph{are} independent
conditioned on $z$. Using our program logic's proof rules for conditionals,
we are able to prove the following judgment capturing this fact:
\[
  \vdash \psl{\top}
  {\ExTWO}
  {\kp{\Exact{\emptyset}}{\Dist[z]} \depand (\kp{\Exact{z}}{\Dist[x]} \sepand \kp{\Exact{z}}{\Dist[y]})}
\]

The full development of the separation logic, consisting of a proof system, a
soundness theorem, along with the detailed verification of the two examples
above, can be found in
\iffull
  \Cref{app:cpsl-full}.
\else
  \cite{extended}.
\fi


\begin{figure}[!t]\small
  \begin{subfigure}[b]{0.48\linewidth}
    \begin{center}
      \[
        \begin{array}{l}
          \Rand{z}{\Bern_{1/2}}; \\
          \Rand{x}{\Bern_{1/2}}; \\
          \Rand{y}{\Bern_{1/2}}; \\
          \Assn{a}{x \lor z}; \\
          \Assn{b}{y \lor z}
        \end{array}
      \]
    \end{center}
    \caption{\ExONE}
    \label{fig:prog-1}
  \end{subfigure}
  \begin{subfigure}[b]{0.48\linewidth}
    \begin{center}
      \[
        \begin{array}{l}
          \Rand{z}{\Bern_{1/2}}; \\
          \Condt{z}{} \\
          \qquad\Rand{x}{\Bern_{p}};
          \Rand{y}{\Bern_{p}} \\
          \mathbf{else} \\
          \qquad\Rand{x}{\Bern_{q}};
          \Rand{y}{\Bern_{q}}
        \end{array}
      \]
    \end{center}
    \caption{\ExTWO}
    \label{fig:prog-2}
  \end{subfigure}
  \caption{Example programs}
  \label{fig:examples}
  \vspace{-5mm}
\end{figure}

\section{Related Work} \label{sec:rw}

\paragraph*{Bunched implications and other non-classical logics}
\LOGIC extends the logic of bunched implications
(BI)~\citep{o1999logic}, and shares many similarities: \LOGIC can be
given a Kripke-style resource semantics, just like BI, and our
completeness proof relies on a general framework for proving
completeness for bunched logics~\citep{docherty2019bunched}. The
non-commutative conjunction and exchange rules are inspired by the
logic CKBI~\citep{docherty2019bunched}. The main difference is that
our exchange rule is reversed, due to our reading of separating
conjunction $\sepand$ as ``can be combined independently'', rather
than ``interleaved''. In terms of models, the probabilistic model of
\LOGIC can be seen as a natural extension of the probabilistic model
for BI~\citep{barthe2019probabilistic}---by lifting distributions to
kernels, \LOGIC is able to reason about dependencies, while
probabilistic BI is not.

There are other non-classical logics that aim to model dependencies.
\emph{Independence-friendly (IF) logic}~\citep{HINTIKKA1989571} and
\emph{dependence logic}~\citep{vaananen_2007} introduce new quantifiers and
propositional atoms to state that a variable depends, or does not depend, on
another variable; these logics are each equivalent in expressivity to
existential second-order logic. More recently,
\citet{DBLP:conf/foiks/0001HKMV18} proposed a probabilistic team semantics for
dependence logic, and \citet{hannula2020descriptive} gave a descriptive
complexity result connecting this logic to real-valued Turing machines. Under
probabilistic team semantics, the universal and existential quantifiers bear a
resemblance to our separating and dependent conjunctions, respectively. It would
be interesting to understand the relation between these two logics, akin to how
the semantics of propositional IF forms a model of
BI~\citep{DBLP:journals/synthese/AbramskyV09}

\paragraph*{Conditional independence, join dependency, and logic}
There is a long line of research on logical characterizations of
conditional independence and join dependency. The literature is too
vast to survey here. On the CI side, we can point to work
by~\citet{GeigerPearl93} on graphical models; on the JD side, the
survey by~\citet{DBLP:conf/icalp/FaginV84} describes the history of
the area in database theory. There are several broadly similar
approaches to axiomatizing the general properties of conditional
dependence, including graphoids~\citep{pearl1985graphoids} and
separoids~\citep{dawid2001}.

\paragraph*{Categorical probability}
The view of conditional independence as a factorization of Markov kernels has
previously been explored~\citep{jacobszanasi2017,DBLP:journals/mscs/ChoJ19,FRITZ2020107239}.
Taking a different approach, \citet{DBLP:journals/entcs/Simpson18} has
recently introduced category-theoretic structures for modeling conditional
independence, capturing CI and JD as well as analogues in heaps and nominal
sets~\citep{pitts_2013}. Roughly speaking, conditional independence in heaps
requires two disjoint portions except for a common overlap contained in the
part that is conditioned; this notion can be smoothly accommodated in our
framework as a \LOGIC\ model where kernels are Kleisli arrows for the identity
monad (\cite{Brotherston2009}) also consider a similar notion of
separation). Simpson's notion of conditional independence in nominal sets
suggests that there might be a \LOGIC model where kernels are Kleisli arrows for
some monad in nominal sets, although the appropriate monad is unclear.

\paragraph*{Program logics}
Bunched logics are well-known for their role in \emph{separation logics},
program logics for reasoning about
heap-manipulating~\citep{DBLP:conf/csl/OHearnRY01} and concurrent
programs~\citep{OHEARN2007271,DBLP:journals/tcs/Brookes07}.  Recently,
separation logics have been developed for probabilistic programs. Our work is
most related to PSL~\citep{barthe2019probabilistic}, where separation models
probabilistic independence. \citet{DBLP:journals/pacmpl/BatzKKMN19} gives a
different, quantitative interpretation to separation in their logic QSL, and
uses it to verify expected-value properties of probabilistic heap-manipulating
programs. Finally, there are more traditional program logics for probabilistic
program. The \textsc{Ellora} logic by \citet{DBLP:conf/esop/BartheEGGHS18} has
assertions for modeling independence, but works with a classical logic. As a
result, basic structural properties of independence must be introduced as
axioms, rather than being built-in to the logical connectives.

\section{Discussion and Future Directions} \label{sec:conc}


We have presented \LOGIC, a new bunched logic to reason about dependence and
independence, together with its Kripke semantics and a sound and complete
proof system. We provided two concrete models, based on Markov and powerset
kernels, that can capture conditional independence-like notions.  We see
several directions for further investigation.

\paragraph*{Generalizing the two models}
The probabilistic and relational models share many similarities: both $\MD$ and
$\MP$ are sets of Kleisli arrows, and use Kleisli composition to interpret
$\odot$; both $\oplus$ operators correspond to parallel composition. Since both
the distribution and powerset monads are commutative strong
monads~\citep{jacobs1994semantics,kock1970monads}, which come with a
\emph{double strength} bi-functor $st_{A,B}: T(A) \times T(B) \rightarrow T(A
\times B)$ that seems suitable for defining $\oplus$, it is natural to consider
more general models based on Kleisli arrows for such monads. Indeed, variants of
conditional independence could make sense in other settings; taking the multiset
monad instead of the powerset monad would lead to a model where we can assert
join dependency in bags, rather than relations, and the free vector space monad
could be connected to subspace models of the graphoid
axioms~\citep{lauritzen-book}.

However, it is not easy to define an operation generalizing $\oplus$ from our
concrete models. The obvious choice---taking $\oplus$ as $f_1 \oplus f_2 = (f_1
\otimes f_2) ; st$---gives a \emph{total} operation, but in our concrete models
$\oplus$ is partial, since it is not possible to compose two arrows that
disagree on their domain overlap. For instance in the probabilistic model,
there is no sensible way to use $\oplus$ to combine a kernel encoding the normal
distribution $\mathcal{N}(0,1)$ on $x$ with another encoding the Dirac
distribution of $x = 1$. We do not know how to model such coherence
requirements between two Kleisli arrows in a general categorical model, and we
leave this investigation to future work.

\paragraph*{Restriction and intuitionistic \LOGIC}
A challenge in designing the program logic is ensuring that formulas in the
assertion logic satisfy \emph{restriction}
\iffull
  (see~\cref{sec:cpsl-assertions}),
\else
  (see~\cite{extended}),
\fi
and one may wonder if a classical version of \LOGIC would be more suitable for
the program logic---if assertions were not required to be preserved under kernel
extensions, it might be easier to show that they satisfy restriction.  However,
a classical logic would require assertions to specify the dependence structure
of \emph{all} variables, which can be quite complicated. Moreover,
intuitionistic logics like probabilistic BI can also satisfy the restriction
property, so the relevant design choice is not classical versus intuitionistic.

Rather, the more important point appears to be whether the preorder can extend
a kernel's domain. If this is allowed---as in \LOGIC---then kernels satisfying
an assertion may have extraneous variables in the domain. However, this choice
also makes the dependent conjunction $P \depand Q$ more flexible: $Q$ does not
need to exactly describe the domain of the second kernel, which is useful since
the range of the first kernel cannot be constrained by $P$. This underlying
tension---allowing the range to be extended, while restricting the domain---is
an interesting subject for future investigation.

\section*{Acknowledgments}
We thank the anonymous reviewers for thoughtful comments and feedback.  This
work was partially supported by the EPSRC grant (EP/S013008/1), the ERC
Consolidator Grant AutoProbe (\#101002697) and a Royal Society Wolfson
Fellowship.  This work was also partially supported by the NSF (\#2023222 and
\#1943130) and Facebook.
  \bibliographystyle{IEEEtranN}
  \bibliography{header,ref}

\newcommand{\SortNoop}[1]{}
\begin{thebibliography}{47}
\providecommand{\natexlab}[1]{#1}
\providecommand{\url}[1]{#1}
\csname url@samestyle\endcsname
\providecommand{\newblock}{\relax}
\providecommand{\bibinfo}[2]{#2}
\providecommand{\BIBentrySTDinterwordspacing}{\spaceskip=0pt\relax}
\providecommand{\BIBentryALTinterwordstretchfactor}{4}
\providecommand{\BIBentryALTinterwordspacing}{\spaceskip=\fontdimen2\font plus
\BIBentryALTinterwordstretchfactor\fontdimen3\font minus
  \fontdimen4\font\relax}
\providecommand{\BIBforeignlanguage}[2]{{%
\expandafter\ifx\csname l@#1\endcsname\relax
\typeout{** WARNING: IEEEtranN.bst: No hyphenation pattern has been}%
\typeout{** loaded for the language `#1'. Using the pattern for}%
\typeout{** the default language instead.}%
\else
\language=\csname l@#1\endcsname
\fi
#2}}
\providecommand{\BIBdecl}{\relax}
\BIBdecl

\bibitem[Kozen(1981)]{kozen81}
\BIBentryALTinterwordspacing
D.~Kozen, ``Semantics of probabilistic programs,'' \emph{Journal of Computer
  and System Sciences}, vol.~22, no.~3, pp. 328--350, 1981. [Online].
  Available: \url{https://doi.org/10.1016/0022-0000(81)90036-2}
\BIBentrySTDinterwordspacing

\bibitem[Gordon et~al.(2014)Gordon, Graepel, Rolland, Russo, Borgstr{\"{o}}m,
  and Guiver]{tabular}
\BIBentryALTinterwordspacing
A.~D. Gordon, T.~Graepel, N.~Rolland, C.~V. Russo, J.~Borgstr{\"{o}}m, and
  J.~Guiver, ``Tabular: a schema-driven probabilistic programming language,''
  in \emph{{ACM} {SIGPLAN--SIGACT} {S}ymposium on {P}rinciples of {P}rogramming
  {L}anguages ({POPL}), San Diego, California}.\hskip 1em plus 0.5em minus
  0.4em\relax {ACM}, 2014, pp. 321--334. [Online]. Available:
  \url{https://doi.org/10.1145/2535838.2535850}
\BIBentrySTDinterwordspacing

\bibitem[Goodman et~al.(2012)Goodman, Mansinghka, Roy, Bonawitz, and
  Tenenbaum]{church}
\BIBentryALTinterwordspacing
N.~D. Goodman, V.~K. Mansinghka, D.~M. Roy, K.~Bonawitz, and J.~B. Tenenbaum,
  ``Church: a language for generative models,'' \emph{CoRR}, 2012. [Online].
  Available: \url{http://arxiv.org/abs/1206.3255}
\BIBentrySTDinterwordspacing

\bibitem[Wood et~al.(2014)Wood, van~de Meent, and Mansinghka]{anglican}
F.~D. Wood, J.~van~de Meent, and V.~Mansinghka, ``A new approach to
  probabilistic programming inference,'' in \emph{International Conference on
  {A}rtificial {I}ntelligence and {S}tatistics ({AISTATS}), Reykjavik,
  Iceland}, 2014, pp. 1024--1032.

\bibitem[Ehrhard et~al.(2018)Ehrhard, Pagani, and Tasson]{tasson}
\BIBentryALTinterwordspacing
T.~Ehrhard, M.~Pagani, and C.~Tasson, ``Measurable cones and stable, measurable
  functions: a model for probabilistic higher-order programming,''
  \emph{Proceedings of the {ACM} on Programming Languages}, no. {POPL}, pp.
  59:1--59:28, 2018. [Online]. Available: \url{https://doi.org/10.1145/3158147}
\BIBentrySTDinterwordspacing

\bibitem[Staton et~al.(2016)Staton, Yang, Wood, Heunen, and Kammar]{kammar}
\BIBentryALTinterwordspacing
S.~Staton, H.~Yang, F.~D. Wood, C.~Heunen, and O.~Kammar, ``Semantics for
  probabilistic programming: higher-order functions, continuous distributions,
  and soft constraints,'' in \emph{{IEEE} {S}ymposium on {L}ogic in {C}omputer
  {S}cience ({LICS}), New York, New York}.\hskip 1em plus 0.5em minus
  0.4em\relax {ACM}, 2016, pp. 525--534. [Online]. Available:
  \url{https://doi.org/10.1145/2933575.2935313}
\BIBentrySTDinterwordspacing

\bibitem[Dahlqvist and Kozen(2020)]{dahlqvist}
\BIBentryALTinterwordspacing
F.~Dahlqvist and D.~Kozen, ``Semantics of higher-order probabilistic programs
  with conditioning,'' \emph{Proceedings of the {ACM} on Programming
  Languages}, no. {POPL}, pp. 57:1--57:29, 2020. [Online]. Available:
  \url{https://doi.org/10.1145/3371125}
\BIBentrySTDinterwordspacing

\bibitem[Barocas et~al.(2019)Barocas, Hardt, and
  Narayanan]{barocas-hardt-narayanan}
S.~Barocas, M.~Hardt, and A.~Narayanan, \emph{Fairness and Machine Learning},
  2019, \url{http://www.fairmlbook.org}.

\bibitem[Barthe et~al.(2019)Barthe, Hsu, and Liao]{barthe2019probabilistic}
G.~Barthe, J.~Hsu, and K.~Liao, ``A probabilistic separation logic,''
  \emph{Proceedings of the {ACM} on Programming Languages}, no. {POPL}, pp.
  55:1--55:30, 2019.

\bibitem[Pearl and Paz(1985)]{pearl1985graphoids}
J.~Pearl and A.~Paz, \emph{Graphoids: A graph-based logic for reasoning about
  relevance relations}.\hskip 1em plus 0.5em minus 0.4em\relax ~: University of
  California (Los Angeles). Computer Science Department, 1985.

\bibitem[O'Hearn and Pym(1999)]{o1999logic}
P.~W. O'Hearn and D.~J. Pym, ``The logic of bunched implications,''
  \emph{Bulletin of Symbolic Logic}, vol.~5, pp. 215--244, 1999.

\bibitem[O'Hearn et~al.(2001)O'Hearn, Reynolds, and
  Yang]{DBLP:conf/csl/OHearnRY01}
\BIBentryALTinterwordspacing
P.~W. O'Hearn, J.~C. Reynolds, and H.~Yang, ``Local reasoning about programs
  that alter data structures,'' in \emph{International Workshop on Computer
  Science Logic (CSL), Paris, France}, 2001, pp. 1--19. [Online]. Available:
  \url{https://doi.org/10.1007/3-540-44802-0\_1}
\BIBentrySTDinterwordspacing

\bibitem[Galmiche et~al.(2019)Galmiche, Marti, and M{\'e}ry]{Galmiche2019}
D.~Galmiche, M.~Marti, and D.~M{\'e}ry, ``Relating labelled and label-free
  bunched calculi in {BI} logic,'' in \emph{Automated Reasoning with Analytic
  Tableaux and Related Methods}.\hskip 1em plus 0.5em minus 0.4em\relax
  Springer International Publishing, 2019, pp. 130--146.

\bibitem[Docherty(2019)]{docherty2019bunched}
S.~Docherty, ``Bunched logics: a uniform approach,'' Ph.D. dissertation, UCL
  (University College London), 2019.

\bibitem[Galmiche and Larchey-Wendling(2006)]{Galmiche2006}
D.~Galmiche and D.~Larchey-Wendling, ``Expressivity properties of {B}oolean
  {BI} through relational models,'' in \emph{Foundations of Software Technology
  and Theoretical Computer Science (FSTTCS), Kolkata, India}.\hskip 1em plus
  0.5em minus 0.4em\relax Springer, 2006, pp. 357--368.

\bibitem[Cao et~al.(2017)Cao, Cuellar, and Appel]{Cao2017}
Q.~Cao, S.~Cuellar, and A.~W. Appel, ``Bringing order to the separation logic
  jungle,'' in \emph{Asian Symposium on Programming Languages and Systems
  (APLAS), Suzhou, China}.\hskip 1em plus 0.5em minus 0.4em\relax Springer,
  2017, pp. 190--211.

\bibitem[Pym et~al.(2004)Pym, O'Hearn, and Yang]{Pym2004}
\BIBentryALTinterwordspacing
D.~J. Pym, P.~W. O'Hearn, and H.~Yang, ``Possible worlds and resources: the
  semantics of {BI},'' \emph{Theoretical Computer Science}, vol. 315, no.~1,
  pp. 257--305, 2004. [Online]. Available:
  \url{http://www.sciencedirect.com/science/article/pii/S0304397503006248}
\BIBentrySTDinterwordspacing

\bibitem[Hoare et~al.(2011)Hoare, M{\"o}ller, Struth, and Wehrman]{Hoare2011}
T.~Hoare, B.~M{\"o}ller, G.~Struth, and I.~Wehrman, ``Concurrent {K}leene
  algebra and its foundations,'' \emph{The Journal of Logic and Algebraic
  Programming}, vol.~80, no.~6, pp. 266--296, 2011.

\bibitem[Goldblatt(1989)]{Goldblatt1989}
\BIBentryALTinterwordspacing
R.~Goldblatt, ``Varieties of complex algebras,'' \emph{Annals of Pure and
  Applied Logic}, vol.~44, no.~3, pp. 173--242, 1989. [Online]. Available:
  \url{http://www.sciencedirect.com/science/article/pii/0168007289900328}
\BIBentrySTDinterwordspacing

\bibitem[Giry(1982)]{giry1982categorical}
M.~Giry, ``A categorical approach to probability theory,'' \emph{Categorical
  aspects of topology and analysis}, pp. 68--85, 1982.

\bibitem[Moggi(1991)]{Moggi91}
\BIBentryALTinterwordspacing
E.~Moggi, ``Notions of computation and monads,'' \emph{Information and
  Computation}, vol.~93, no.~1, pp. 55--92, 1991, selections from 1989 IEEE
  Symposium on Logic in Computer Science. [Online]. Available:
  \url{http://www.sciencedirect.com/science/article/pii/0890540191900524}
\BIBentrySTDinterwordspacing

\bibitem[Panangaden(2009)]{Panangaden2009}
P.~Panangaden, \emph{Labelled Markov Processes}.\hskip 1em plus 0.5em minus
  0.4em\relax Imperial College Press, 2009.

\bibitem[Fagin(1977)]{DBLP:journals/tods/Fagin77}
\BIBentryALTinterwordspacing
R.~Fagin, ``Multivalued dependencies and a new normal form for relational
  databases,'' \emph{{ACM} Trans. Database Syst.}, vol.~2, no.~3, pp. 262--278,
  1977. [Online]. Available: \url{https://doi.org/10.1145/320557.320571}
\BIBentrySTDinterwordspacing

\bibitem[Abiteboul et~al.(1995)Abiteboul, Hull, and
  Vianu]{abiteboul1995foundations}
S.~Abiteboul, R.~Hull, and V.~Vianu, \emph{Foundations of databases}.\hskip 1em
  plus 0.5em minus 0.4em\relax ~: Addison-Wesley Reading, 1995, vol.~8.

\bibitem[Dawid(1979)]{dawid1979}
A.~P. Dawid, ``Conditional independence in statistical theory,'' \emph{Journal
  of the Royal Statistical Society: Series B (Methodological)}, vol.~41, no.~1,
  pp. 1--15, 1979.

\bibitem[Dawid(2001)]{dawid2001}
------, ``Separoids: A mathematical framework for conditional independence and
  irrelevance,'' \emph{Annals of Mathematics and Artificial Intelligence},
  vol.~32, no. 1-4, pp. 335--372, 2001.

\bibitem[Pearl and Verma(1987)]{DBLP:conf/aaai/PearlV87}
\BIBentryALTinterwordspacing
J.~Pearl and T.~Verma, ``The logic of representing dependencies by directed
  graphs,'' in \emph{{AAAI} Conference on Artificial Intelligence, Seattle,
  WA}, 1987, pp. 374--379. [Online]. Available:
  \url{http://www.aaai.org/Library/AAAI/1987/aaai87-067.php}
\BIBentrySTDinterwordspacing

\bibitem[Hintikka and Sandu(1989)]{HINTIKKA1989571}
\BIBentryALTinterwordspacing
J.~Hintikka and G.~Sandu, ``Informational independence as a semantical
  phenomenon,'' in \emph{Logic, Methodology and Philosophy of Science VIII},
  ser. Studies in Logic and the Foundations of Mathematics.\hskip 1em plus
  0.5em minus 0.4em\relax Elsevier, 1989, vol. 126, pp. 571--589. [Online].
  Available:
  \url{http://www.sciencedirect.com/science/article/pii/S0049237X08700661}
\BIBentrySTDinterwordspacing

\bibitem[V\"a\"an\"anen(2007)]{vaananen_2007}
J.~V\"a\"an\"anen, \emph{Dependence Logic: A New Approach to Independence
  Friendly Logic}, ser. London Mathematical Society Student Texts.\hskip 1em
  plus 0.5em minus 0.4em\relax Cambridge University Press, 2007.

\bibitem[Durand et~al.(2018)Durand, Hannula, Kontinen, Meier, and
  Virtema]{DBLP:conf/foiks/0001HKMV18}
\BIBentryALTinterwordspacing
A.~Durand, M.~Hannula, J.~Kontinen, A.~Meier, and J.~Virtema, ``Probabilistic
  team semantics,'' in \emph{International Symposium on Foundations of
  Information and Knowledge Systems ({FoIKS}), Budapest, Hungary}, ser. Lecture
  Notes in Computer Science, vol. 10833.\hskip 1em plus 0.5em minus 0.4em\relax
  Springer, 2018, pp. 186--206. [Online]. Available:
  \url{https://doi.org/10.1007/978-3-319-90050-6\_11}
\BIBentrySTDinterwordspacing

\bibitem[Hannula et~al.(2020)Hannula, Kontinen, Van~den Bussche, and
  Virtema]{hannula2020descriptive}
M.~Hannula, J.~Kontinen, J.~Van~den Bussche, and J.~Virtema, ``Descriptive
  complexity of real computation and probabilistic independence logic,'' in
  \emph{{IEEE} {S}ymposium on {L}ogic in {C}omputer {S}cience ({LICS}),
  Saarbr{\"u}cken, Germany}, 2020, pp. 550--563.

\bibitem[Abramsky and
  V{\"{a}}{\"{a}}n{\"{a}}nen(2009)]{DBLP:journals/synthese/AbramskyV09}
\BIBentryALTinterwordspacing
S.~Abramsky and J.~A. V{\"{a}}{\"{a}}n{\"{a}}nen, ``From {IF} to {BI},''
  \emph{Synthese}, vol. 167, no.~2, pp. 207--230, 2009. [Online]. Available:
  \url{https://doi.org/10.1007/s11229-008-9415-6}
\BIBentrySTDinterwordspacing

\bibitem[Geiger and Pearl(1993)]{GeigerPearl93}
\BIBentryALTinterwordspacing
D.~Geiger and J.~Pearl, ``Logical and algorithmic properties of conditional
  independence and graphical models,'' \emph{The Annals of Statistics},
  vol.~21, no.~4, pp. 2001--2021, 1993. [Online]. Available:
  \url{http://www.jstor.org/stable/2242326}
\BIBentrySTDinterwordspacing

\bibitem[Fagin and Vardi(1984)]{DBLP:conf/icalp/FaginV84}
\BIBentryALTinterwordspacing
R.~Fagin and M.~Y. Vardi, ``The theory of data dependencies - an overview,'' in
  \emph{International Colloquium on Automata, Languages and Programming
  (ICALP), Antwerp, Belgium}, 1984, pp. 1--22. [Online]. Available:
  \url{https://doi.org/10.1007/3-540-13345-3\_1}
\BIBentrySTDinterwordspacing

\bibitem[Jacobs and Zanasi(2017)]{jacobszanasi2017}
\BIBentryALTinterwordspacing
B.~Jacobs and F.~Zanasi, ``A formal semantics of influence in bayesian
  reasoning,'' in \emph{International Symposium on Mathematical Foundations of
  Computer Science ({MFCS}), Aalborg, Denmark}, ser. Leibniz International
  Proceedings in Informatics, vol.~83.\hskip 1em plus 0.5em minus 0.4em\relax
  Schloss Dagstuhl--Leibniz Center for Informatics, 2017, pp. 21:1--21:14.
  [Online]. Available: \url{https://doi.org/10.4230/LIPIcs.MFCS.2017.21}
\BIBentrySTDinterwordspacing

\bibitem[Cho and Jacobs(2019)]{DBLP:journals/mscs/ChoJ19}
\BIBentryALTinterwordspacing
K.~Cho and B.~Jacobs, ``Disintegration and bayesian inversion via string
  diagrams,'' \emph{Math. Struct. Comput. Sci.}, vol.~29, no.~7, pp. 938--971,
  2019. [Online]. Available: \url{https://doi.org/10.1017/S0960129518000488}
\BIBentrySTDinterwordspacing

\bibitem[Fritz(2020)]{FRITZ2020107239}
\BIBentryALTinterwordspacing
T.~Fritz, ``A synthetic approach to markov kernels, conditional independence
  and theorems on sufficient statistics,'' \emph{Advances in Mathematics}, vol.
  370, pp. 107--239, 2020. [Online]. Available:
  \url{http://www.sciencedirect.com/science/article/pii/S0001870820302656}
\BIBentrySTDinterwordspacing

\bibitem[Simpson(2018)]{DBLP:journals/entcs/Simpson18}
\BIBentryALTinterwordspacing
A.~Simpson, ``Category-theoretic structure for independence and conditional
  independence,'' in \emph{Conference on the Mathematical Foundations of
  Programming Semantics (MFPS), Halifax, Canada}, 2018, pp. 281--297. [Online].
  Available: \url{https://doi.org/10.1016/j.entcs.2018.03.028}
\BIBentrySTDinterwordspacing

\bibitem[Pitts(2013)]{pitts_2013}
A.~M. Pitts, \emph{Nominal Sets: Names and Symmetry in Computer Science}, ser.
  Cambridge Tracts in Theoretical Computer Science.\hskip 1em plus 0.5em minus
  0.4em\relax Cambridge University Press, 2013.

\bibitem[Brotherston and Calcagno(2009)]{Brotherston2009}
\BIBentryALTinterwordspacing
J.~Brotherston and C.~Calcagno, ``Classical {BI}: A logic for reasoning about
  dualising resources,'' in \emph{{ACM} {SIGPLAN--SIGACT} {S}ymposium on
  {P}rinciples of {P}rogramming {L}anguages ({POPL}), Savannah, Georgia}.\hskip
  1em plus 0.5em minus 0.4em\relax {ACM}, 2009, pp. 328–--339. [Online].
  Available: \url{https://doi.org/10.1145/1480881.1480923}
\BIBentrySTDinterwordspacing

\bibitem[O’Hearn(2007)]{OHEARN2007271}
\BIBentryALTinterwordspacing
P.~W. O’Hearn, ``Resources, concurrency, and local reasoning,''
  \emph{Theoretical Computer Science}, vol. 375, no.~1, pp. 271--307, 2007,
  festschrift for John C. Reynolds{'}s 70th birthday. [Online]. Available:
  \url{http://www.sciencedirect.com/science/article/pii/S030439750600925X}
\BIBentrySTDinterwordspacing

\bibitem[Brookes(2007)]{DBLP:journals/tcs/Brookes07}
\BIBentryALTinterwordspacing
S.~Brookes, ``A semantics for concurrent separation logic,'' \emph{Theoretical
  Computer Science}, vol. 375, no. 1--3, pp. 227--270, 2007. [Online].
  Available: \url{https://doi.org/10.1016/j.tcs.2006.12.034}
\BIBentrySTDinterwordspacing

\bibitem[Batz et~al.(2019)Batz, Kaminski, Katoen, Matheja, and
  Noll]{DBLP:journals/pacmpl/BatzKKMN19}
\BIBentryALTinterwordspacing
K.~Batz, B.~L. Kaminski, J.~Katoen, C.~Matheja, and T.~Noll, ``Quantitative
  separation logic: a logic for reasoning about probabilistic pointer
  programs,'' \emph{Proceedings of the {ACM} on Programming Languages}, no.
  {POPL}, pp. 34:1--34:29, 2019. [Online]. Available:
  \url{https://doi.org/10.1145/3290347}
\BIBentrySTDinterwordspacing

\bibitem[Barthe et~al.(2018)Barthe, Espitau, Gaboardi, Gr{\'{e}}goire, Hsu, and
  Strub]{DBLP:conf/esop/BartheEGGHS18}
\BIBentryALTinterwordspacing
G.~Barthe, T.~Espitau, M.~Gaboardi, B.~Gr{\'{e}}goire, J.~Hsu, and P.~Strub,
  ``An assertion-based program logic for probabilistic programs,'' in
  \emph{European Symposium on Programming (ESOP), Thessaloniki, Greece}, 2018,
  pp. 117--144. [Online]. Available:
  \url{https://doi.org/10.1007/978-3-319-89884-1\_5}
\BIBentrySTDinterwordspacing

\bibitem[Jacobs(1994)]{jacobs1994semantics}
B.~Jacobs, ``Semantics of weakening and contraction,'' \emph{Annals of pure and
  applied logic}, vol.~69, no.~1, pp. 73--106, 1994.

\bibitem[Kock(1970)]{kock1970monads}
A.~Kock, ``Monads on symmetric monoidal closed categories,'' \emph{Archiv der
  Mathematik}, vol.~21, no.~1, pp. 1--10, 1970.

\bibitem[Lauritzen(1996)]{lauritzen-book}
S.~Lauritzen, \emph{\BIBforeignlanguage{English}{Graphical Models}}.\hskip 1em
  plus 0.5em minus 0.4em\relax Clarendon Press, 1996.

\end{thebibliography}
  \clearpage

  \iffull
    \onecolumn

\appendix

\subsection{Section \ref{sec:logic}: omitted proof}
\label{app:logic}

\begin{lemma}
	\label{sep2dep}
	$P \sepand Q \vdash P \depand Q$
\end{lemma}
\begin{proof}
	For better readability, we break the proof tree down into two components.
	\begin{prooftree}\small
	\AxiomC{}
	\RightLabel{$\depand$-\textsc{Right Unit}}
	\UnaryInfC{$P \vdash P \depand \sepid$}
	\AxiomC{}
	\RightLabel{$\depand$-\textsc{Left Unit}}
	\UnaryInfC{$Q \vdash \sepid \depand Q$}
	\RightLabel{$\sepand$-\textsc{Conj}}
	\BinaryInfC{$P \sepand Q \vdash (P \depand \sepid) \sepand (\sepid \depand Q)$}
	\AxiomC{}
	\RightLabel{\textsc{RevEx}}
	\UnaryInfC{$(P \depand \sepid) \sepand (\sepid \depand Q) \vdash
			(P \sepand \sepid) \depand (\sepid \sepand Q)$}
	\RightLabel{\textsc{Cut}}
	\BinaryInfC{$P \sepand Q \vdash (P \sepand \sepid) \depand (\sepid \sepand Q)$}
	\end{prooftree}
	With $P \sepand Q \vdash (P \sepand \sepid) \depand (\sepid \sepand Q)$, we construct the following
	\begin{prooftree}\small
	\AxiomC{$P \sepand Q \vdash (P \sepand \sepid) \depand (\sepid \sepand  Q)$}
	\AxiomC{}
	\RightLabel{$\sepand$-\textsc{Unit}}
	\UnaryInfC{$P \sepand I \vdash P$}
	\AxiomC{}
	\RightLabel{$\sepand$-\textsc{Comm}}
	\UnaryInfC{$\sepid \sepand Q \vdash Q \sepand \sepid$}
	\AxiomC{}
	\RightLabel{$\sepand$-\textsc{Unit}}
	\UnaryInfC{$Q \sepand \sepid \vdash Q$}
	\RightLabel{\textsc{Cut}}
	\BinaryInfC{$\sepid \sepand Q \vdash Q$}
	\RightLabel{$\depand$-\textsc{Conj}}
	\BinaryInfC{$(P \sepand \sepid) \depand (\sepid \sepand Q) \vdash P \depand Q$}
	\RightLabel{\textsc{Cut}}
	\BinaryInfC{$P \sepand Q \vdash P \depand Q$}
	\end{prooftree}
This proof uses the admissible rule $\textsc{Cut}$, which can be derived as follows:
\begin{prooftree}\small
\AxiomC{$Q \vdash R$}
\RightLabel{$\land 2$}
\UnaryInfC{$P \land Q \vdash R$}
\RightLabel{$\rightarrow$}
\UnaryInfC{$P \vdash Q \rightarrow R$}
\AxiomC{$P \vdash Q$}
\RightLabel{\textsc{MP}}
\BinaryInfC{$P \vdash R$}
\end{prooftree}
\end{proof}

\subsection{Section~\ref{sec:sc}, Soundness and Completeness: Omitted Details} \label{app:completeness}

\algebraicsc*

\begin{proof}
Soundness can be established by a straightforward induction on the proof rules.
For completeness, we can define a Lindenbaum-Tarski algebra by quotienting
$\mathrm{Form_{\LOGIC}}$ by the equivalence relation $P \equiv Q$ iff $P \vdash Q$ and $Q \vdash P$ derivable.
This yields a \LOGIC algebra, and moreover, $[P]_{\equiv} \leq [Q]_{\equiv}$ iff
$[P \rightarrow Q]_{\equiv} = [\top]_{\equiv}$ iff
$\top \vdash P \rightarrow Q$ derivable iff $P \vdash Q$ derivable. Hence for any $P, Q$ such that
$P \vdash Q$ is \emph{not} derivable, in the Lindenbaum-Tarski algebra
(with the canonical interpretation sending formulas to their equivalence class)
$[P]_{\equiv} \not\leq [Q]_{\equiv}$ holds, establishing completeness.
\end{proof}

A \emph{filter} on a bounded distributive lattice $\mathbb{A}$ is a non-empty set
$F \subseteq A$ such that, for all $x, y \in A$, (1) $x \in F$ and $x \leq y$ implies $y \in F$; and
(2) $x, y \in F$ implies $x \land y \in F$. It is a \emph{proper} filter if it additionally satisfies
(3) $\bot \not\in F$, and a \emph{prime} filter if in addition it also satisfies
(4)   $x \lor y \in F$ implies $x \in F$ or $y \in F$.
The order dual version of these definitions gives the notions of ideal, proper ideal and prime ideal.
We denote the sets of proper and prime filters of $\mathbb{A}$ by $\mathbb{F}_{\mathbb{A}}$ and
$\mathbb{PF}_{\mathbb{A}}$ respectively, and the sets of proper and prime ideals of
$\mathbb{A}$ by $\mathbb{I}_{\mathbb{A}}$ and $\mathbb{PI}_{\mathbb{A}}$ respectively.

To prove that prime filter frames are \LOGIC frames we require an auxiliary lemma that can be used to establish the existence of prime filters. First some terminology: a \emph{$\subseteq$-chain} is a sequence of sets $(X_{\alpha})_{\alpha < \lambda}$ such that $\alpha \leq \alpha'$ implies $X_{\alpha} \subseteq X_{\alpha'}$. A basic fact about proper filters (ideals) is that the union of a $\subseteq$-chain of proper filters (ideals) is itself a proper filter (ideal). We lift the terminology to $n$-tuples of sets by determining $(X^1_{\alpha}, \ldots, X^n_{\alpha})_{\alpha < \lambda}$ to be a $\subseteq$-chain if each $(X^i_{\alpha})_{\alpha < \lambda}$ is a $\subseteq$-chain.

\begin{definition}[Prime Predicate]
A \emph{prime predicate} is a map $P: \mathbb{F}_{\mathbb{A}}^n \times \mathbb{I}_{\mathbb{A}}^m \rightarrow \{ 0, 1\}$, where $n, m \geq 0$ and $n + m \geq 1$, such that
	\begin{itemize}
		\item[a)] Given a $\subseteq$-chain $(F^0_{\alpha}, \ldots, F^n_{\alpha}, I^0_{\alpha}, \ldots,
				 I^m_{\alpha})_{\alpha < \lambda}$ of proper filters/ideals,
				\[min \{ P(F^0_{\alpha}, \ldots, I^m_{\alpha}) \mid \alpha < \lambda \} \leq P(\bigcup_\alpha
					F^0_\alpha, \ldots, \bigcup_{\alpha} I^m_{\alpha});\]
		\item[b)] $P(\ldots, H_0 \cap H_1, \ldots) \leq \max\{P(\ldots, H_0, \ldots), P(\ldots, H_1, \ldots)\}$. \qedhere
	\end{itemize}
\end{definition}

Intuitively, a prime predicate is a property of proper filter/ideal sequences whose truth value is inherited by unions of chains, and is witnessed by one of $H_0$ or $H_1$ whenever witnessed by $H_0 \cap H_1$. The proof of the next lemma can be found in \cite{docherty2019bunched}.

\begin{lemma}[Prime Extension Lemma {\cite[Lemma 5.7]{docherty2019bunched}}] \label{lem:primeextension}
If $P$ is an $(n+m)$-ary prime predicate and $F_0, \ldots, F_n, I_0, \ldots, I_m$ an $(n+m)$-tuple
of proper filters and ideals such that $P(F_0, \ldots, F_n, I_0, \ldots, I_m) = 1$
then there exists a $(n+m)$-tuple of prime filters and ideals
$F^{pr}_0, \ldots, F^{pr}_n$,
$I^{pr}_0, \ldots I^{pr}_m$ such that
$P(F^{pr}_0, \ldots, F^{pr}_n, I^{pr}_0, \ldots I^{pr}_m)=1$. \qed
\end{lemma}

Now, whenever prime filters are required that satisfy a particular property
(for example, an existentially quantified consequent of a frame axiom),
it is sufficient to show that the property defines a prime predicate and
there exists proper filters satisfying it. We also note the following useful properties of
\LOGIC algebras, which are special cases of those found in \cite[Proposition 6.2]{docherty2019bunched}.

\begin{lemma}
Given any \LOGIC algebra $\mathbb{A}$, for all $a, b, c \in A$ and $\circ \in \{ \sepand, \depand \}$, the following properties hold:
	\[ \begin{array}{cl}
		(a \lor b) \circ c = (a \circ c) \lor (b \circ c) & a \circ (b \lor c) = (a \circ b) \lor (a \circ c) \\
		a \leq a' \text{ and } b \leq b' \text{ implies } a \circ b \leq a' \circ b' & \bot \circ a = \bot = a \circ \bot
	    \end{array} \] \qedhere
\end{lemma}

\primeisdibi*

\begin{proof}
All but one of the frame axioms can be verified in an identical fashion to the analogous proof for BI
\cite[Lemma 6.24]{docherty2019bunched}, and $\oplus_{\mathbb{A}}$ and $\odot_{\mathbb{A}}$ are both Up-Closed and Down-Closed. We focus on the novel frame axiom: Reverse Exchange.
For readability we omit the $\mathbb{A}$ subscripts on operators.
Assume there are prime filters such that
$F_x \supseteq F_x' \in F_y \oplus F_z$, $F_y \in F_{y_1} \odot F_{y_2}$
and $F_{z} \in F_{z_1} \odot F_{z_2}$.
We will prove that
	\[ \PrimePredicate{F, G}{F_x \in F \odot G \text{ and } F \in F_{y_1} \oplus F_{z_1} \text{ and } G \in F_{y_2} \oplus F_{z_2}} \]
	\jia{Why not $F_x \supseteq F_x'' \in F \odot G$ here}
is a prime predicate, abusing notation to allow $\odot$ and $\oplus$ to be defined for non-prime filters.

For a), suppose $(F_{\alpha}, G_{\alpha})_{\alpha \leq \lambda}$ is a $\subseteq$-chain
such that for all $\alpha$, $P(F_\alpha, G_\alpha) = 1$. Call $F = \bigcup_\alpha F_\alpha$
and $G = \bigcup_\alpha G_\alpha$. We must show that $P(F,G) = 1$.
Let $a \in F$, $b \in G$. Then $a \in F_\alpha$, $b \in G_\beta$ for some
$\alpha$, $\beta$. Wolog, we may assume $\alpha \leq \beta$.
Then since $F_x \in F_\beta \odot G_\beta$, we have that $a \depand b \in F_x$ as required, so $F_x \in F \odot G$.
$F \in F_{y_1} \oplus F_{z_1} \text{ and } G \in F_{y_2} \oplus F_{z_2}$ hold trivially.

For b), suppose for contradiction that $P(F \cap F', G) = 1, P(F, G) = 0$ and $P(F', G) = 0$.
From $P(F \cap F', G) = 1$ we know $F, F' \in F_{y_1} \oplus F_{y_2}$:
for all $a \in F_{y_1}, b \in F_{y_2}, a \sepand b \in F \cap F' \subseteq F, F'$.
So the only way this can be the case is if $F_x \not\in F \odot G$ and $F_x \not\in F' \odot G$.
Hence there exists $a \in F, b \in G$ such that $a \depand b \not\in F_x$, and $a' \in F', b' \in G$
such that $a' \depand b' \not\in F_x$. It follows by properties of filters that
$a \lor a' \in F \cap F'$ and $b'' = b \land b' \in G$. Hence $(a \lor a') \sepand b'' \in F_x$ by assumption,
and $(a \lor a') \sepand b'' = (a \sepand b'') \lor (a' \sepand b'')$.
Since $F_x$ is prime, this means either $a \sepand b'' \in F_x$ or $a' \sepand b'' \in F_x$.
But that's not possible: $a \sepand b'' \leq a \sepand b$ and $a' \sepand b'' \leq a' \sepand b'$,
so whichever holds results in a contradiction.
Hence either $P(F, G) = 1$ or $P(F', G) = 1$ as required.
The argument for the second component is similar.

Now consider $F = \{ c \mid \exists a \in F_{y_1}, b \in F_{z_1}(c \geq a \sepand b)\}$ and
$G = \{ c \mid \exists a \in F_{y_2}, b \in F_{z_2}(c \geq a \sepand b)\}$.
These are both proper filters. Focusing on $F$ (both arguments are essentially identical),
it is clearly upwards-closed. Further, it is closed under $\land$: if $c, c' \in F$ because
$c \geq a \sepand b$ and $c' \geq a' \sepand b'$ for $a, a' \in F_{y_1}$ and $b, b' \in F_{z_1}$
then $c \land c' \geq (a \sepand b) \land (a' \sepand b') \geq
(a \land a') \sepand (b \land b')$, with $a \land a' \in F_{y_1}$ and $b \land b' \in F_{z_1}$.
It is proper, because if $\bot \in F$, then there exists $a \in F_{y_1}$ and $b \in F_{z_1}$
such that $a * b = \bot$. Let $c \in F_{y_2}$ and $d \in F_{z_2}$ be arbitrary.
Then by our initial assumption, $a \depand c \in  F_y$ and
$b \depand d \in  F{z}$.
Hence $(a \depand c) \sepand (b \depand d) \in F_{x'} \subseteq F_x$.
However,  by the Reverse Exchange algebraic axiom,
$(a \depand c) \sepand (b \depand d) \leq (a * b) \depand (c * d) = \bot \depand (c * d) = \bot$.
By upwards-closure, $\bot \in F_x$, which is supposed to be a prime, and therefore proper, filter, which gives a contradiction.

By definition, $F \in F_{y_1} \oplus F_{z_1}$ and $G \in F_{y_2} \oplus F_{z_2}$.
To see that $F_x \in F \odot G$, let $c \in F$ (with $c \geq a \sepand b$ for some $a \in F_{y_1}$ and $b \in F_{z_1}$) and
$c' \in G$ (with $c' \geq a' \sepand b'$ for some $a' \in F_{y_2}$ and $b \in F_{z_2}$).
By assumption $a \depand a' \in  F_y$ and $b \depand b' \in F_z$, and so
$(a \depand a') \sepand (b \depand b') \in F_{x'} \subseteq F_x$.
By the algebraic Reverse Exchange axiom, we obtain
$(a \sepand b) \depand (a' \sepand b') \in F_x$, and by monotonicity of $\depand$ and
upwards-closure of $F_x$ we obtain
$c \depand c' \in F_x$. Hence $P(F, G) = 1$ and by
Lemma \ref{lem:primeextension} there are prime $F, G$ with $P(F, G) = 1$.
This verifies that the Reverse Exchange frame axiom holds.
\jia{what is monotonicity of $\depand$}
\end{proof}

\complexisdibi*

\begin{proof}
	We focus on the Reverse Exchange algebraic axiom (the other \LOGIC algebra properties can be proven in identical fashion to the analogous proof for BI \cite[Lemma 6.22]{docherty2019bunched}). Suppose $x \in (A \triangleright B) \bullet (C \triangleright D)$.
	Then there exists $x', y, z$ such that $x \sqsupseteq x' \in y \oplus z$, with $y \in A \triangleright B$ and $z \in C \triangleright D$. In turn, there thus exists $y_1, y_2, z_1, z_2$ such that $y \in y_1 \odot y_2$ and $z \in z_1 \odot z_2$ with
	$y_1 \in A, y_2 \in B, z_1 \in C$ and $z_2 \in D$.
	By the Reverse Exchange frame axiom, there exist $u, v$ such that $u \in y_1 \oplus z_1$, $v \in y_2 \oplus z_2$ and $x' \in u \odot v$.
	Hence $u \in A \bullet C$, $v \in B \bullet D$ and $x' \in (A \bullet C) \triangleright (B \bullet D)$. Since $x' \sqsubseteq x$ and
	$(A \bullet C) \triangleright (B \bullet D)$ is an upwards-closed set, $x \in (A \bullet C) \triangleright (B \bullet D)$ as required.
\end{proof}

Now clearly every persistent valuation $\mathcal{V}$ on a Kripke frame $\mathcal{X}$ generates an algebraic interpretation $\llbracket - \rrbracket_{\mathcal{V}}$ on $Com(\mathcal{X})$ with the property that $x \models_{\mathcal{V}} P$ iff $x \in \llbracket P \rrbracket$ (note that the complex algebra operations are defined precisely as the corresponding semantic clauses). Similarly, by the Representation Theorem, given an algebraic interpretation $\llbracket - \rrbracket$ on $\mathbb{A}$, a persistent valuation $\mathcal{V}_{\llbracket - \rrbracket}$ on $Pr(\mathbb{A})$ can be defined by $\mathcal{V}_{\llbracket - \rrbracket}(p) = \{ F \in \mathbb{PF}_{\mathbb{A}} \mid \llbracket p \rrbracket \in F \} = \theta_{\mathbb{A}}(\llbracket p \rrbracket)$. That $\theta$ is a monomorphism into $Com(Pr(\mathbb{A}))$ establishes that, for all $P \in \mathrm{Form_{\LOGIC}}$, $F \models_{\mathcal{V}_{\llbracket - \rrbracket}} P$ iff $\llbracket P \rrbracket \in F$.

\kripkesc*

\begin{proof}
Assume $P \not\models Q$. Then there exists a \LOGIC model $(\mathcal{X}, \mathcal{V})$ and a state $x \in X$ such that $x \models P$ but $x \not\models Q$. Hence $\llbracket P \rrbracket_{\mathcal{V}} \not\subseteq \llbracket Q \rrbracket_{\mathcal{V}}$ in $Com(\mathcal{X})$, so, by Theorem \ref{thm:algebraicsc}, $P \vdash Q$ is not derivable. Now assume $P \vdash Q$ is not derivable. By Theorem \ref{thm:algebraicsc} there exists a \LOGIC algebra
$\mathbb{A}$ and an interpretation $\llbracket - \rrbracket$ such that $\llbracket P \rrbracket \not\leq \llbracket Q \rrbracket$. From this it can be established that there is a prime filter $F$ on $\mathbb{A}$ such that $ \llbracket P \rrbracket \in F$ and $\llbracket Q \rrbracket \not\in F$. Hence
$F \models_{\mathcal{V}_{\llbracket - \rrbracket}} P$ but $F \not\models_{\mathcal{V}_{\llbracket - \rrbracket}} Q$, so $P \not\models Q$.
\end{proof}

\subsection{Section \ref{sec:probmodel}, probabilistic model: omitted proofs}
			\label{app:probabilistic}

			\paragraph{Remark}
			In the following, we sometimes abbreviate $\dom(f_i)$ as $\dm{i}$ and
			$\range(f_i)$	as $\rg{i}$.

%
In the proof of \Cref{MdisLOGIC} we use that $\MD$ is closed under $\oplus$ and $\odot$, which we prove next.
\begin{lemma}
\label{Md:closed}
$\MD$ is closed under $\oplus$ and $\odot$.
\end{lemma}
\begin{proof}
For any $f_1, f_2 \in \MD$, we need to show that
		\begin{itemize}[leftmargin=*]
					\item If $f_1 \oplus f_2$ is defined, then $f_1 \oplus f_2 \in \MD$.
Recall that  $f_1 \oplus f_2$ is defined if and only if $\rg{1} \cap \rg{2} = \dm{1} \cap \dm{2}$, which implies that $(\rg{1} \cup \rg{2}) \setminus (\dm{1} \cup \dm{1}) =(\rg{1} \setminus \dm{1}) \cup (\rg{2} \setminus \dm{2}) $, and $(\rg{1} \setminus \dm{1}) \cap (\rg{2} \setminus \dm{2}) = \emptyset$.

	So we can split any memory assignment on $(\rg{1} \cup \rg{2}) \setminus (\dm{1} \cup \dm{2})$ into two disjoint parts, one on $\rg{1} \setminus \dm{1}$, another on $\rg{2} \setminus \dm{2}$.

State	$f_1 \oplus f_2$ preserves the input because  for any $d \in \R{\dm{1} \cup \dm{2}}$, we can obtain ($\star$):
\begin{align*}
&(\pi_{\dm{1} \cup \dm{2}} (f_1 \oplus f_2))(d)(d) \\
&= \sum_{x} (f_1 \oplus f_2)(d)(d \bowtie x)  \tag{$x \in \R{(\rg{1} \cup \rg{2}) \setminus (\dm{1} \cup \dm{2} )}$} \\
&\stackrel\dagger= \sum_{x_1, x_2} f_1(d^{\dm{1}})(d^{\dm{1}} \bowtie x_1) \cdot f_2(d^{\dm{2}})(d^{\dm{2}} \bowtie x_2) \tag{$x_1 \in \R{\rg{1} \setminus \dm{1}},\, x_2 \in \R{\rg{2} \setminus \dm{2}}$}\\
&= \left(\sum_{x_1 \in \R{\rg{1} \setminus \dm{1}}} f_1(d^{\dm{1}})(d^{\dm{1}} \bowtie x_1) \right) \cdot \left(\sum_{ x_2 \in \R{\rg{2} \setminus \dm{2}}} f_2(d^{\dm{2}})(d^{\dm{2}} \bowtie x_2) \right) \notag \\
&= 1 \cdot 1 =1 \tag{Using $f_1, f_2 \in \MD$} \\
\end{align*}
Step $\dagger$ follows using $(\rg{1} \cup \rg{2}) \setminus (\dm{1} \cup \dm{1}) =(\rg{1} \setminus \dm{1}) \cup (\rg{2} \setminus \dm{2})$ and $(\rg{1} \setminus \dm{1}) \cap (\rg{2} \setminus \dm{2}) = \emptyset$.
Then, for any $d \in \R{\dm{1} \cup \dm{2}}$, $(f_1 \oplus f_2)(d)$ is a distribution since:
\begin{align*}
&\sum_{m \in \R{\rg{1} \cup \rg{2}}} (f_1 \oplus f_2)(d)(m) \\
&= \sum_{m \in \R{\rg{1} \cup \rg{2}}} f_1(d^{\dm{1}})(m^{\rg{1}}) \cdot f_2(d^{\dm{2}})(m^{\rg{2}}) \\
&\stackrel\ddagger= \sum_{x_1,x_2} f_1(d^{\dm{1}})(d^{\dm{1}} \bowtie x_1) \cdot f_2(d^{\dm{2}})(d^{\dm{2}} \bowtie x_2) \tag{$x_1 \in \R{\rg{1} \setminus \dm{1}},\, x_2 \in \R{\rg{2} \setminus \dm{2}}$}\\
&=1 \tag{Using ($\star$)}
\end{align*}
Step $\ddagger$ follows using $(\rg{1} \setminus \dm{1}) \cap (\rg{2} \setminus \dm{2}) = \emptyset$, and the $f_i$ term is 0 when $d^{\dm{i}} \neq m^{\dm{i}}$.

Thus, $f_1 \oplus f_2$ is a kernel in $\MD$.
\item If $f_1 \odot f_2$ is defined, then $f_1 \odot f_2 \in \MD$. Recall that $f_1 \odot f_2: \R{\dm{1}} \rightarrow \DR{\rg{2}}$ is defined iff $\rg{1} = \dm{2}$.
$f_1 \odot f_2$ preserves the input because for any $d \in \R{\dm{1}}$, we can obtain $(\spadesuit)$
\begin{align*}
&(\pi_{\dm{1}} f_1 \odot f_2) (d)(d)\\
 &= \sum_{x \in \R{\rg{2} \setminus \dm{1}}} (f_1 \odot f_2)(d)(d\bowtie x) \notag \\
&=\sum_{x \in \R{\rg{2} \setminus \dm{1}}} f_1(d)(d \bowtie x^{\rg{1} \setminus \dm{1}}) \cdot f_2(d \bowtie x^{\rg{1} \setminus \dm{1}})(d \bowtie x) \notag \\
&=\sum_{x_1} f_1(d)(d \bowtie x_1) \cdot \left( \sum_{x_2} f_2(d \bowtie x_1)(d \bowtie x_1 \bowtie x_2) \right) \tag{$x_1 \in \R{\rg{1} \setminus \dm{1}}$, $x_2 \in \R{\rg{2} \setminus \rg{1}} $} \\
&=\sum_{x_1 \in \R{\rg{1} \setminus \dm{1}}} \left( f_1(d)(d \bowtie x_1) \cdot 1 \right) \tag{Using $f_2 \in \MD$ } \notag \\
&=1
\end{align*}
Then, for any $d \in \dm{1}$, $(f_1 \odot f_2)(d)$ is a distribution as
\begin{align*}
\sum_{m \in \rg{2}} (f_1 \odot f_2)(d)(m) &= \sum_{m \in \rg{2}} f_1(d)(m^{\rg{1}}) \cdot f_2(m^{\rg{1}})(m) \tag{\Cref{eq:simplemu} } \\
																																										&\stackrel \heartsuit=\sum_{x \in \rg{2} \setminus \dm{1}} f_1(d)(d \bowtie x^{\rg{1} \setminus \dm{1}}) \cdot f_2(d \bowtie x^{\rg{1} \setminus \dm{1}})(d \bowtie x)\\
 &=1 \tag{Using ($\spadesuit$)}
\end{align*}
Step $\heartsuit$ follows since the $f_i$ term is 0 when $d^{\dm{i}} \neq m^{\dm{i}}$.

Thus $f_1 \odot f_2$ is a kernel in $\MD$. \qedhere
\end{itemize}
\end{proof}

            \begin{lemma}
	\label{Md:tmodel}
              The probabilistic model $\MD$ is a $\T$-model defined in~\cref{def:tmodel},
 for $\T=\DD$.
            \end{lemma}

\begin{proof}$\MD$ satisfies condition (1)--(4) and (10) by construction, so we only prove (5)--(9).
            \begin{itemize}
            \item[(5)] We show that when  $(f \oplus g) \oplus h$ and $f \oplus (g \oplus h)$ are defined,  $(f \oplus g) \oplus h = f \oplus (g \oplus h)$. Consider
            $f\colon \m{S} \rightarrow \DR{S \cup T}$, $g\colon \m{U} \rightarrow \DR{U \cup V}$, and $h\colon \m{W} \rightarrow \DR{W \cup X}$.
                  For any $d \in \R{S \cup U \cup W}$, and $m \in \R{S \cup T \cup U \cup V \cup W \cup X}$,
                  \begin{align*}
                    ((f \oplus g) \oplus h) (d)(m)&=  \big(f(d^S)(m^{S\cup T}) \cdot g(d^{U})(m^{U\cup V}) \big)  \cdot h (d^{W})(m^{W\cup X}) \tag{def. $\oplus$}\\
                                                         &=   f(d^S)(m^{S\cup T}) \cdot \big(g(d^{U})(m^{U\cup V})   \cdot h (d^{W})(m^{W\cup X})\big)  \\
                                                         &= (f \oplus (g \oplus h)) (d)(m)
                  \end{align*}

					\item[(6)] When $f_1 \oplus f_2$ and $f_2 \oplus f_1$ defined, $f_1 \oplus f_2 = f_2 \oplus f_1$.

 For any $d \in \R{\dm{1} \cup \dm{2}}$, $m \in \DR{\rg{1} \cup \rg{2}}$ such that $d \bowtie m$ is defined,
	\begin{align*}
(f_1 \oplus f_2) (d)(m) &\defeq f_1(d^{\dm{1}})(m^{\rg{1}}) \cdot f_2(d^{\dm{2}})(m^{\rg{2}}) =  f_2(d^{\dm{2}})(m^{\rg{2}}) \cdot f_1(d^{\dm{1}})(m^{\rg{1}}) = (f_2 \oplus f_1)(d)(m)
\end{align*}

Thus,  $f_1 \oplus f_2 = f_2 \oplus f_1$.

              \item[(7)] For any $f\colon  \R{A} \rightarrow \DR{A \cup X} \in M$, and any $S \subseteq A$, we must show
                \begin{align*}
                  f \oplus \unit_{S} = f
                \end{align*}
		Since $S \subseteq A$, we have $\dom(f \oplus \unit_S) = A \cup S = A = \dom(f)$ and $\range(f \oplus \unit_S) = A \cup X \cup S = A \cup X = \range(f)$.
		For any $d \in \R{A}$, and any $r \in \R{A \cup X}$ such that $d \otimes r$ is defined, we have
				\begin{align*}
				(f \oplus \unit_S) (d)(r) &= f(d)(r) \cdot \unit(d^S)(r^S) \\
				     				  &= f(d)(r) \cdot 1 = f(d)(r)
			        \end{align*}
		Hence, $f \oplus \unit_S = f$.

              \item[(8)] We show that when both $(f_1 \oplus f_2) \odot (f_3 \oplus f_4)$ and $(f_1 \odot f_3) \oplus (f_2 \odot f_4)$ are defined, it hold that
                \begin{align*}
                  (f_1 \oplus f_2) \odot (f_3 \oplus f_4) &=(f_1 \odot f_3) \oplus (f_2 \odot f_4).
                \end{align*}
              First note that the well-definedness of both terms we can conclude that $D_1 \subseteq R_1 = D_3 \subseteq R_3$, $D_2 \subseteq R_2 = D_4 \subseteq R_4$, where $D_i = \dom(f_i)$ and $R_i = \range(f_i)$. Moreover, both terms are of type $\R{D_1 \cup D_2} \rightarrow \DR{R_3 \cup R_4}$, and, for any $d \in D_1 \cup D_2$ and $m \in R_3 \cup R_4$:
		\begin{align*}
		\big(	(f_1 \oplus f_2) \odot (f_3 \oplus f_4) \big) (d) (m) &= (f_1 \oplus f_2)(d)(m^{R_1 \cup R_2}) \cdot (f_3 \oplus f_4)(m^{D_3 \cup D_4})(m) \tag{\Cref{eq:simplemu}}\\
		 				 &= \big( f_1(d^{D_1}) (m^{R_1}) \cdot f_2(d^{D_2})(m^{R_2}) \big) \cdot \big( f_3(m^{D_3}) (m^{R_3}) \cdot f_4(m^{D_4})(m^{R_4}) \big) \\
		\big(	(f_1 \odot f_3) \oplus (f_2 \odot f_4) \big)(d) (m)  &= (f_1 \odot f_3)(d^{D_1})(m^{R_3}) \cdot (f_2 \odot f_4)(d^{D_2})(m^{R_3}) \\
 &= \big( f_1(d^{D_1})(m^{R_1}) \cdot f_3(d^{D_3})(m^{R_3}) \big) \cdot  \big( f_2(d^{D_2})(m^{R_2}) \cdot f_4(d^{D_4})(m^{R_4}) \big)\\
&= \big( f_1(d^{D_1}) (m^{R_1}) \cdot f_2(d^{D_2})(m^{R_2}) \big) \cdot \big( f_3(m^{D_3}) (m^{R_3}) \cdot f_4(m^{D_4})(m^{R_4}) \big) \\
		\end{align*}
		Thus, $(f_1 \odot f_3) \oplus (f_2 \odot f_4)  = (f_1 \oplus f_2) \odot (f_3 \oplus f_4) $.
	\item [(9)] Proved in~\cref{Md:closed}
		\qedhere
            \end{itemize}
\end{proof}

\MdisLOGIC*
			\begin{proof}
				By~\cref{Mframeaxioms} that all $\T$-models are \LOGIC frames and by~\cref{Md:tmodel} that $\MD$ is a $\T$-model, $\MD$ is a \LOGIC frame.
			\end{proof}

\subsection{Section \ref{sec:relmodel}, relational model: omitted proofs}
			\label{app:relational}
%
%
For the proof of \Cref{thm:relframe} we need the following closure property.
	\begin{lemma}
	\label{Mp:closed}
	$\MP$ is closed under $\oplus$ and $\odot$.
	\end{lemma}
	\begin{proof}
	For any $f_1, f_2 \in \MP$, we need to show that :
	\begin{itemize}[leftmargin=*]
	\item If $f_1 \oplus f_2$ is defined, then $f_1 \oplus f_2 \in \MP$.
		Recall that  $f_1 \oplus f_2$ is defined if and only if $\rg{1} \cap \rg{2} = \dm{1} \cap \dm{2}$, which implies that
		\begin{align*}
		(\dm{1} \cup \dm{2}) \cap \rg{1} = (\dm{1} \cap \rg{1}) \cup (\dm{2} \cap \rg{1}) = \dm{1} \cup (\dm{2} \cap \dm{1}) = \dm{1} \\
		(\dm{1} \cup \dm{2}) \cap \rg{2} = (\dm{1} \cap \rg{2}) \cup (\dm{2} \cap \rg{2}) = (\dm{1} \cap \dm{2}) \cup \dm{2} = \dm{2}
		\end{align*}
	We show $f_1 \oplus f_2$ also preserves the input: for any $d \in \R{\dm{1} \cup \dm{2}}$,
                      \begin{align*}
                        (\pi_{\dm{1} \cup \dm{2}} (f_1 \oplus f_2))(d)
                        &= \pi_{\dm{1} \cup \dm{2}} ((f_1 \oplus f_2)(d)) \\
                        &= \pi_{\dm{1} \cup \dm{2}} f_1(d^{\dm{1}}) \bowtie f_2(d^{\dm{2}}) \\
                        &\stackrel\dagger= \pi_{\dm{1}} f_1(d^{\dm{1}}) \bowtie \pi_{\dm{1}} f_2(d^{\dm{2}})\\
                        &= \{d^{\dm{1}}\} \bowtie \{d^{\dm{2}}\} \tag{Because $f_1, f_2 \in \MP$} \\
                        &= \{d\}.
                      \end{align*}
         Step $\dagger$  follows because $(\dm{1} \cup \dm{2}) \cap \rg{1} = \dm{1}$ and $(\dm{1} \cup \dm{2}) \cap \rg{2} = \dm{2}$.

	\item If $f_1 \odot f_2$ is defined, then $f_1 \odot f_2 \in \MP$.
	Recall $f_1 \odot f_2$ is defined iff $\rg{1} = \dm{2}$, and gives a map of type  $\R{\dm{1}} \rightarrow \DR{\rg{2}}$.
	We show that $f_1 \odot f_2$ preserves the input: for any $d \in \R{\dm{1}}$,
			\begin{align*}
			(\pi_{\dm{1}} f_1 \odot f_2) (d) &= (\pi_{\dm{1}} f_1) (d) \tag{Because $\dm{1} \subseteq \rg{1} = \dm{2}$} \\
    		      &= \unit_{\dm{1}}(d)
			\end{align*}
	Thus, $\pi_{\dm{1}} f_1 \odot f_2 = \unit_{\dm{1}}$ and hence $f_1 \odot f_2$ preserves the input. \qedhere
	\end{itemize}
	\end{proof}

            \begin{lemma}
	\label{Mp:tmodel}
              The relational model $\MP$ is a $\T$-model~\cref{def:tmodel}
for the monad $\T=\PP$.
            \end{lemma}

				\begin{proof}
				 $\MP$ satisfies conditions (1)--(4) and (10) by construction, so we only prove (5)--(9).
            \begin{itemize}
              \item[(5)] We show that when both  $(f \oplus g) \oplus h$ and $f \oplus (g \oplus h)$ are defined, $(f \oplus g) \oplus h = f \oplus (g \oplus h)$. Consider
            $f\colon \m{S} \rightarrow \PR{S \cup T}$, $g\colon \m{U} \rightarrow \PR{U \cup V}$, and $h\colon \m{W} \rightarrow \PR{W \cup X}$.
                  For any $d \in \R{S \cup U \cup W}$,
                  \begin{align*}
			((f \oplus g) \oplus h) (d)&=  \big(f(d^{S}) \bowtie  f_2(d^{U}) \big) \bowtie f_3(d^{V}) \\
	     &= f(d^{S}) \bowtie  \big(g(d^{U}) \bowtie h(d^{V}) \big)  \tag{By associativity of $\bowtie$} \\
                                                         &= (f \oplus (g \oplus h)) (d)
                  \end{align*}
					\item[(6)] When both $f_1 \oplus f_2$ and $f_2 \oplus f_1$ are defined, they are equal.

Analogous to $\MD$, instead of followed from the commutativity of $\cdot$, it follows from the commutativity of $\bowtie$.

              \item[(7)] For any $f\colon  \R{A} \rightarrow \PR{A \cup X}$, and any $S \subseteq A$,  we must show
                \begin{align*}
                  f \oplus \unit_{S} = f
                \end{align*}
		Since $S \subseteq A$, so $\dom(f \oplus \unit_S) = A \cup S = A = \dom(f)$, and $\range(f \oplus \unit_S) = A \cup X \cup S = A \cup X = \range(f)$.
		For any $d \in \R{A}$, we have
		\begin{align*}
		(f \oplus \unit_S) (d) &= f(d) \bowtie \unit_S(d^S) = f(d) \bowtie \{d^S\} = f(d)
		\end{align*}
		Hence, $f \oplus \unit_S = f$.

              \item[(8)]We show that when both $(f_1 \oplus f_2) \odot (f_3 \oplus f_4)$ and $(f_1 \odot f_3) \oplus (f_2 \odot f_4)$ are defined, it hold that
                \begin{align*}
                  (f_1 \oplus f_2) \odot (f_3 \oplus f_4) &=(f_1 \odot f_3) \oplus (f_2 \odot f_4).
                \end{align*}
             Take $D_i = \dom(f_i)$ and $R_i = \range(f_i)$ and note that well-definedness of the above terms implies that $R_1=D_3$ and $R_2=D_4$. Both terms have type $\R{D_1 \cup D_2} \rightarrow \PR{R_3 \cup R_4}$, and, for any $d \in D_1 \cup D_2$:
                           {\small
                  \begin{align*}
                    \big((f_1 \oplus f_2) \odot (f_3 \oplus f_4) \big) (d)
                    &= \{v \mid u \in (f_1 \oplus f_2)(d) ,\,   v \in (f_3 \oplus f_4)(u) \} \\
                    &= \{v \mid u \in f_1(d^{D_1}) \bowtie f_2(d^{D_2}) ,\,  v \in f_3(u^{D_3}) \bowtie f_4(u^{D_4}) \} \tag{Def. $\oplus$}\\
                    &= \{v \mid v \in f_3(x) \bowtie f_4(y),\, x \in f(d^{D_1}) ,\,  y \in g(d^{D_2}) \} \tag{$\star$}\\
                   & = \{ v_1 \bowtie v_2 \mid v_1 \in f_3(x),\, v_2 \in f_4(y),\, x \in f(d^{D_1}) ,\, y \in g(d^{D_2}) \} \\
                    \big((f_1 \odot f_3) \oplus (f_2 \odot f_4) \big)(d)
                    &= (f_1 \odot f_3)(d^{D_1}) \bowtie (f_2 \odot f_4)(d^{D_2})\tag{Def. $\oplus$}\\
                    &= \{ v_1 \mid u_1 \in f_1(d^{D_1}) ,\,  v_1 \in f_3(u_1) \} \bowtie \{ v_2 \mid u_2 \in f_2(d^{D_2}) ,\, v_2 \in f_4(u_2)\} \tag{Def. $\otimes$}\\
                   &= \{ v_1\bowtie v_2  \mid v_1 \in f_3(u_1),\, v_2 \in f_4(u_2), u_1 \in f_1(d^{D_1}),\, u_2 \in f_2(d^{D_2})\}
                  \end{align*}
                  }
                  The step marked with ($\star$) follows from the fact that $R_1=D_3$ and $R_2=D_4$ implies that for any $u \in f(d^{D_1}) \bowtie g(d^{D_2}) $, we have $u^{D_3} = x \in f_1(d^{D_1})$ and $u^{D_3} = y \in f_1(d^{D_1})$. \qedhere
					\item [(9)] Proved in~\cref{Mp:closed}.
            \end{itemize}
                \end{proof}

\thmrelframe*
			\begin{proof}

				By~\cref{Mframeaxioms} that all $\T$-models are \LOGIC frames and by~\cref{Mp:tmodel} that $\MP$ is a $\T$-model, $\MP$ is a \LOGIC frame.
			\end{proof}

				\subsection{Section \ref{section:condindep}, Conditional Independence: Omitted Details}
			\label{app:verifyCI}
			First, we prove~\cref{lemma:condfirst} so we can use~\cref{findexact} for \MD.

 \begin{lemma}[Disintegration]
	\label{lemma:condfirst}
  If $f = f_{1} \odot f_{2}$ , then $\pi_{R_1} f = f_{1}$.
  Conversely, if $\pi_{R_1} f = f_{1}$, then there exists $g$ such that $f = f_1
  \odot g$.
\end{lemma}

\begin{proof}
  For the forwards direction, suppose that $f = f_{1} \odot f_{2}$. Then,
  \begin{align*}
    \pi_{R_1} f &= \pi_{R_1} (f_1 \odot f_2) = f_1 \odot (\pi_{R_1} f_2)
                = f_1 \odot \unit_{\R{\rg{1}}}
				= f_1.
  \end{align*}
  Thus, $\pi_{R_1} f = f_{1}$. For the converse, assume $\pi_{R_1} f = f_1$.
  Define $g: \R{R_1} \rightarrow \DR{\range(f)}$ such that for any $r \in
  \R{R_1}$, $m \in \R{\range(f)}$ such that $r \bowtie m$ is defined, let
  \[
    g(r)(m) \defeq
    \begin{cases}
      \frac{f(r^{\dm{1}})(m)}{(\pi_{R_1} f)(r^{\dm{1}})(r)} &: (\pi_{R_1} f)(r^{\dm{1}})(r) \neq 0\\
      0 &: (\pi_{R_1} f)(r^{\dm{1}})(r) = 0
    \end{cases}
  \]
  We need to check that $g \in \MD$. Fixing any $r \in \R{R_1}$,
		denote the distribution $\Pr_{f(r^{D_1})}$ as $\mu_r$, then
  \begin{align*}
    (\pi_{R_1} f)(r^{\dm{1}})(r) &= \frac{\mu_r (\range(f) = m)}{\mu_r (\rg{1} = r)}
    = \mu_r (\range(f) = m \mid \rg{1} = r) \tag{if $(\pi_{R_1} f)(r^{\dm{1}})(r) \neq 0$} \\
    \sum_{m \in \R{\range(g)}} g(r) (m) &= \sum_{m \in \R{\range(g)}} \mu_r (\range(f) = m \mid \rg{1} = r) = 1
  \end{align*}
  so $g$ does map any input to a distribution, and $g$ preserves the input.

  By their types, $f_1 \odot g$ is defined, and for any $d \in \R{\dm{1}}$, $m
  \in \R{\range(f)}$ such that $d \bowtie m$ is defined. If  $(\pi_{R_1}
  f)(d)(m^{R_1}) \neq 0$, then
  \begin{align*}
    (f_1 \odot g) (d)(m) = f_1(d)(m^{\rg{1}}) \cdot g(m^{\rg{1}}) (m) &= f_1(d)(m^{\rg{1}}) \cdot \frac{f(m^{\dm{1}})(m)}{(\pi_{R_1} f)(m^{\dm{1}})(m)}\\
																																																																						& = f_1(d)(m^{\rg{1}}) \cdot \frac{f(m^{\dm{1}})(m)}{f_1 (m^{\dm{1}})(m^{\rg{1}})}
                         \\
                         &= f(d)(m) \tag{$d \bowtie m$ is defined iff $d = m^{\dm{1}}$}
  \end{align*}
  If $(\pi_{R_1} f)(d)(m^{R_1}) \neq 0$, then $f(d)(m) = 0$, and
  \(
    (f_1 \odot g) (d)(m)
    = f_1(d)(m^{\rg{1}}) \cdot g(m^{\rg{1}}) (m)
    = 0
    = f(d)(m).
  \)
  Thus, $f_1 \odot g = f$.
\end{proof}
\theoprob*
\begin{proof}
This result follows by combining~\Cref{Md:probhelper} and~\Cref{findexact}.
\end{proof}
		\begin{lemma}
	\label{Md:probhelper}
	For a distribution $\mu$ on $\Var$, $S, X, Y \subseteq \Var$,
	there exist $f_1\colon  \R{\emptyset} \rightarrow \DR{S}$, $f_2 \colon \R{S} \rightarrow \DR{S \cup X}$, $f_3 \colon \R{S} \rightarrow \DR{S \cup Y}$, such that $f_1 \odot (f_2 \oplus f_3) \sqsubseteq f_{\mu}$, if and only if $\pearlindep{X}{Y}{S}$ and also $X \cap Y \subseteq S$.
\end{lemma}

\begin{proof}
	\textbf{Forward direction}: Assume the existence of $f_1, f_{2}, f_{3}$ satisfying $f_1 \odot (f_2 \oplus f_3) \sqsubseteq f_{\mu}$. We must prove $\pearlindep{X}{Y}{S}$ and $X \cap Y \subseteq S$.
	\begin{enumerate}[leftmargin=*]
		\item 	${X \cap Y \subseteq S}$: $f_{2} \oplus f_{3}$ defined implies $(X \cup S) \cap (Y \cup S) \subseteq S \cap S$.  Thus,
				$X \cap Y \subseteq S$
\item $\pearlindep{X}{Y}{S}$:
By assumption,  $f_{1} \odot (f_{2} \oplus f_{3}) \sqsubseteq f_{\mu}$.
\cref{lemma:condfirst} gives us  $f_{1} \odot (f_{2} \oplus f_{3}) = \pi_{S \cup X \cup Y} (f_{\mu})$, and $f_1 = \pi_S (f_{\mu})$.
For any $m \in \R{X \cup Y \cup S}$, $m^X \bowtie m^Y \bowtie m^S$ is defined. Thus,
			\begin{align*}
			\mu(X = m^X, Y = m^Y, S = m^S ) &= (\pi_{X \cup Y \cup S} \mu)(m^X \bowtie m^Y \bowtie m^S) \tag {By definition $\mu$}\\
			&= \pi_{X \cup Y \cup S} (f_{\mu}) (\empmem) (m^X \bowtie m^Y \bowtie m^S) \\
			&= f_{1} \odot (f_{2} \oplus f_{3})(\empmem) (m^X \bowtie m^Y \bowtie m^S )
	                \end{align*}
Similarly, $\mu(S = m^S) \defeq (\pi_S \mu)(m^S)$. We have $f_{1} = \pi_S(f_{\mu}) $, and so
			\begin{align}
			\mu(S = m^S) &= (\pi_S \mu)(m^S) = \left(\pi_S (f_{\mu})\right) (\empmem) (m^S)
			=f_{1} (\empmem) (m^S) \label{prob:s}
	                 \end{align}
By definition of conditional probability, when $\mu(S = m^S) \neq 0$,
			\begin{align*}
			\mu(X = m^X, Y = m^Y \mid S = m^S ) &= \frac{ \mu(X = m^X, Y = m^Y, S = m^S )}{ \mu(S = m^S )}\\
				&= \frac{ f_{1} \odot (f_{2} \oplus f_{3})(\empmem)(m^S \bowtie m^X \bowtie m^Y)}{f_{1}(\empmem)(m^S)}
\end{align*}
By~\cref{eq:simplemu}:
\(
f_{1} \odot (f_{2} \oplus f_{3})(\empmem)(m^S \bowtie m^X \bowtie m^Y) = f_{1}(\empmem) (m^S) \cdot (f_{2} \oplus f_{3}) (m^S )(m^S \bowtie m^X \bowtie m^Y).
\)\
Thus,
\begin{align}
\mu(X = m^X, Y = m^Y \mid S = m^S ) &= \frac{ f_{1} \odot (f_{2} \oplus f_{3})(\empmem)(m^S \bowtie m^X \bowtie m^Y)}{f_{1}(\empmem)(m^S)} \notag \\
	 &= (f_{2} \oplus f_{3}) (m^S )(m^S \bowtie m^X \bowtie m^Y) \notag \\
		&= f_{2}(m^{S})(m^{X \cup S}) \cdot f_{3}(m^{S})(m^{Y \cup S}) \label{XYgivenS}
\end{align}
Let $f_2' = f_2 \oplus \unit_{\R{Y}}$, $f_3' = f_3 \oplus \unit_{\R{X}}$. By~\cref{oplus2odot},
	\begin{align*}
f_1 \odot (f_2 \oplus f_3) &= f_1 \odot f_2 \odot (f_3 \oplus \unit_{\R{X}}) = f_1 \odot f_2 \odot f_3' \\
f_1 \odot (f_2 \oplus f_3) &= f_1 \odot (f_3 \oplus f_2) = f_1 \odot f_3 \odot (f_2 \oplus \unit_{\R{Y}}) = f_1 \odot f_3 \odot f_2'
      \end{align*}
\Cref{lemma:condfirst} gives us $\pi_{X \cup S} (f_{\mu}) = f_{1} \odot f_{2}$,
and $\pi_{Y \cup S} (f_{\mu}) =f_{1} \odot f_{3}$,
Therefore,
\begin{align}
	\mu(X = m^X, S = m^S) &\defeq (\pi_{X \cup S} \mu) (m^S \otimes m^X) \notag \\
	&= (\pi_{X \cup S} (f_{\mu}))(\empmem ) (m^S \otimes m^X) \notag \\
	&= (f_{1} \odot f_{2})(\empmem) (m^S \otimes m^X)\\ 
	&= f_{1}(\empmem)(m^S) \cdot f_{2}(m^{S}) (m^{S} \otimes m^{X}) \notag \\
	\mu(Y = m^Y, S = m^S) &\defeq (\pi_{Y \cup S} \mu) (m^S \otimes m^Y) \notag\\
	&= (\pi_{Y \cup S} (f_{\mu}) (\empmem) (m^S \otimes m^Y) \notag \\
	&= (f_{1} \odot f_{3})(\empmem) (m^S \otimes m^Y) \\
	&= f_{1}(\empmem)(m^S) \cdot f_{3}(m^S)(m^{S} \otimes m^{Y}) \notag
\end{align}
Thus, by definition of conditional probability.
\begin{align}
	\mu(X = m^X \mid S = m^S) &= \frac{\mu(X = m^X, S =m^S) }{\mu(S = m^S)} \notag \\
			&= \frac{f_{1}(\empmem)(m^S) \cdot f_{2}(m^{S}) (m^{S \cup X})}{f_{1}(\empmem)(m^{S})} \notag \\
			&= f_{2}(m^{S}) (m^{S \cup X}) \label{prob:xconds} \\
	\mu(X = m^Y \mid S = m^S) &= \frac{\mu(X = m^X, S =m^S) }{\mu(S = m^S)} \notag \\
			&= \frac{f_{1}(\empmem)(m^S) \cdot f_{3}(m^S)(m^{S \cup Y})}{f_{1}(\empmem)(m^S)} \notag \\
			&= f_{3}(m^S)(m^{S \cup Y})\label{prob:yconds}
\end{align}
Substituting \cref{prob:xconds} and \cref{prob:yconds} into the equation \cref{XYgivenS}, we have
\begin{align*}
	\mu(X = m^X, Y = m^Y \mid S = m^S ) = \mu(X = m^X \mid S = m^S) \cdot \mu(X = m^Y \mid S = m^S))
\end{align*}
Thus, $X, Y$ are conditionally independent given $S$. This completes the proof for the first direction.
	\end{enumerate}
					\textbf{Backward direction:} We want to show that if  $\pearlindep{X}{Y}{S}$ and $X \cap Y \subseteq S$ then $f_1 \odot (f_2 \oplus f_3) \sqsubseteq f_{\mu}$.
Given $\mu$, we define $f_{1}=\pi_S (f_{\mu})$ and construct $f_{2}, f_{3}$ as follows:

Let $f_{2}: \R{S} \rightarrow \DR{S \cup X}$. For any $s \in \R{S}$, $x \in \R{X}$ such that $s \otimes x$ is defined, when $f_1(\empmem)(s) \neq 0$, let
\begin{align*}
	f_{2}(s)(s \otimes x) \defeq \frac{(\pi_{ S \cup X} f_{\mu})(\empmem) (s \otimes x) }{f_{1}(\empmem)(s)}
\end{align*}
(When $f_1(\empmem)(s) = 0$, we can define $f_2(s)(s \otimes x)$ arbitrarily as long as $f_2(s)$ is a distribution, because that distribution will be zeroed out in $f_1 \odot (f_2 \oplus f_3)$ anyway. )

Similarly, let $f_{3}: \R{S} \rightarrow \DR{S \cup Y}$.
For any $s \in \R{S}$, $x \in \R{Y}$ such that $s \otimes y$ is defined, when $f_1(\empmem)(s) \neq 0$, let
\begin{align*}
	f_{3}(s)(s \otimes y) \defeq \frac{(\pi_{ S \cup Y} f_{\mu}) (s \otimes y) }{f_{1}(\empmem)(s)}
\end{align*}
By construction, $f_1, f_2, f_3$ each has the type needed for the lemma. We are left to prove that given any $s \in \R{S}$, $f_{2}$ and $f_{3}$ are kernels in $\MD$, $f_{1} \odot (f_{2} \oplus f_{3})$ is defined, and $f_{1} \odot (f_{2} \oplus f_{3}) \sqsubseteq f_{\mu}$.

\begin{itemize}
	\item State $f_{2}$ is in $\MD$.

		We need to show that for any $s \in \R{S}$, $f_2(s)$ forms a distribution, and also $f_2$ preserves the input.
		For any $s \in \R{S}$, by equation \cref{prob:s}, $f_{1}(\empmem)(s) = \mu(S = s)$.

		If $f_{1}(\empmem)(s) = 0 $, then we define $f_2(s)$ arbitrarily but make sure $f_2(s)$ is a distribution.

		If $f_{1}(\empmem)(s) \neq 0 $:		for any $x \in \R{X}$ such that $s \otimes x$ is defined, $(\pi_{ S \cup X} f_{\mu})(\empmem) (s \otimes x) = \mu(S = s, X = x)$, so
		\begin{align*}
		&	f_{2}(s)(s \otimes x) = \frac{(\pi_{ S \cup X} f_{\mu})(\empmem) (s \otimes x) }{f_{1}(\empmem)(s)} \\
			= &\frac{\mu(S = s, X = x)}{\mu(S = s)}
			= \mu(X = x \mid S = s)
		\end{align*}
		Thus, $f_2(s)$ is a distribution for any $s \in \R{S}$.

		Also, $f_2(s)(s \otimes x)$ is non-zero only when $s \otimes x$ is defined, i.e., when $(s \otimes x)^S = s$. So $(\pi_S f_2)(s)(s) = \sum_{x \in \R{X}} f_2(s)(s \otimes x) = 1$, and thus $\pi_S f_2= \unit_{\R{S}}$. Therefore, $f_2$ preserves the input.

		Therefore, $f_{2} \in \MD$.

	\item State $f_{3}$ is in $\MD$. Similar as above.

	\item State $f_{1} \odot (f_{2} \oplus f_{3})$ is defined.

		$f_2 \oplus f_3$ is defined because $\rg{2} \cap \rg{3} = (S \cup X) \cap (S \cup Y) = S \cup (X \cap Y)$,
		and by assumption, $X \cap Y \subseteq S$, so $S \cup (X \cap Y) = S = \dm{2} \cap \dm{3}$.
		Then $f_{1} \odot (f_{2} \oplus f_{3})$ is defined because $\dom(f_{2} \oplus f_{3}) = \dm{2} \cup \dm{3} = S \cup S = S = \range(f_{1})$.

	\item State $f_{1} \odot (f_{2} \oplus f_{3})\sqsubseteq f_{\mu}$.

		It suffices to show that there exists $g$ such that $(f_1 \odot (f_2 \oplus f_3)) \odot g = f_{\mu}$.

		For any $s \in \R{S}$, $x \in \R{X}$, $y \in \R{Y}$ such that $s \otimes x \otimes y$ is defined,
		\begin{align}
			f_{1} \odot (f_{2} \oplus f_{3})(\empmem)(s \otimes x \otimes y) &= f_{1}(\empmem)(s) \cdot f_{2} \oplus f_{3}(s)(s \otimes x \otimes y) \notag \\
				&= f_{1}(\empmem)(s) \cdot \left(f_{2}(s)(s \otimes x) \cdot f_{3}(s)(s \otimes y)\right) \notag \\
				&= \mu(S = s) \cdot \left(\mu(X = x \mid S = s) \cdot \mu(Y = y \mid S = s)\right) \label{eq:unfolding}
		\end{align}
		Because $X, Y$ are conditionally independent given $S$ in the distribution $q$, so
		\begin{align}
			\mu(X = x \mid S = s) \cdot \mu(Y = y \mid S = s) = \mu(X = x, Y = y \mid S= s) \label{eq:condassumption}
		\end{align}
		Substituting \cref{eq:condassumption} into \cref{eq:unfolding}, we have
		\begin{align*}
			f_{1} \odot (f_{2} \oplus f_{3})(\empmem)(s \otimes x \otimes y) &= \mu(S = s) \cdot \mu(X = x, Y=y \mid S = s) \\
				&= \mu(X = x, Y = y, S = s)
		\end{align*}
		Let $g: \R{X \cup Y \cup S} \rightarrow \DR{\Val}$ such that for any $d \in \R{X \cup Y \cup S}$, $m \in \R{\Val}$ such that $d \otimes m$ is defined, let
		\[g(d)(m) = \mu (\Val = m \mid X \cup Y \cup S = d) \]
		Then, $(f_1 \odot (f_2 \oplus f_3)) \odot g$ is defined, and
		\begin{align*}
			(f_1 \odot (f_2 \oplus f_3) \odot g) (\empmem) (m) &= (f_1 \odot (f_2 \oplus f_3) )(\empmem)(m^{X \cup Y \cup S}) \cdot g (m^{X \cup Y \cup S})(m) \\
		&= \mu(\Val = m)
		\end{align*}
		Thus, $(f_1 \odot (f_2 \oplus f_3)) \odot g = f_{\mu}$, and therefore $f_{1} \odot (f_{2} \oplus f_{3}) \sqsubseteq f_{\mu}$.
	\end{itemize}
This completes the proof for the backwards direction.
					\end{proof}

					\begin{lemma}
			 \label{concrete:XcapYinS}
							 If $X, Y$ are conditionally independent given $S$, then values on $X \cap Y$ is determined given values on $S$.
		 \end{lemma}
		 \begin{proof}
				For any $x \in \R{X}$, $y \in \R{Y}$, $s \in \R{S}$, $m \in \R{M}$,
				 when $\mu( X = x, Y = y, M = m \mid S = s) \neq 0$,
					it must $x \otimes y \otimes s \otimes m$ is defined.
					Note that $x \otimes y \otimes s \otimes m$ defined only if $m = \pi_M x = \pi_M y$, which implies that $m \otimes x = x$, $m \otimes y = y$, $m \otimes x \otimes y = x \otimes y$.

					Let $M = X \cap Y$, $\wh{X} = X \setminus Y$, $\wh{Y} = Y \setminus X$.
					By assumption, $X, Y $ are conditionally independent given $S$ , so
					 $x \in \R{X}$, $y \in \R{Y}$, $s \in \R{S}$, $m \in \R{M}$
			 \begin{align*}
				 \mu(X = x  \mid S = s) \cdot \mu(Y = y \mid S = s) & = \mu(X = x , Y= y \mid S = s),
				\end{align*}
				which implies that, if we denote $x' = \pi_{\wh{X}} x$, $y' = \pi_{\wh{Y}} y$,
\begin{align}
				\mu(\wh{X} = x', M = m \mid S = s) \cdot \mu(\wh{Y} = y', M =m \mid S = s)  &=  \mu(\wh{X} = x' , \wh{Y} = y', M = m \mid S = s) \label{eq:XYMgivenS}
				\end{align}
				For any probabilistic events $E_1, E_2, E_3$, $\mu(E_1 , E_2 \mid E_3) = \mu(E_1 \mid E_2, E_3) \cdot \mu(E_2 \mid E_3)$. Thus,~\cref{eq:XYMgivenS} implies that
\begin{align}
 &  \mu(\wh{X} = x' \mid M = m, S = s)
	\cdot \mu(\wh{Y} = y' \mid M =m,  S = s)  \cdot \mu(M = m \mid S = s)
	=
	\mu(\wh{X} = x' , \wh{Y}= y' \mid M = m, S = s) \label{eq:XYgivenMS}
\end{align}
Then, for any $s \in \R{S}, m \in \R{M}$ such that $m \otimes s$ is defined and $\mu(M =m, S=s) \neq 0$,
\begin{align}
&\sum_{x' \in \R{\wh{X}}, y' \in \R{\wh{Y}} } \mu(\wh{X} = x' \mid M = m, S = s)
	\cdot \mu(\wh{Y} = y' \mid M =m,  S = s)  \cdot \mu(M = m \mid S = s) \notag\\
	=&
\sum_{x' \in \R{\wh{X}}, y' \in \R{\wh{Y}} } \mu(\wh{X} = x' , \wh{Y}= y' \mid M = m, S = s)
\tag{Because of ~\cref{eq:XYgivenMS}} \\
	=& 1 \label{eq:rhssum}
\end{align}

Meanwhile,  for any $s \in \R{S}, m \in \R{M}$ such that $m \otimes s$ is defined and  $\mu(M =m, S=s) \neq 0$,
\begin{align}
	&\sum_{x' \in \R{\wh{X}}, y' \in \R{\wh{Y}} } \mu(\wh{X} = x' \mid M = m, S = s)
	\cdot \mu(\wh{Y} = y' \mid M =m,  S = s)  \cdot \mu(M = m \mid S = s)\notag \\
	=&	\left(\sum_{x' \in \R{\wh{X}}, y' \in \R{\wh{Y}} } \mu(\wh{X} = x' \mid M = m, S = s)
	\cdot \mu(\wh{Y} = y' \mid M =m,  S = s) \right)  \cdot \mu(M = m \mid S = s) \notag\\
	=&	\left(\sum_{x' \in \R{\wh{X}}} \mu(\wh{X} = x' \mid M = m, S = s)  \right)
\cdot \left( \sum_{y' \in \R{\wh{Y}} } \mu(\wh{Y} = y' \mid M =m,  S = s) \right)  \cdot \mu(M = m \mid S = s) \notag \\
	=& 1  \cdot  \mu(M = m \mid S = s) \label{eq:lhssum}
\end{align}
Combining~\cref{eq:lhssum} and~\cref{eq:rhssum}, we derive  $ \mu(M = m \mid S = s) = 1$. That is, when $\pearlindep{X}{Y}{S}$, whether $M \supseteq S$ or not, $m \otimes s$ is defined and $\mu(M =m, S=s) \neq 0$ implies $\mu(M =m \mid S=s) = 1$. Thus,  $\pearlindep{X}{Y}{S}$ renders values on $X \cap Y$ deterministic given values on $S$.
%
 \end{proof}

					\subsection{Section~\ref{sec:relational}, Join Dependency: Omitted Details}
					We again prove a disintegration lemma for \MP~\cref{Mp:lemma:condfirst} so that we can use~\cref{findexact} on \MP.

			\begin{lemma}[Disintegration]
	\label{Mp:lemma:condfirst}
  If $f = f_{1} \odot f_{2}$ and $D_2 = R_1$, then $\pi_{R_1} f = f_{1}$.
  Conversely, if $\pi_{R_1} f = f_{1}$, then there exists $g$ such that $f = f_1
  \odot g$.
\end{lemma}
\begin{proof}Assume $f = f_{1} \odot f_{2}$ and $D_2 = R_1$.  Then,
 \begin{align*}
    \pi_{R_1} f &= \pi_{R_1} (f_1 \odot f_2) = f_1 \odot (\pi_{R_1} f_2)
                = f_1 \odot \unit_{\R{\rg{1}}}
				= f_1.
  \end{align*}

 Conversely, assume $\pi_{R_1} f = f_1$. Define $g\colon \R{R_1} \rightarrow \PR{R_2}$ by $g(r) = \{ s\otimes r \mid s\in f(r^{D_1})\}$.
 \[(f_1 \odot g)(d) = \{ u \mid u\in g(r), r \in f_1(d) \} = \{ s\otimes r \mid s\in f(r^{D_1}), r \in \pi_{R_1}f(d) \}= \{ s \mid s\in f(d) \}=f(d).\]
\end{proof}

\theorel*
\begin{proof}
The result follows from combining~\cref{Mp:mainhelper} and~\cref{findexact}.
\end{proof}
\begin{lemma}
	\label{Mp:mainhelper}
	For a relation $R$ on $\Val$, $X, Y \subseteq \Val$,
	there exists $f_1\colon \R{\emptyset} \rightarrow \PR{X \cap Y}$,
	$f_2 \colon \R{X \cap Y} \rightarrow \PR{X}$,
	$f_3 \colon \R{X \cap Y} \rightarrow \PR{Y}$,
	such that $f_1 \odot (f_2 \oplus f_3) \sqsubseteq f_R$,
	if and only if $R^{X \cup Y} = R^X \bowtie R^Y$.
\end{lemma}

\begin{proof}\textbf{Forward Direction: } Assuming there exist $f_1, f_{2}, f_{3}$ such that $f_1 \odot (f_2 \oplus f_3) \sqsubseteq f_R$,
	we want to show that $R^{X \cup Y} = R^X \bowtie R^Y$.

	We have $f_{1} \odot (f_{2} \oplus f_{3}) \sqsubseteq f_R$ and $f_R$ with empty domain. Hence, there exists $h \in \MP$ such that
					 \[f_R = \left( f_{1} \odot (f_{2} \oplus f_{3})\right) \odot h.\]
Thus, $f_1 \odot (f_2 \oplus f_3) = \pi_{X \cup Y} f_R$, and so $f_1 \odot (f_2 \oplus f_3) (\empmem) = R^{X \cup Y}$.

Similarly to the reasoning in~\cref{Md:probhelper}, by~\cref{oplus2odot}, we have
\begin{align*}
	f_1 \odot f_2 \sqsubseteq f_1 \odot (f_2 \oplus f_3) \\
f_1 \odot f_3 \sqsubseteq	f_1 \odot (f_2 \oplus f_3)
\end{align*}
Then, as above,  $f_1 \odot f_2 = \pi_{X} f_R$, $f_{1} \odot f_3= \pi_{Y} (f_R)$. So, $f_1 \odot f_2 (\empmem) = R^{X}$, $f_1 \odot f_3 (\empmem) = R^{Y}$.

By definition of $\oplus$ and $\odot$,
					 \begin{align*}
					 f_{1} \odot ( f_{2} \oplus f_{3}) (\empmem)
&= \{ u \bowtie v \mid u \in f_{1}( \empmem) \textbf{ and } v \in f_{2} \oplus f_{3}( u )\} \\
&= \{ u \bowtie v \mid u \in f_{1}( \empmem) \textbf{ and } v \in \{v_1 \bowtie v_2 \mid v_1 \in f_{2}(u ) \textbf{ and } v_2 \in f_{3}( u )\} \} \\
&=\{ u \bowtie (v_1 \bowtie v_2) \mid u \in f_{1}( \empmem) \textbf{ and } v_1 \in f_{2}(u ) \textbf{ and } v_2 \in f_{3}(u ) \}
					 \end{align*}
					 Since $\bowtie$ is idempotent, i.e., $u \bowtie u = u$, commutative and associative, we have
					 \begin{align*}
	&u \bowtie (v \bowtie w) =  (u \bowtie u) \bowtie (v \bowtie w)
	= (u \bowtie v) \bowtie (u \bowtie w).
					 \end{align*}
Therefore, we can convert the previous equality into
					 \begin{align*}
	&f_{1} \odot ( f_{2} \oplus f_{3}) (\empmem)  = \{ (u \bowtie v_1) \bowtie ( u \bowtie v_2) \mid u \in f_{1}( \empmem) \textbf{ and } v_1 \in f_{2}(u ) \textbf{ and } v_2 \in f_{3}(u ) \} \\
							=& \big\{ u \bowtie v_1 \mid u \in f_{1}( \empmem) \textbf{ and } v_1 \in f_{2}(u )  \big\} \bowtie
				 \big\{ u \bowtie v_2 \mid u \in f_{1}( \empmem) \textbf{ and } v_2 \in f_{3}(u )  \big\}\\
							=& (f_1 \odot f_2) (\empmem) \bowtie (f_1 \odot f_3) (\empmem)
					 \end{align*}

Thus, $R^{X \cup Y} = R^X \bowtie R^Y$.

This completes the proof for the first direction.

\noindent\textbf{Backward direction: }
If $R^{X \cup Y} = R^X \bowtie R^Y$, then we want to show that there exist
$f_1\colon \R{\emptyset} \rightarrow \PR{X \cap Y}$,
	$f_2 \colon \R{X \cap Y} \rightarrow \PR{X}$,
	$f_3 \colon \R{X \cap Y} \rightarrow \PR{Y}$, such that $f_1 \odot (f_2 \oplus f_3) \sqsubseteq f_R$.

	Let $f_{1} = f_R^{X \cap Y}$ and define $f_{2}: \R{{X \cap Y}} \rightarrow \PR{X }$ by having
 \[f_{2}(s) := \{ r \in R^X \mid r^{{X \cap Y}} = s \}\]
 for all $s \in \R{{X \cap Y}}$.
 Define $f_3: \R{{X \cap Y}} \rightarrow \PR{Y }$ by having
 \[f_3(s) = \{ r \in R^Y \mid r^{{X \cap Y}} =s \}\]
 for all $s \in \R{{X \cap Y}}$.

	\begin{itemize}
					\item
	By construction, $f_1, f_2, f_3$ have the desired types.
				 \item States $f_{2}, f_{3}$ are both in $\MP$.

$f_2$ preserves the input because for any $s \in \R{X \cap Y}$, $f_2(s)$ as a relation only includes tuples whose projection to $X \cap Y$ is equals to $s$. Thus, $f_2$ is in $\MP$.

Similarly, $f_3$ is in $\MP$.

					 \item  $f_1 \odot (f_2 \oplus f_3) \sqsubseteq f_R$.

 First, by their types, $f_{1} \odot (f_{2} \oplus f_3) $ is defined, and
	\begin{align}
	 f_{1} \odot (f_{2} \oplus f_3) (\empmem)  &=
	 \{ u \bowtie v \mid u \in f_{1}( \empmem ) \textbf{ and } v \in (f_{2} \oplus f_3)( u )\} \\
	     &=\{ u \bowtie v \mid u \in f_{1}(\empmem) \textbf{ and } v \in f_{2}(u^{\dm{2}}) \bowtie f_3(u^{\dm{3}})\} \notag \\
&=\{ u \bowtie v \mid u \in f_{1}(\empmem) \textbf{ and } v \in f_{2}(u) \bowtie f_3(u)\} \tag{By $\dm{2} =\dm{3} = {X \cap Y}$ } \\
	     &= \{ u \bowtie (v_i \bowtie v_j) \mid u \in f_{1}(\empmem) \textbf{ and } v_i \in f_{2}(u) \textbf{ and } v_j \in f_3( u) \} \notag \\
       &=\{ (u \bowtie v_i)  \bowtie ( u \bowtie v_j) \mid u \in f_{1}(\empmem) \textbf{ and } v_i \in f_{2}(u) \textbf{ and } v_j \in f_3( u) \} \tag{because $\bowtie$ is idempotent, associative, commutative} \\
        &=\{ u \bowtie v_i  \mid u \in f_{1}(\empmem) \textbf{ and } v_i \in f_{2}(u) \}
					       \bowtie
       \{ u \bowtie v_j  \mid u \in f_{1}(\empmem) \textbf{ and } v_j \in f_3(u) \}  \label{eq:relation1}
 \end{align}
	Recall that we define $f_1$ such that $f_{1}(\empmem) = R^{X \cap Y}$, and $f_{2}(s) := \{ r \in R \mid r^{{X \cap Y}} = s \} $, so
 \begin{align}
	 \{ u \bowtie v_i  \mid u \in R^{X \cap Y} \textbf{ and } v_i \in f_{2}(u) \}
	&= \{ (u \bowtie v_i)  \mid u \in R^{X \cap Y} \textbf{ and }
	v_i \in \{ r \in R^X \mid r^{X \cap Y} = u \} \} \notag \\
	&= \{ v_i \mid v_i \in \{ r \in R^X \mid r^{X \cap Y} \in R^{X \cap Y} \} \} \notag \\
	&= R^X \label{eq:relation2}
 \end{align}$f_1 \odot (f_2 \oplus f_3) \sqsubseteq f_{\mu}$
 Analogously,
 \begin{align}
	 \{ u \bowtie v_j  \mid u \in f_{1}(\empmem) \textbf{ and } v_j \in f_3(u) \}
	&= R^Y \label{eq:relation3}
 \end{align}
 Substituting \cref{eq:relation2} and \cref{eq:relation3} into \cref{eq:relation1},  we have
 \begin{align*}
	 f_{1} \odot (f_{2} \oplus f_3) (\empmem) = R^X \bowtie R^Y
 \end{align*}
	By assumption, $R^X \bowtie R^Y =  R^{X \cup Y}$.
	Thus,
	$f_{1} \odot (f_{2} \oplus f_3)(\empmem) =   R^{X \cup Y} $,
	and $f_{1} \odot (f_{2} \oplus f_3) = \pi_{X \cup Y} f_R$.
	By~\cref{Mp:lemma:condfirst}, this implies that $f_1 \odot (f_2 \oplus f_3) \sqsubseteq f_R$.
				\end{itemize}
				 Thus, the constructed $f_1, f_2, f_3$ satisfy all requirements.
			 \end{proof}

		\subsection{Section \ref{sec:graphoid}, graphoid axioms: Omitted Details}
\label{app:separoid}
\begin{lemma}
	\label{graphoid:symmetry}
	The following judgment is derivable in \LOGIC:
              \[
                \vdash P \depand (Q \sepand R)\rightarrow P \depand (R \sepand Q) .
              \]
\end{lemma}
            \begin{proof}
              We have the derivation:
              \begin{prooftree}
                \AxiomC{~}
                \RightLabel{\textsc{Ax}}
                \UnaryInfC{$P \vdash P$}
                \AxiomC{~}
                \RightLabel{$\sepand$-\textsc{Comm}}
                \UnaryInfC{$Q \sepand R \vdash R \sepand Q$}
                \RightLabel{$\depand$-\text{Conj}}
                \BinaryInfC{$ P \depand (Q \sepand R)\vdash P \depand (R \sepand Q)$}
                \RightLabel{$\rightarrow$}
                \UnaryInfC{$\vdash P \depand (Q \sepand R)\rightarrow P \depand (R \sepand Q)$}
              \end{prooftree}
            \end{proof}

\begin{lemma}
	\label{graphoid:decomposition}
	The following judgment is derivable in \LOGIC:
              \[
                \vdash P \depand (Q \sepand (R \land S))\rightarrow   P \depand (Q \sepand R)  \land  P \depand (Q \sepand S) .
              \]
\end{lemma}
            \begin{proof}
              We have the derivation:
              \begin{prooftree}
                \AxiomC{~}
                \LeftLabel{\textsc{Ax}}
                \UnaryInfC{$P \vdash P$}
                \AxiomC{~}
                \LeftLabel{\textsc{Ax}}
                \UnaryInfC{$Q \vdash Q$}
                \AxiomC{~}
                \LeftLabel{\textsc{Ax}}
                \UnaryInfC{$R \land S \vdash R \land S$}
                \RightLabel{$\land 3$}
                \UnaryInfC{$R \land S \vdash R$}
                \RightLabel{$\sepand$-\textsc{Conj}}
                \BinaryInfC{$Q \sepand (R \land S) \vdash Q \sepand R$}
                \RightLabel{ $\depand$-\textsc{Conj}}
                \BinaryInfC{$P \depand (Q \sepand (R \land S)) \vdash P \depand (Q \sepand R) $}
                \AxiomC{Similar to left}
                \UnaryInfC{$P \depand (Q \sepand (R \land S)) \vdash P \depand (Q \sepand S) $}
                \LeftLabel{$\land 1$}
                \BinaryInfC{$ P \depand (Q \sepand (R \land S)) \vdash P \depand (Q \sepand R)  \land  P \depand (Q \sepand S)$}
                \LeftLabel{$\rightarrow$}
                \UnaryInfC{$\vdash P \depand (Q \sepand (R \land S))\rightarrow   P \depand (Q \sepand R)  \land  P \depand (Q \sepand S)$}
              \end{prooftree}
            \end{proof}

            \begin{lemma}[Weak Union]
              \label{graphoid:weaku}
														The following judgment is valid in any $\T$-model where Disintegration holds (see ~\cref{lemma:condfirst} and~\cref{Mp:lemma:condfirst} for Disintegration):
              \begin{align*}
                \models \indep{Z}{X}{Y \cup W} \rightarrow \indep{Z \cup W}{X}{Y}
              \end{align*}
            \end{lemma}
            \begin{proof}
              Let $M$ be a $\T$-model. If $f \models \indep{Z}{X}{Y \cup W}$,
              by~\cref{findexact}, there exist $f_1, f_2, f_3 \in M$ such that
              $f_1 \odot (f_2 \oplus f_3) \sqsubseteq f$, $f_1 : \m{\emptyset}
              \rightarrow \TR{Z}$, $f_2 : \m{Z} \rightarrow \TR{Z \cup X}$, $f_3
              : \m{Z} \rightarrow \TR{Z \cup Y \cup W}$.

              Let $f_3^1 = \pi_{Z \cup W} f_3$, then by Disintegration there exists $f_3^2 \in M$ such that $f_3 = f_3^1 \odot f_3^2$.

              Since $f_1 \odot (f_2 \oplus f_3) \sqsubseteq f$, and $f$ has empty domain, there must exists $v \in M$ such that
              \begin{align*}
                f &= f_1 \odot (f_2 \oplus f_3) \odot v \\
                  & = f_1 \odot f_3 \odot (\unit_{Z \cup Y \cup W} \oplus f_2) \odot v \tag{By~\cref{oplus2odot}} \\
																		& = f_1 \odot f_3 \odot (\unit_{Y \cup W} \oplus f_2) \odot v \tag{By $\dom(f_2) = Z$} \\
                  & = f_1 \odot (f_3^1 \odot f_3^2) \odot (\unit_{Y \cup W} \oplus f_2) \odot v  \\
                  & = f_1 \odot f_3^1 \odot (f_3^2 \odot (\unit_{Y \cup W} \oplus f_2)) \odot v  \\
                  & = f_1 \odot f_3^1 \odot ((f_2 \oplus \unit_W)  \oplus f_3^2) \odot v \tag{$\dagger$}
              \end{align*}
														where $\dagger$ follows from \cref{odot2oplus} and $\dom(f_2 \oplus \unit_W) = Z  \cup W \subseteq \range(f_3^1)$.

														Thus, $ f_1 \odot f_3^1 \odot ((f_2 \oplus \unit_W) \oplus f_3^2)  \sqsubseteq f$.

              Note that	$f_1 \odot f_3^1$ has type $\m{\emptyset} \rightarrow \T{\m{Z \cup W}}$, so $f_1 \odot f_3^1 \models \pair{\emptyset}{Z \cup W}$.

              State $f_2 \oplus \unit_W$ has type $\m{Z \cup W} \rightarrow \T(\m{Z \cup W \cup X}$, so $f_2 \oplus \unit_W \models \pair{Z \cup W}{X}$.

              State $f_3^2 $ has type $\m{Z \cup W} \rightarrow \TR{Z \cup W \cup Y}$, so $f_3^2\models \pair{Z \cup W}{Y}$.

              Therefore, $ f_1 \odot f_3^1 \odot ((f_2 \oplus \unit_W)  \oplus f_3^2) \models \pair{\emptyset}{Z \cup W} \depand \pair{Z \cup W}{X} \sepand \pair{Z \cup W}{Y}$.

              By persistence, $f \models \indep{Z \cup W}{X}{Y} $, and Weak Union is valid.
            \end{proof}

    \begin{lemma}[Contraction]
      \label{graphoid:contraction}
      The following judgment is valid in any $\T$-model:
      \begin{align*}
        \models (\indep{Z}{X}{Y}) \land (\indep{Z \cup Y}{X}{W}) \rightarrow \indep{Z}{X}{Y \cup W}
      \end{align*}
    \end{lemma}
		\begin{proof}
      Let $M$ be a $\T$-model.
      If $h \models  (\indep{Z}{X}{Y} ) \land ( \indep{Z \cup Y}{X}{W} )$, then
      \begin{itemize}
        \item $h \models \indep{Z}{X}{Y}$.
          By \cref{findexact}, there exists $f_1, f_2, f_3$ such that $f_1: \m{\emptyset} \rightarrow \TR{Z}$, $f_2 : \m{Z} \rightarrow \TR{Z \cup X}$, $f_3 : \m{Z} \rightarrow \TR{Z \cup Y}$, and $f_1 \odot (f_2 \oplus f_3) \sqsubseteq h$.

          Note $f_1 \odot (f_2 \oplus f_3)$ has type $\m{\emptyset} \rightarrow \TR{Z \cup Y \cup Z}$.

        \item $ h \models  \indep{Z \cup Y}{X}{W}$.
          By \cref{findexact}, there exists $g_1, g_2, g_3$ such that $g_1: \m{\emptyset} \rightarrow \TR{Z \cup Y}$, $g_2 : \m{Z \cup Y} \rightarrow \TR{Z \cup Y \cup X}$, $g_3 : \m{Z \cup Y} \rightarrow \TR{Z \cup Y \cup W}$, and $g_1 \odot (g_2 \oplus g_3) \sqsubseteq h$.

          Note $g_1 \odot g_2$ has type $\m{\emptyset} \rightarrow \TR{Z \cup Y \cup X}$.
      \end{itemize}
      By~\cref{uniqueness}, $f_1 \odot (f_2 \oplus f_3) = g_1 \odot g_2$.
      \begin{align*}
        g_1 \odot (g_2 \oplus g_3)
                        &= g_1 \odot (g_2 \oplus \unit_{Z \cup Y}) \odot (\unit_{Z \cup Y \cup X} \oplus g_3) \tag{By!\cref{oplus2odot}}\\
																								&= g_1 \odot g_2 \odot  (\unit_{Z \cup X} \oplus g_3)  \tag{Because $Z \cup Y \subseteq \dom(g_2),\, Y \subseteq \dom(g_3)$ } \\
                        &= f_1 \odot (f_2 \oplus f_3) \odot (\unit_{Z \cup X} \oplus g_3) \tag{$ f_1 \odot (f_2 \oplus f_3) = g_1 \odot g_2$}\\
                        &= f_1 \odot \big( (f_2 \odot \unit_{Z \cup X} ) \oplus (f_3 \odot  g_3) \big) \tag{By \ref{revexeq}}\\
                        &= f_1 \odot  \big(f_2  \oplus (f_3 \odot  g_3) \big)
      \end{align*}
      By their types, it is easy to see that $f_1 \models \pair{\emptyset}{Z}$, $f_2 \models \pair{Z}{X}$, $f_3 \odot g_3 \models \pair{Z}{Y \cup W}$. So,
      \[f_1 \odot (f_2 \oplus (f_3 \odot g_3)) \models \indep{Z}{X}{Y \cup W}. \]
      Also, note that	$h \sqsupseteq  g_1 \odot (g_2 \oplus g_3) = f_1 \odot (f_2 \oplus (f_3 \odot g_3))$, so by persistence,
      \[h \models \pair{\emptyset}{Z} \depand (\pair{Z}{X} \sepand \pair{Z}{Y \cup W}).\qedhere\]
		\end{proof}

    \subsection{Section \ref{sec:cpsl}, Conditional Probabilistic Separation Logic}
    \label{app:cpsl-full}

    As our final application, we design a separation logic for probabilistic
    programs. We work with a simplified probabilistic imperative language with
    assignments, sampling, sequencing, and conditionals; our goal is to show how
    a \LOGIC-based program logic could work in the simplest setting. Following
    the design of PSL~\citep{barthe2019probabilistic}, a richer program logic
    could also layer on constructs for deterministic assignment and
    deterministic control flow (conditionals and loops) at the cost of
    increasing the complexity of the programming language and semantics. We do
    not foresee difficulties in implementing these extensions, and we leave them
    for future work.

		\subsection{A basic probabilistic programming language}

		\paragraph*{Program syntax}
		Let $\Var$ be a fixed, finite set of program variables. We will consider the
		following programming language:
		\begin{align*}
			\Exp \ni e &::= x \in \Var \mid \ktt \mid \kff \mid e \land e' \mid e \lor e'\mid \cdots
			\\
			\Com \ni c &::= \Skip
			\mid \Assn{x}{e}
			\mid \Rand{x}{\Bern_{p}} \quad (p \in [0, 1]) \\
      & \qquad
      \mid \Seq{c}{c'}
			\mid \Cond{x}{c}{c'}
		\end{align*}
    We assume that all variables and expressions are Boolean-valued, for
    simplicity. The only probabilistic command is $\Rand{x}{\Bern_{p}}$, which
    draws from a $p$-biased coin flip (i.e., probability of $\ktt$ is $p$) and
    stores the result in $x$; for instance, $\Rand{x}{\Bern_{1/2}}$
    samples from a fair coin flip. 

		\begin{figure*}[!t]
      \hrule
      \vspace{1mm}
			\begin{align*}
				\denot{\Assn{x}{e}}\mu
				&\defeq \dbind(\mu, m \mapsto \dunit(m[x \mapsto \denot{e}(m)])) \\
				\denot{\Rand{x}{\Bern_p}}\mu
				&\defeq \dbind(\mu, m \mapsto \dbind(\text{Bern}_p, v \mapsto \dunit(m[x \mapsto v]))) \\
				\denot{\Seq{c}{c'}}\mu
				&\defeq \denot{c'}(\denot{c}\mu) \\
				\denot{\Cond{b}{c}{c'}}\mu
				&\defeq
				\dconv{p}
				{(\denot{c} \dcond{\mu}{\denot{b = \ktt}})}
				{(\denot{c'} \dcond{\mu}{\denot{b = \kff}})}
				\qquad \text{where } p \defeq \mu(\denot{b = \ktt})
			\end{align*}
	\vspace{-6mm}
			\caption{Program semantics}
      \vspace{1mm}
      \hrule	\vspace{-5mm}
			\label{fig:semantics}
		\end{figure*}

		\paragraph*{Program semantics}
    Following~\citet{kozen81}, we give programs a denotational semantics as
    \emph{distribution transformers} $\denot{c} : \DR{\Var} \to \DR{\Var}$, see \Cref{fig:semantics}. To define the semantics of
    randomized conditionals, we will use operations for conditioning to split
    control flow, and convex combinations to merge control flow. More formally,
    let $\mu \in \DD(A)$ be a distribution, let $S \subseteq A$ be an event, and
    let $\mu(S)$ be the probability of $S$ in $\mu$. Then the conditional
    distribution of $\mu$ given $S$ is:
		\[
			(\dcond{\mu}{S})(a) \defeq \begin{cases}
				\frac{\mu(a)}{\mu(S)} &: a \in S, \mu(S)\neq 0 \\
				0 &: a \notin S .
			\end{cases}
		\]

		For convex combination, let $p \in [0, 1]$ and $\mu_1, \mu_2 \in
		\DD(A)$. We define:
		\[
			(\dconv{p}{\mu_1}{\mu_2})(a) \defeq p \cdot \mu_1(a) + (1 - p) \cdot \mu_2(a) .
		\]
		When $p = 0$ or $p = 1$, we define $\oplus_p$ lazily:
		$\dconv{0}{\mu_1}{\mu_2} \defeq \mu_2$ and $\dconv{1}{\mu_1}{\mu_2}
		\defeq \mu_1$. 
		 Conditioning
		and convex combination are inverses in the following sense: $\mu =
		\dconv{\mu(S)}{(\dcond{\mu}{S})}{(\dcond{\mu}{\overline{S}})}$.

	\begin{example}
    \Cref{fig:examples_ap} introduces two more example programs. The program \ExONE
    (\Cref{fig:prog-1_ap}) generates a distribution where two random observations
    share a common cause. Specifically, $z$, $x$, and $y$ are independent random
    samples, and $a$ and $b$ are values computed from $(x, z)$ and $(y, z)$,
    respectively. Intuitively, $z$, $x$, $y$ could represent independent noisy
    measurements, while $a$ and $b$ could represent quantities derived from
    these measurements. Since $a$ and $b$ share a common source of randomness
    $z$, they are not independent. However, $a$ and $b$ are independent
    conditioned on the value; this is a textbook example of conditional
    independence.

    The program \ExTWO (\Cref{fig:prog-2_ap}) is a bit more complex: it branches on
    a random value $z$, and then assigns $x$ and $y$ with two independent
    samples from $\Bern_p$ in the true branch, and $\Bern_q$ in the false
    branch. While we might think that $x$ and $y$ are independent at the end of
    the program since they are independent at the end of each branch, this is
    not true because their distributions are different in the two branches. For
    example, suppose that $p = 1$ and $q = 0$. Then at the end of the first
    branch $(x, y) = (\ktt, \ktt)$ with probability $1$, while at the end of the
    second branch $(x, y) = (\kff, \kff)$ with probability $1$. Thus observing
    whether $x = \ktt$ or $x = \kff$ determines the value of $y$---clearly, $x$
    and $y$ can't be independent. However, $x$ and $y$ \emph{are} independent
    conditioned on $z$. Verifying this example relies on the proof rule for
    conditionals.
    \end{example}

				\begin{figure}[!t]\small
      \hrule
			\begin{subfigure}[b]{0.48\linewidth}
				\begin{center}
					\[
\begin{array}{l}
	\Rand{z}{\Bern_{1/2}}; \\
	\Rand{x}{\Bern_{1/2}}; \\
	\Rand{y}{\Bern_{1/2}}; \\
	\Assn{a}{x \lor z}; \\
	\Assn{b}{y \lor z}
\end{array}
					\]
				\end{center}
				\caption{\ExONE}
				\label{fig:prog-1_ap}
			\end{subfigure}
			\begin{subfigure}[b]{0.48\linewidth}
				\begin{center}
					\[
\begin{array}{l}
	\Rand{z}{\Bern_{1/2}}; \\
	\Condt{z}{} \\
	\qquad\Rand{x}{\Bern_{p}};
	\Rand{y}{\Bern_{p}} \\
	\mathbf{else} \\
	\qquad\Rand{x}{\Bern_{q}};
	 \Rand{y}{\Bern_{q}}
\end{array}
					\]
				\end{center}
				\caption{\ExTWO}
				\label{fig:prog-2_ap}
			\end{subfigure}
			\caption{Example programs}
			\label{fig:examples_ap}
      \vspace{1mm}
      \hrule	\vspace{-5mm}
		\end{figure}

    \subsection{\SYSTEM: Assertion Logic}
    \label{sec:cpsl-assertions}

    Like all program logics, \SYSTEM is constructed in two layers: the
    \emph{assertion logic} describes program states---here, probability
    distributions---while the \emph{program logic} describes probabilistic
    programs, using the assertion logic to specify pre- and post-conditions.
        Our starting point for the assertion logic is the probabilistic model of
    \LOGIC introduced in \Cref{sec:models}, with atomic assertions as in
    \Cref{sec:CI}. 
    However, it turns out that the full logic \LOGIC is not
    suitable for a program logic. The main problem is that not all formulas in
    \LOGIC satisfy a key technical condition, known as \emph{restriction}.

    \begin{definition}[Restriction]
      A formula $P$ satisfies \emph{restriction} if: a
      Markov kernel $f$ satisfies $P$ if and only if there exists $f'
      \sqsubseteq f$ such that $\range(f') \subseteq \FV(P)$ and $f' \models P$.
    \end{definition}

    The reverse direction is immediate by persistence, but the forward direction
    is more delicate. Restriction was first considered
    by~\citet{barthe2019probabilistic} while developing PSL: formulas satisfying restriction are
    preserved if the program does not modify variables appearing in the formula.
    This technical property is crucial to supporting Frame-like rules in
    PSL, which are also used to derive general versions of rules
    for assignment and sampling, so failure of the restriction property imposes
    severe limitations on the program logic. In PSL, assertions were drawn from BI with atomic
    formulas for modeling random variables. Using properties specific to
    probability distributions, they showed that their logic is well-behaved with
    respect to restriction: \emph{all formulas} satisfy this property. However,
    \LOGIC is richer than BI, and there are simple
    formulas where restriction fails.

    \begin{example}[Failure of restriction]
      Consider the formula $P \defeq \top \depand \kp{\Exact{x}}{\Dist[x]}$, and
      consider the kernel $f : \R{z} \to \DR{x, z}$ with $f(z \mapsto c) \defeq
      \dunit(x \mapsto c, z \mapsto c)$. Letting $f_1 : \R{z} \to \DR{x, z}$ and
      $f_2 \colon \R{x, z} \to \DR{x, z}$ with $f_1(z \mapsto c) \defeq \dunit(x
      \mapsto c, z \mapsto c) \models \top$ and $f_2 \defeq \unit_{\R{x}} \oplus
      \unit_{\R{z}} \models \kp{\Exact{x}}{\Dist[x]}$, we have $f = f_1 \odot f_2
      \models P$.  Any subkernel $f' \sqsubseteq f$ satisfying $P$ and
      witnessing restriction must be of type $f' : \R{x} \to \DR{x}$,
      but it is not hard to check that there is no such subkernel.
    \end{example}

    To address this problem, we will identify a fragment of \LOGIC that
    satisfies restriction and is sufficiently rich to support an interesting program logic. Intuitively,
    restriction may fail for $P$ when a kernel satisfying $P$ (i) implicitly requires
    unexpected variables in its domain, or (ii) does not describe needed variables
    in its range. Thus, we employ syntactic
    conditions to approximate which variables \emph{may} appear in the domain
    ($\FFV$), and which variables \emph{must} appear in the range ($\SFV$).

    \begin{definition}[$\FFV$ and $\SFV$]
      For the formulas in $\Frestrict$ generated by probabilistic
      atomic propositions, conjunctions ($\land$, $\sepand$, $\depand$) and
      disjunction ($\lor$), we define two sets of variables:
      \begin{align*}\small
        \FFV(\top) = \FFV(\bot) &\defeq \emptyset & \SFV(\top) = \SFV(\bot) &\defeq \emptyset \\
        \FFV\kp{A}{B} &\defeq \FV(A) & \SFV\kp{A}{B} &\defeq \FV(A) \cup \FV(B) \\
        \FFV(P \land Q) &\defeq \FFV(P) \cup \FFV(Q) & \SFV(P \land Q) &\defeq \SFV(P) \cup \SFV(Q) \\
        \FFV(P \sepand Q) &\defeq \FFV(P) \cup \FFV(Q) & \SFV(P \sepand Q) &\defeq \SFV(P) \cup \SFV(Q) \\
        \FFV(P \depand Q) &\defeq \FFV(P) \cup \FFV(Q) & \SFV(P \depand Q) &\defeq \SFV(P) \cup \SFV(Q) \\
        \FFV(P \lor Q) &\defeq \FFV(P) \cup \FFV(Q) & \SFV(P \lor Q) &\defeq \SFV(P) \cap \SFV(Q)
      \end{align*}
    \end{definition}

    Now, we have all the ingredients to introduce our assertions. The logic
    \RLOGIC is a fragment of \LOGIC with atomic propositions $\mathcal{AP}$,
    with formulas $\Frestrict$ defined by the following grammar:
    \begin{align*}
      P, Q &::= \mathcal{AP}
            \mid \top
            \mid \bot
            \mid P \lor Q
            \mid P \sepand Q \\
					&\mid P \depand Q \quad (\FFV(Q) \subseteq \SFV(P)) \\
           &\mid P \land Q \quad (\SFV(P) = \SFV(Q) = \FV(P) = \FV(Q)) .
    \end{align*}
    The side-condition for $P \depand Q$ ensures that variables used by
    $Q$ are described by $P$.
				The side-condition for
    $P \land Q$ is the most restrictive---to understand why we need it, consider
    the following example.

    \begin{example}[Failure of restriction for $\land$]
      Consider the formula $P \defeq \kp{\Exact{\emptyset}}{\Dist[x]} \land
      \kp{\Exact{\emptyset}}{\Dist[y]}$,
and kernel $f : \R{z} \to
      \DR{x, y, z}$ with $f(z \mapsto \ktt)$ being the distribution with $x$ a
      fair coin flip, $y = x$, and $z = \ktt$,
and $f(z \mapsto \kff)$ being the
      distribution with $x$ a fair coin flip, $y = \neg x$, and $z = \kff$.
Then, there exist
$f_1 : \R{\emptyset} \to \DR{x}$ and $f_2 : \R{\emptyset} \to \DR{y}$
such that $f_1 \sqsubseteq f$ and $f_2 \sqsubseteq f$.
Since $f_1 \models
      \kp{\Exact{\emptyset}}{\Dist[x]}$ and $f_2 \models
\kp{\Exact{\emptyset}}{\Dist[y]}$, it follows $f \models P$. But, because $z$ is correlated with $(x,y)$,
there is no kernel $f' : \R{\emptyset} \to \DR{x, y}$ satisfying $P$ such that $f' \sqsubseteq f$.
    \end{example}

    When we take atomic propositions from \Cref{sec:CI}, formulas are pairs of
    sets of variables: $\pair{A}{B}$ where $A, B \subseteq \Var$. With these
    atoms, all formulas in \RLOGIC satisfy restriction. Before showing this
    property, however, we will enrich the atomic propositions
    to describe more fine-grained information about the domain and range of
    kernels:
    \begin{description}
      \item[Domain.] Given a kernel $f$, the existing atomic propositions can
        only describe properties that hold for all (well-typed) inputs $m$ to
        $f$. We would like to be able to describe properties that hold for only
        certain inputs, e.g., for memories $m$ where a variable $z$ is true.
      \item[Range.] Given any input $m$ to a kernel $f$, the existing atomic
        propositions can only guarantee the presence of variables in the output
        distribution $f(m)$. We would like describe more precise information
        about $f(m)$, e.g., that certain variables are independent conditioned
        on a \emph{particular} value of $m$, rather on all values of $m$.
    \end{description}
    Our strategy will be to extend atomic propositions to all pairs of logical
    formula $\kp{D}{R}$, where $D$ is a logical formula over the kernel domain
    (i.e., memories), while $R$ is a logical formula over the kernel range
    (i.e., distributions over memories).

    To describe memories, we take a simple propositional logic for the domain
    logic.
    \begin{definition}[Domain logic]
      The \emph{domain logic} has formulas $D$ of the form $S : p_d$, where $S
      \subseteq \Var$ is a subset of variables and:
      \(
     p_d ::= x = e
        \mid \top
        \mid \bot
        \mid p_d \land p_d'
        \mid p_d \lor p_d' .
      \)
      A formula $S : p_d$ is \emph{satisfied} in $m \in \R{T}$, written $m
      \models_d S : p_d$, if $S = T$ and $p_d$ holds in $m$.
    \end{definition}
    We can read $S : p_d$ as ``memories over $S$ such that $p_d$'' and abbreviate $S : \top$ as just $S$. To describe distributions over
    memories, we adapt probabilistic
    BI~\citep{barthe2019probabilistic} for the range logic.
    \begin{definition}[Range logic]
      The \emph{range logic} has the following formulas from probabilistic BI:
      \[
        p_r ::= \Dist[S] \quad (S \subseteq \Var)
        \mid x \sim d
        \mid x = e
        \mid \top
        \mid \bot
        \mid p_r \land p_r'
        \mid p_r \sepand p_r' .
      \]
    \end{definition}
    We give a semantics where states are distributions over memories: $M_r = \{
    \mu : \DR{S} \mid S \subseteq \Var \}$. We define a preorder on states via
    $\mu_1 \sqsubseteq_r \mu_2$ if and only if $dom(\mu_1) \subseteq
    dom(\mu_2)$ and $\pi_{dom(\mu_1)}\mu_2 = \mu_1$, and we define a partial
    binary operation on states: if $dom(\mu_1) = S_1 \cup T$ and $dom(\mu_2) =
    S_2 \cup T$ with $S_1, S_2, T$ disjoint, and $\pi_T\mu_1 = \pi_T\mu_2 =
    \dunit(m)$ for some $m \in \R{T}$, then
    \[
      \mu_1 \oplus_r \mu_2 \defeq \pi_{S_1}\mu_1 \otimes \dunit(m) \otimes \pi_{S_2}\mu_2
    \]
    where $\otimes$ takes the independent product of two distributions over
    disjoint domains; otherwise $\oplus_r$ is not defined. This operation
    generalizes the monoid from probabilistic BI to allow combining
    distributions with overlapping domains if the distributions over the
    overlap are deterministic and equal; this mild generalization is useful
    for our setting, where distributions often have deterministic variables
    (e.g., variables corresponding to the input of kernels).

    Then, we define the semantics of the range logic as:
    \[
    \begin{array}{@{}ll}
       \mu \models_r \top &\phantom{iff}\text{always} \qquad\qquad \mu \models_r \bot \qquad\text{ never} \\
        \mu \models_r [S] &\text{ iff } S \subseteq dom(\mu) \\
        \mu \models_r x \sim d &\text{ iff } x \in dom(\mu)
        \text{ and } \pi_x \mu = \denot{d}m_v,
        \text{ where }  \dunit(m_v) = \pi_{\FV(d)} \mu \\
        \mu \models_r x = e &\text{ iff } \{ x \}, \FV(e) \subseteq dom(\mu) \text{ and } \mu(\denot{x = e}) = 1 \\
        \mu \models_r p_r \land p_r' &\text{ iff } \mu \models_r p_r \text{ and } \mu \models_r p_r' \\
        \mu \models_r p_r \sepand p_r' &\text{ iff there exists }
        \mu_1 \oplus_r \mu_2 \sqsubseteq \mu \text{ with }
        \mu_1 \models_r p_r \text{ and } \mu_2 \models_r p_r' .
    \end{array}
\]
    Now, we can give a semantics to our enriched atomic propositions.

    \begin{definition}
      Given a kernel $f$ and atomic proposition $\kp{D}{R}$, we define a
      persistent semantics:
      \[
        f \models \kp{D}{R} \text{ iff there exists } f' \sqsubseteq f
        \text{ such that } m \models_d D \text{ implies }
        m \in dom(f') \text{ and } f(m) \models_r R .
      \]
    \end{definition}

    Atomic propositions satisfy the following axiom schemas, inspired by Hoare
    logic.

    \begin{restatable}{proposition}{axschemasound} \label{prop:ax:cpsl}
      The following axiom schemas for atomic propositions are sound.
      \begin{align}
        & \kp{S : p_d}{p_r} \land \kp{S : p_d'}{p_r'} \rightarrow \kp{S : p_d \land p_d'}{p_r \land p_r'}
        \qquad \text{if } \FV(p_r) = \FV(p_r')
        \tag{\textsc{AP-And}} \label{ax:and} \\
        & \kp{S : p_d}{p_r} \land \kp{S : p_d'}{p_r'} \rightarrow \kp{S : p_d \lor p_d'}{p_r \lor p_r'}
        \tag{\textsc{AP-Or}} \label{ax:or} \\
        & \kp{S : p_d}{p_r} \sepand \kp{S' : p_d'}{p_r'} \rightarrow \kp{S \cup S': p_d \land p_d'}{p_r \sepand p_r'}
        \tag{\textsc{AP-Par}} \label{ax:par} \\
        & p_d' \rightarrow p_d
        \text{ and } \models_r p_r \rightarrow p_r'
        \text{ implies } \models \kp{S : p_d}{p_r} \to \kp{S : p_d'}{p_r'}
        \tag{\textsc{AP-Imp}} \label{ax:cons}
      \end{align}
    \end{restatable}

    Finally, formulas in \RLOGIC satisfy restriction.

    \begin{restatable}[Restriction in \RLOGIC]{theorem}{restriction}
      \label{thm:restriction}
      Let $P \in \Frestrict$ with atomic propositions $\kp{D}{R}$,
      as described above. Then $f \models P$ if and only if there exists $f'
      \sqsubseteq f$ such that $range(f') \subseteq \FV(P)$ and $f' \models P$.
    \end{restatable}
    \begin{proof}[Proof sketch.]
      By induction on $P$, proving a stronger statement: $f \models P$ if and
      only if there exists $f' \sqsubseteq f$ such that $dom(f') \subseteq
      \FFV(P)$, and $\SFV(P) \subseteq range(f') \subseteq \FV(P)$.
    \end{proof}

    \subsection{\SYSTEM: program logic}\label{sec:cpsl_logic}

    \begin{figure*}\small
      \hrule
      \vspace{1mm}
      \begin{mathpar}
        \inferrule*[Left=Assn]
        { x \not\in \FV(e) \cup \FV(P) }
        { \vdash \psl{P}{\Assn{x}{e}}{P \depand \kp{\FV(e)}{ x = e }} }
        \and
        \inferrule*[Left=Samp]
        { x \not\in \FV(d) \cup \FV(P) }
        { \vdash \psl{P}{\Rand{x}{d}}{P \depand \kp{\FV(d)}{ x \sim d }} }
        \\
        \inferrule*[Left=Skip]
        {~}
        { \vdash \psl{P}{\Skip}{P} }
        \and
        \inferrule*[Left=Seqn]
        { \vdash \psl{P}{c}{Q} \\ \vdash \psl{Q}{c'}{R} }
        { \vdash \psl{P}{\Seq{c}{c'}}{R} }
        \\
        \inferrule*[Left=DCond]
        { \vdash \psl{\kp{\Exact{\emptyset}}{b = \ktt} \depand P}{c}{\kp{\Exact{\emptyset}}{b = \ktt} \depand \kp{b : b = \ktt}{Q_1}}
        \\\\ \vdash \psl{\kp{\Exact{\emptyset}}{b = \kff} \depand P}{c'}{\kp{\Exact{\emptyset}}{b = \kff} \depand \kp{b : b = \kff}{Q_2}} }
        { \vdash \psl{\kp{\Exact{\emptyset}}{\Dist[b]} \depand P} {\Cond{b}{c}{c'}}
        {\kp{\Exact{\emptyset}}{\Dist[b]} \depand (\kp{b : b = \ktt}{Q_1} \land \kp{b : b = \kff}{Q_2}) } }
        \\
        \inferrule*[Left=Weak]
	{ \vdash \psl{P}{c}{Q} \\\\ \models P' \rightarrow P \land Q \rightarrow Q' }
        { \vdash \psl{P'}{c}{Q'} }
		\and

        \inferrule*[Left=Frame]
        { \vdash \psl{P}{c}{Q} \\
          \FV(R) \cap \MV(c) = \emptyset \\\\
          \FV(Q) \subseteq \SFV(P) \cup \WV(c) \\
        \RV(c) \subseteq \SFV(P) }
        { \vdash \psl{P \sepand R}{c}{Q \sepand R} }
      \end{mathpar}
	\vspace{-6mm}
      \caption{Proof rules: \SYSTEM}
      \label{fig:cpsl-rules}
      \vspace{1mm}
      \hrule 	\vspace{-5mm}
    \end{figure*}

  With the assertion logic set, we are now ready to introduce our program logic.
  Judgments in \SYSTEM have the form $\psl{P}{c}{Q}$, where $c \in \Com$ is a
  probabilistic program and $P, Q \in \mathrm{Form_{\RLOGIC}}$ are restricted
  assertions. As usual, a program in a judgment maps states satisfying the
  pre-condition to states satisfying the post-condition.

  \begin{restatable}[\SYSTEM Validity]{definition}{cpslvalid}
    A \SYSTEM judgment $\psl{P}{c}{Q}$ is \emph{valid}, written $\models
    \psl{P}{c}{Q}$, if for every input distribution $\mu \in \DR{\Var}$ such
    that the lifted input $f_{\mu} : \R{\emptyset} \to \DR{\Var}$
    satisfies $f_{\mu} \models P$, the lifted output satisfies
    $f_{\denot{c} \mu} \models Q$.
  \end{restatable}

  The proof rules of \SYSTEM are presented in \Cref{fig:cpsl-rules}. Note that all rules implicitly require that assertions are from
		\RLOGIC, e.g., the rule \textsc{Assn} requires
  that the post-condition $P \depand \kp{\FV(e)}{x = e}$ is a formula in
  \RLOGIC, which in turn requires that $\FV(e) = \FFV\kp{\FV(e)}{x = e}
  \subseteq \SFV(P)$.

  The rules \textsc{Skip}, \textsc{Seqn}, \textsc{Weak} are standard, we comment on the other, more interesting rules.
		\textsc{Assn} and \textsc{Samp} allow forward reasoning across assignments and random sampling commands. In both
  cases, a pre-condition that does not mention the assigned variable $x$ is
  augmented with new information tracking the value or distribution of $x$, and
  variables $x$ may depend on.

  \textsc{DCond} allows reasoning about probabilistic control flow, and the
  ensuing conditional dependence that may result. The main pre-condition $P$ is
  allowed to depend on the guard variable $b$---recalling that $\FFV(P)
  \subseteq \SFV\kp{\emptyset}{[b]}$---and $P$ is preserved as a pre-condition
  for both branches. The post-conditions allows introducing new facts $\kp{b : b
  = \ktt}{Q_1}$ and $\kp{b : b = \ktt}{Q_2}$, which are then combined in the
  post-condition of the entire conditional command. As in PSL, the rule for
  conditionals does not allow the branches to modify the guard $b$---this
  restriction is needed to accurately associate each post-condition to each
  branch.

  Finally, \textsc{Frame} is the frame rule for \SYSTEM. Much like in PSL, the
  rule involves three classes of variables: $\MV(c)$ is the set of variables
  that $c$ may write to, $\RV(c)$ is the set of variables that $c$ may read from
  the input, and $\WV(c)$ is the set of variables that $c$ must write to; these
  variable sets are defined in \Cref{app:cpsl}. Then, \textsc{Frame} is
  essentially the same as in PSL. The first side-condition $\FV(R) \cap \MV(c)$
  ensures that the framing condition is not modified---this condition is fairly
  standard. The second and third side-conditions are more specialized. First,
  the variables described by $Q$ in the post-condition are either already
  described by $P$ in the pre-condition, or are written by $c$.  Second, the
  variables read by $c$ must be described by $P$ in the pre-condition. These two
  side-conditions ensure that variables mentioned by $Q$ that were not already
  independent of $R$ are freshly written, and freshly written variables are
  derived from variables that were already independent of $R$.

  \begin{restatable}[\SYSTEM Soundness]{theorem}{cpslsound}
  \label{thm:cpslsound}
    \SYSTEM is sound: derivable judgments are valid.
  \end{restatable}
  \begin{proof}[Proof sketch]
    By induction on the proof derivation. The restriction property is used
    repeatedly to constrain the domains and ranges of kernels witnessing
    different sub-assertions, ensuring that pre-conditions about unmodified
    variables continue to hold in the post-condition.
  \end{proof}

  \subsection{Example: proving conditional independence for programs}\label{sec:cpsl_ex}
  Now, we show how to use \SYSTEM to verify our two example programs in
  \Cref{fig:examples_ap}. In both cases, we will prove a conditional independence
  assertion as the post-condition. We will need some axioms for implications
  between formulas in \RLOGIC; these axioms are valid in our probabilistic model
  $\MD$.

		\begin{restatable}{proposition}{RLOGICAx}
			(\textsc{Axioms for \RLOGIC})
    The following axioms are sound, assuming both precedent and antecedent are
    in $\Frestrict$.
    \begin{align}
        &(P \depand Q) \depand R \rightarrow P \depand (Q \sepand R)
        \tag{\textsc{Indep-1}} \label{ax:dep2sep} \\
        & P \depand Q \rightarrow P \sepand Q
        \qquad\qquad\qquad \text{if $\FFV(Q) = \emptyset$}
        \tag{\textsc{Indep-2}} \label{ax:dep2sep2} \\
        & P \depand Q \rightarrow P \depand (Q \sepand \kp{\Exact{S}}{\Dist[S]})
        \tag{\textsc{Pad}} \label{ax:padding} \\
        & (P \sepand Q) \depand (R \sepand S) \rightarrow (P \depand R) \sepand (Q \depand S)
        \tag{\textsc{RestExch}} \label{ax:exch}
    \end{align}
  \end{restatable}
  We briefly explain the axioms. \ref{ax:dep2sep} holds because $P \depand (Q
  \sepand R) \in \Frestrict$ implies that $R$ only mentions variables that are
  guaranteed to be in $P$. \ref{ax:dep2sep2} holds because any kernel witnessing
  $Q$ depends on no variables and thus independent of any kernel
  witnessing $P$. \ref{ax:padding} allows conjoining $\kp{S}{\Dist[S]}$ to the
  second conjunct; since $P \depand (Q \sepand \kp{S}{\Dist[S]})$ is in \RLOGIC,
  $S$ can only mention variables that are already in $P$. Finally, \ref{ax:exch}
  shows that the standard exchange law holds for restricted assertions.
  We defer the proof to \Cref{app:extraaxioms}.

  We also need the following axioms for a particular form of atomic propositions, in addition to the axioms for general atomic propositions in \Cref{prop:ax:cpsl}.

		\begin{restatable}{proposition}{APAx}
			(\textsc{Axioms for atomic propositions})
    The following axioms are sound.
    \begin{align}
        & \kp{S}{\Dist[A] \sepand \Dist[B]} \rightarrow \kp{S}{\Dist[A]} \sepand \kp{S}{\Dist[B]}
        \qquad\text{if } A \cap B \subseteq S
        \tag{\textsc{RevPar}} \label{ax:revpar} \\
        & \kp{S}{\Dist[A] \sepand \Dist[B]} \rightarrow \kp{S}{\Dist[A \cup B]}
        \tag{\textsc{UnionRan}} \label{ax:unionran} \\
        & \pair{A}{B} \depand \pair{B}{C} \rightarrow \pair{A}{C}
        \tag{\textsc{AtomSeq}} \label{ax:atomseq} \\
        & \pair{A}{B} \rightarrow \pair{A}{A} \depand \pair{A}{B}
        \tag{\textsc{UnitL}} \label{ax:repeatdom} \\
        & \pair{A}{B} \rightarrow \pair{A}{B} \depand \pair{B}{B}
        \tag{\textsc{UnitR}} \label{ax:repeatran}
    \end{align}
			\end{restatable}

  We defer the proof to \Cref{app:extraaxioms}.

  Now, we will describe how to verify our example programs, $\ExONE$ and
  $\ExTWO$. Throughout, we must ensure that all formulas used in \SYSTEM rules
  or \RLOGIC axioms are in $\Frestrict$. The product $\depand$ raises a tricky
  point: $\Frestrict$ is not closed under reassociating $\depand$, so we add
  parentheses for formulas that must be in \RLOGIC. However, we may soundly use
  the full proof system of \LOGIC when proving implications between \RLOGIC
  assertions, since \RLOGIC is a fragment of \LOGIC.

  \paragraph*{Verification of \ExONE}
  We aim to prove the following judgment:
  \[
    \vdash \psl{\top}
    {\ExONE}
    {\kp{\Exact{\emptyset}}{\Dist[z]} \depand (\kp{\Exact{z}}{\Dist[a]} \sepand \kp{\Exact{z}}{\Dist[b]})}
  \]
  By~\Cref{theo:prob}, this shows that $a,b$ are conditionally independent given
  $z$ at the end of the program. Using \textsc{Samp} to handle the sampling for
  $z, x, y$, we can prove the assertion:
  $\empdom{\Dist[z]} \depand \empdom{\Dist[x]} \depand \empdom{\Dist[y]}$.
  Using
  Axioms~\ref{ax:padding},~\ref{ax:repeatdom},~\ref{ax:par},~\ref{ax:unionran},
  and $\depand$ \textsc{Assoc}, this assertion implies
  $\empdom{\Dist[z]} \depand \pair{z}{z,x} \depand \pair{z}{z,y}$.
  We take this as the pre-condition before assigning to $a$ and assigning to
  $b$. After the assignments, \textsc{Assn} proves:
    \[\Big( \big( \empdom{z} \depand \pair{z}{z,x} \depand \pair{z}{z,y} \big) \depand \pair{z,x}{a} \Big) \depand \pair{z, y}{b}.\]
  Then, we can apply~\ref{ax:dep2sep} to derive:
    \(\kp{\Exact{\emptyset}}{\Dist[z]} \depand
    \big(\kp{\Exact{z}}{\Dist[z,x]} \depand  \kp{\Exact{z, x}}{\Dist[a]}\big)
    \sepand
    \big(\kp{\Exact{z}}{\Dist[z,y]} \depand \kp{\Exact{z, y}}{\Dist[b]} \big)\).
  By Axiom~\ref{ax:atomseq}, we obtain the desired post-condition:
   $ \kp{\Exact{\emptyset}}{\Dist[z]} \depand (\kp{\Exact{z}}{\Dist[a]} \sepand \kp{\Exact{z}}{\Dist[b]})$. \qed

  \paragraph*{Verification of \ExTWO}
  We aim to show the following judgment:
  \[
    \vdash \psl{\top}
    {\ExTWO}
    {\kp{\Exact{\emptyset}}{\Dist[z]} \depand (\kp{\Exact{z}}{\Dist[x]} \sepand \kp{\Exact{z}}{\Dist[y]})}
  \]
  By~\Cref{theo:prob}, this shows that $x, y$ are conditionally independent
  given $z$ at the end of the program. Starting with the sampling statement for
  $z$, applying \textsc{Samp} and Axiom~\ref{ax:dep2sep2} gives:
  \[
    \vdash \psl{\top}{\Rand{z}{\Bern_{1/2}}}{\empdom{\Dist[z]} \depand \top} .
  \]
  To reason about the branching, we use \textsc{DCond}.  We start with the first
  branch.  By~\textsc{Samp},~\textsc{Weak} and~\textsc{Seq}, we have
    $\vdash \psl{\empdom{z = \ktt} \depand \top}
    {\Rand{x}{\Bern_p} \depand \Rand{y}{\Bern_p}}
    {\empdom{z = \ktt} \depand \empdom{\Dist[x]} \depand \empdom{\Dist[y]} }$.
  As before,
  Axioms~\ref{ax:padding},~\ref{ax:repeatdom},~\ref{ax:par},~\ref{ax:unionran},
  together with $\depand$ \textsc{Assoc} give the post-condition
  \[
    \empdom{z = \ktt} \depand \pair{z}{z,x} \depand \pair{z}{z,y} .
  \]
  Applying Axiom~\ref{ax:dep2sep}, we can show
    $\empdom{z = \ktt} \depand (\pair{z}{z,x} \sepand \pair{z}{z,y})$
  at the end of the branch. Thus:
   $ \vdash \psl{\empdom{z = \ktt} \depand \top}
    {\Rand{x}{\Bern_p} \depand \Rand{y}{\Bern_p}}
    {\empdom{z = \ktt} \depand \kp{z : z = \ktt}{\Dist[z,x] \sepand \Dist[z,y]}}$.
  The second branch is similar:
  \[
    \vdash \psl{\empdom{z = \kff} \depand \top}
    {\Rand{x}{\Bern_q} \depand \Rand{y}{\Bern_q}}
    {\empdom{z = \kff} \depand \kp{z : z = \kff}{\Dist[z,x] \sepand \Dist[z,y]}} .
  \]
  Applying \textsc{DCond}, we have:
  \[
    \vdash \psl{\empdom{\Dist[z]}}
    {\ExTWO}
    { \empdom{\Dist[z]} \depand (\kp{z : z = \ktt}{\Dist[z,x] \sepand \Dist[z,y]} \land \kp{z = \kff}{\Dist[z,x] \sepand \Dist[z,y]} )}
    .
  \]
  By~\ref{ax:or}, the postcondition implies
  $ \empdom{\Dist[z]} \depand \kp{(z : z = \ktt \lor z = \kff)}{\Dist[z,x] \sepand \Dist[z,y] \lor \Dist[z,x] \sepand \Dist[z,y]} $.
  In the domain and range logic, we have:
    $\models_d z : \top \to z : (z = \ktt \lor z = \kff)$ and
    \[\models_ r\Dist[z,x] \sepand \Dist[z,y] \lor \Dist[z,x] \sepand \Dist[z,y] \to \Dist[z,x] \sepand \Dist[z,y].\]
  So~\ref{ax:cons} implies
  $\empdom{\Dist[z]} \depand \kp{\Exact{z}}{\Dist[z,x] \sepand \Dist[z,y]}$.
  We can then apply~\ref{ax:revpar} because $\{ z, x \} \cap \{ z, y \} = z$,
  deriving the postcondition
   $ \empdom{\Dist[z]} \depand (\kp{\Exact{z}}{\Dist[z,x]} \sepand \kp{\text{z}}{\Dist[z,y]})$.
  By Axiom~\ref{ax:split}, we obtain the desired post-condition:
    $\empdom{\Dist[z]} \depand (\kp{\Exact{z}}{\Dist[x]} \sepand \kp{\text{z}}{\Dist[y]})$. \qed

			\subsection{Section \ref{sec:cpsl-assertions}, atomic propositions: Omitted Details}
			\label{app:atomic}

			As we described in \cref{sec:cpsl-assertions}, atomic formulas for
			\SYSTEM are of the form $\kp{D}{R}$. The domain assertions $D$ are of
			the form $S : \phi_d$, where $S$ is a set of variables and $\phi_d$
			describes memories, and the range assertions $R$ are of the form
			$\phi_r$, where $\phi_r$ is from a fragment of probabilistic BI.

			\axschemasound*
			\begin{proof}
				We check each of the axioms.
				\begin{description}[leftmargin=*]
					\item[Case: \ref{ax:and}.]
Suppose that $w \models \kp{S : p_d}{p_r} \land \kp{S :
p_d'}{p_r'}$. By semantics of atomic propositions, there exists
$w_1 \sqsubseteq_k w$ and $w_2 \sqsubseteq_k w$ such that for
all $m \in \R{S}$ such that $m \models_d p_d \land p_d'$, we
have $w_1(m) \models_r p_r$ and $w_2(m) \models_r p_r'$. By
restriction (\cref{thm:restriction}), we may assume that
$\range(w_1) = \FV(p_r) = \FV(p_r') = \range(w_2)$. Thus,
\cref{prop:subkernel} implies that $w_1 = w_2$, and so $w
\models \kp{S : p_d \land p_d'}{p_r \land p_r'}$.
					\item[Case: \ref{ax:or}.]
Immediate, by semantics of $\lor$.
					\item[Case: \ref{ax:par}.]
Suppose that $w \models \kp{S : p_d}{p_r} \sepand \kp{S' :
	p_d'}{p_r'}$. We will show that $w \models \kp{S \cup S' : p_d
\sepand p_d'}{p_r \sepand p_r'}$.

By semantics of atomic propositions, there exists $w_1
\sqsubseteq_k w$ and $w_2 \sqsubseteq_k w$ such that $w_1 \oplus
w_2 \sqsubseteq w$, and for all $m_1 \in \R{S}$ such that $m_1
\models_d p_d$, we have $w_1(m_1) \models_r p_r$, and for all
$m_2 \in \R{S'}$ such that $m_2 \models_d p_d'$, we have
$w_2(m_2) \models_r p_r'$.

Now for any $m \in \R{S \cup S'}$ such that $m \models_d p_d
\land p_d'$, we have $m^S \models_d p_d$ and $m^{S'} \models_d
p_d'$. Thus $w_1(m^S) \models_r p_r$ and $w_2(m^{S'}) \models_r
p_r'$. Letting $T = S \cap S'$ and $T_1 = S \setminus T$; $T_2 =
S' \setminus T$ be disjoint sets, and noting that $w_1, w_2$
both preserve inputs on $T$, we have:
\begin{align*}
	w_1 \oplus w_2(m)
&= \pi_{T_1}w_1(m^S) \otimes \dunit(m^T) \otimes \pi_{T_2}w_2(m^{S'}) \\
&= (\pi_{T_1}w_1(m^S) \otimes \dunit(m^T)) \oplus_r (\dunit(m^T)\otimes \pi_{T_2}w_2(m^{S'})) \\
&= w_1(m^S) \oplus_r w_2(m^{S'}) \\
&\models_r p_r \sepand p_r'
\end{align*}
Thus, $w \models \kp{S \cup S' : p_d \sepand p_d'}{p_r \sepand p_r'}$.
					\item[Case: \ref{ax:cons}.]
Immediate, by semantics of $\rightarrow$.
				\end{description}
			\end{proof}


			For the proof of \Cref{thm:restriction}, we need the following characterization of $g\sqsubseteq f$.

			\begin{proposition} \label{prop:subkernel}
				Let $f$ be a Markov kernel, and let $D \subseteq \dom(f) \subseteq R \subseteq
				\range(f)$. Then we have $\pi_{R}(f(m)) = g(m')$ for all $m' \in \R{D}, m \in
				\R{\dom(f)}$ such that $m^{D} = m'$ if and only if $g \sqsubseteq f$ and
				$\dom(g) = D, \range(g) = R$.
			\end{proposition}
			\begin{proof}
				For the reverse direction, suppose that $f = (g \oplus \unit_{S}) \odot v$,
				with $S$ disjoint from $\dom(g)$. Since $\range(g) \subseteq \dom(v)$, we
				have:
				\begin{align*}
					\pi_{R}(f(m)) &= \pi_{R}((g \oplus \unit_{S})(m)) \\
	&= \pi_{R}(g(m^{D}) \oplus \unit_{S}(m^S)) \\
	&= \pi_{R}(g(m^{D})) \otimes \pi_{R}(\unit_{S}(m^S)) \\
	&= g(m^{D}) \\
	&= g(m') .
				\end{align*}
				For the forward direction, evidently $\dom(g) = D$ and $\range(g) = R$. Since
				$f$ preserves input to output, we have $\pi_{\dom(f)}(g(m')) =
				\pi_{\dom(f)}(f(m)) = \unit(m')$ so $g$ preserves input to output and $g$ is a
				Markov kernel. We claim that $g \sqsubseteq f$. First, consider $g \oplus
				\unit_{\dom(f) \setminus D}$; write $D' = \dom(f) \setminus D$. For any $m \in
				\R{\dom(f)}$, we have:
				\begin{align*}
					\pi_{D' \cup R}(f(m)) &= \pi_{R}(f(m)) \otimes \pi_{D'}(f(m)) \\
			&= g(m^D) \otimes \unit_{D'}(m^{D'}) \\
			&= (g \oplus \unit_{D'})(m) .
				\end{align*}
				So by \cref{lemma:condfirst}, for every $m \in \R{\dom(f)}$ there exists a family
				of kernels $g'_m : \R{D' \cup R} \to \DR{\range(f)}$ such that
				\[
					f(m) = \dbind((g \oplus \unit_{D'})(m), g'_m)
				\]
				Defining $g'(m) \triangleq g'_{m^{\dom(f)}}(m)$, we have:
				\[
					f(m) = ((g \oplus \unit_{D'}) \odot g') (m)
				\]
				and so $g \sqsubseteq f$.
			\end{proof}

			We prove that all assertions in the restricted logic \RLOGIC satisfy
			restriction.

			\restriction*
			\begin{proof}
				The reverse direction is immediate from persistence. For the forward
				direction, we argue by induction with a stronger hypothesis. If $f
				\models P$, we call a state $f'$ a \emph{witness} of $f \models P$
				if $f' \sqsubseteq f$, $\SFV(P) \subseteq \range(f') \subseteq
				\FV(P)$, $\dom(f') \subseteq \FFV(P)$, and $f' \models P$.  We show
				that $f \models P$ implies that there is a witness $f' \models P$,
				by induction on $P$.

				\begin{description}[leftmargin=*]
					\item [Case $\kp{D}{R}$:]
We will use two basic facts, both following from the form of the
domain and range assertions:
\begin{enumerate}
	\item If $m \models_d D$, then $dom(m) = \FV(D)$.
	\item If $\mu \models_r R$, then $dom(\mu) \supseteq \FV(D)$.
\end{enumerate}
$f \models \kp{D}{R}$ implies that there exists
$f' \sqsubseteq f$ such that
for any $m \in M_d$ such that $m \models_d D$, $f'(m)$ is defined and $f'(m) \models_r R$.

Let $T = \range(f') \cap (\FV(D) \cup \FV(R))$. We claim that $\pi_T f'$ is the desired witness for $f \models P$.
\begin{itemize}
	\item $\pi_T f'$ is defined and  $\pi_T f' \sqsubseteq f$ because:
		\begin{align*}
			\dom(f') &= dom(m) \tag{for any $m \in M_d$ such that $m \models_d D$}\\
&= \FV(D) \notag \\
&\subseteq T \notag .
		\end{align*}
		Thus $\pi_T f'$ is defined, and $\pi_T  f' \sqsubseteq f' \sqsubseteq f$.
	\item $\range(\pi_T f') = T \subseteq \FV(D) \cup \FV(R) = \FV(P)$.
	\item $\pi_T f' \models \kp{D}{R}$: For any
		$m \in M_d$ such that $m \models_d D$, $f'(m)$ is a
		distribution. Based on the restriction theorem for
		probabilistic BI, $\pi_{\FV(R) \cap \range(f')} (f'(m))
		\models R$ too. Since $T \supseteq \FV(R) \cap \range(f')$,
		persistence in $M_r$, implies $\pi_{T} (f'(m)) \models R$.
		By definition of
		marginalization on kernels, $(\pi_{T} f') (m) = \pi_{T} (f' (m))$. Since
		$(\pi_{T} f') (m) \models R$, we have $\pi_{T} f' \models \kp{D}{R}$ as well.
	\item  $\FFV(P) = \FV(D)$, so
			$\dom(\pi_T f') = dom(m) = \FV(D) = \FFV(P)$.
	\item  $\SFV(P) = \FV\kp{D}{R} = \FV(D) \cup \FV(R)$, so
		\begin{align*}
			\range(\pi_T f') &\supseteq dom((\pi_T f')(m))\tag{for any $m \in M_d$ such that $m \models_d D$}\\
		&\supseteq \FV(D) \cup \FV(R) \tag{By $(\pi_T f')(m) \models R$} \\
		&= \SFV(P) .
		\end{align*}
\end{itemize}
so $\pi_T f'$ is a desired witness for $f \models P$.
					\item [Case $Q \land R$:]
Assuming $\SFV(Q) = \FV(Q) = \SFV(R) = \FV(R)$.
By definition, $f \models Q \land R$ implies that $f \models Q$ and $f \models R$.
By induction,	there exists $f' \sqsubseteq f$ such that $\SFV(Q)
= \range(f') = \FV(Q)$, $\dom(f') \subseteq \FFV(Q)$, and $f'
\models Q$, and there exists $f'' \sqsubseteq f$ such that
$\SFV(R) = \range(f'') = \FV(R)$, $\dom(f'') \subseteq \FFV(R)$
and $f'' \models R$.
Thus, $\range(f') = \range(f'')$.

Note that $\dom(f') = \dom(f) \cap \range(f')$ because in our models,
$f' \sqsubseteq f$ implies that there exists $S$  and some $v$ such that $f = (f' \oplus \eta_{S}) \odot v$, and we can make $S$ disjoint of $\dom(f')$ and $\range(f')$ wolog.
Then, $\dom(f) = \dom(f' \oplus S) = \dom(f') \cup S$,
and $\range(f') = \range(f' \oplus S) \setminus S$,
so $\dom(f) \cup \range(f') \subseteq \dom(f')$.
Meanwhile, since $\dom(f') \subseteq \dom(f)$ and $\dom(f') \subseteq \range(f')$, $\dom(f') \subseteq \dom(f) \cap \range(f')$.
So $\dom(f') = \dom(f) \cap \range(f')$.
Similarly, $\dom(f'') \subseteq \dom(f) \cap \range(f'')$, so
$\range(f') = \range(f'')$ implies that $\dom(f') = \dom(f')$.

Since $\dom(f') = \dom(f'')$ and
$\range(f') = \range(f'')$, \cref{prop:subkernel} implies that
$f' = f''$. This is the desired witness: $f' = f'' \models Q$
and $f' = f'' \models R$.
					\item [Case $Q \lor R$:]
$f \models Q \lor R$ implies that $f \models Q$ or $f \models
R$.

                        Without loss of generality, suppose $f \models Q$.
By induction,	there exists $f' \sqsubseteq f$ such that $\SFV(Q)
\subseteq \range(f') \subseteq \FV(Q)$, $\dom(f') \subseteq
\FFV(Q)$. Then:
\begin{align*}
	\range(f') &\subseteq \FV(Q) \cup \FV(R) = \FV(P) \\
	\range(f') &\supseteq \SFV(Q) \cap \SFV(R) =\SFV(P) \\
	\dom(f') &\subseteq \FV(Q) \cup \FV(R) = \FFV(P) .
\end{align*}
Thus, $f'$ is a desired witness.
					\item [Case $Q \depand R$:]
Assuming $\FFV(R) \subseteq \SFV(Q)$.

$f \models Q \depand R$ implies that there exists $f_1, f_2$ such that
$f_1 \odot f_2 = f$, $f_1 \models Q$, and $f_2 \models R$. $f_1 \odot f_2$ is defined so $\range(f_1) = \dom(f_2)$.
By induction, there exists $f_1' \sqsubseteq f_1$ such that $f_1' \models Q$, $\SFV(Q) \subseteq \range(f_1') \subseteq \FV(Q)$ and $\dom(f_1') \subseteq \FFV(Q)$, and  there exists $f_2' \sqsubseteq f_2$ such that $f_2' \models Q$, $\SFV(R) \subseteq \range(f_2') \subseteq \FV(R)$, and $\dom(f_2') \subseteq \FFV(R)$.

Now,$\wh{f} = f_1' \odot (f_2' \oplus \unit_{\range(f_1') \setminus \dom(f_2')} )$ is defined because $\dom(f_2') \subseteq \FFV(R) \subseteq \SFV(Q) \subseteq \range(f_1')$. Then, we have
\begin{align*}
	\wh{f} &\models Q \depand R \\
	\range(\wh{f}) &= \range(f_1') \cup \range(f_2') \subseteq \FV(Q) \cup \FV(R) = \FV(P) \\
	\range(\wh{f}) &= \range(f_1') \cup \range(f_2') \supseteq \SFV(Q) \cup \SFV(R) = \SFV(P) \\
	\dom(\wh{f}) &= \dom(f_1') \subseteq \FFV(Q) = \FFV(P) .
\end{align*}
$f_1' \sqsubseteq f$, $f_2' \oplus \unit_{\range(f_1') \setminus \dom(f_2')}  \oplus  \sqsubseteq f_2$, so by~\cref{odotdownwards},
$\wh{f} = f_1' \odot (f_2' \oplus \unit_{\range(f_1') \setminus \dom(f_2')} ) \sqsubseteq f_1 \odot f_2 = f$.

Thus, $\wh{f}$ is a desired witness.

%
					\item [Case $Q \sepand R$:]
$f \models Q \sepand R$ implies that there exists $f_1, f_2$ such that
$f_1 \oplus f_2 \sqsubseteq f$, $f_1 \models Q$, and $f_2 \models R$.

By induction, there exists $f_1' \sqsubseteq f_1$ such that $f_1' \models Q$, $\SFV(Q) \subseteq \range(f_1') \subseteq \FV(Q)$ and $\dom(f_1') \subseteq \FFV(Q)$, and  there exists $f_2' \sqsubseteq f_2$ such that $f_2' \models Q$, $\SFV(R) \subseteq \range(f_2') \subseteq \FV(R)$, and $\dom(f_2') \subseteq \FFV(R)$.
By downwards closure of $\oplus$, $f_1' \oplus f_2'$ is defined
and $f_1' \oplus f_2' \sqsubseteq f_1 \oplus f_2 \sqsubseteq f$.
We have $f_1' \oplus f_2' \models Q \sepand R$, and
\begin{align*}
	\range(f_1' \oplus f_2') &= \range(f_1') \cup \range(f_2') \subseteq \FV(Q) \cup \FV(R) = \FV(P) \\
	\range(f_1' \oplus f_2') &= \range(f_1') \cup \range(f_2') \supseteq \SFV(Q) \cup \SFV(R) = \SFV(P) \\
	\dom(f_1' \oplus f_2') &= \dom(f_1') \cup \dom(f_2') \subseteq \FFV(Q) \cup \FFV(R) = \FFV(P) .
\end{align*}
Thus, $f_1' \oplus f_2'$ is a desired witness.
				\end{description}
			\end{proof}

			\subsection{Section \ref{sec:cpsl_logic}, \SYSTEM: Omitted Details}
			\label{app:cpsl}

			To prove soundness for \SYSTEM (\Cref{thm:cpslsound}), we rely on a few lemmas about program
			semantics.

			\begin{lemma} \label{lem:assn-sem-ker}
				Suppose that $e$ is an expression not containing $x$, and let $\mu \in
				\DR{\Var}$. Then:
				\[
					f_{\denot{\Assn{x}{e}} \mu}
					= f_{\mu}
					\odot (m \mapsto \dunit(m^{\Var \setminus \{ x \}}))
					\odot ((m_1 \mapsto \dunit(m_1 \cup (x \mapsto \denot{e}(m_1))))
					\oplus (m_2 \mapsto \dunit(m_2)))
				\]
				where $m_1 \in \R{\Var \setminus \{ x \}}$ and $m_2 \in \R{\Var \setminus \{ x
				\} \setminus \FV(e)}$.
			\end{lemma}

			\begin{lemma} \label{lem:samp-sem-ker}
				Suppose that $d$ is a distribution expression not containing $x$, and let $\mu
				\in \DR{\Var}$. Then:
				\[
					f_{\denot{\Rand{x}{d}} \mu}
					= f_{\mu}
					\odot (m \mapsto \dunit(m^{\Var \setminus \{ x \}}))
					\odot ((\denot{d} \odot (v \mapsto [x \mapsto v])) \oplus (m_1 \mapsto \dunit(m_1)) \oplus (m_2 \mapsto \dunit(m_2)))
				\]
				where $m_1 \in \R{\Var \setminus \{ x \}}$ and $m_2 \in \R{\Var \setminus \{ x
				\} \setminus \FV(d)}$, and $\denot{d} : \R{\FV(d)} \to \DD(\Val)$.
			\end{lemma}

                  The rule \textsc{Frame} relies on simple syntactic conditions
                  for approximating which variables \emph{may} be read, which
                  variables \emph{must} be written before they are read, and
                  which variables \emph{may} be modified.

                  \begin{definition}\label{def:var-cond}
                    $\RV, \WV, \MV$ are defined as follows:
                    \begin{mathpar}
                      \RV(\Assn{x}{e}) \defeq \FV(e) \and
                      \RV(\Rand{x}{d}) \defeq \FV(d) \\
                      \RV(\Seq{c}{c'}) \defeq \RV(c) \cup (\RV(c') \setminus \WV(c)) \and
                      \RV(\Cond{b}{c}{c'}) \defeq \FV(b) \cup \RV(c) \cup \RV(c')
                    \end{mathpar}
                    \hrule
                    \begin{mathpar}
                      \WV(\Assn{x}{e}) \defeq \{ x \} \setminus \FV(e) \and
                      \WV(\Rand{x}{d}) \defeq \{ x \} \setminus \FV(d) \\
                      \WV(\Seq{c}{c'}) \defeq \WV(c) \cup (\WV(c') \setminus \RV(c)) \and
                      \WV(\Cond{b}{c}{c'}) \defeq (\WV(c) \cap \WV(c')) \setminus \FV(b)
                    \end{mathpar}
                    \hrule
                    \begin{mathpar}
                      \MV(\Assn{x}{e}) \defeq \{ x \} \and
                      \MV(\Rand{x}{d}) \defeq \{ x \} \\
                      \MV(\Seq{c}{c'}) \defeq \MV(c) \cup \MV(c') \and
                      \MV(\Cond{b}{c}{c'}) \defeq \MV(c) \cup \MV(c')
                    \end{mathpar}
                  \end{definition}

                  Other analyses are possible, so long as non-modified variables are preserved
                  from input to output, and output modified variables depend only on input read
                  variables.

                  \begin{lemma}[Soundness for \RV, \WV, \MV~\citep{barthe2019probabilistic}]\label{lem:fv}
                    Let $\mu' = \denot{c}\mu$, and let $R = \RV(c), W = \WV(c),
                    C = \Var \setminus \MV(c)$. Then:
                    \begin{enumerate}
                      \item Variables outside of $\MV(c)$ are not modified:
                        $\pi_{C}(\mu') =
                        \pi_{C}(\mu)$.
                      \item The sets $R$ and $W$ are disjoint.
                      \item There exists $f : \R{R} \to \DR{\MV(c)}$ with
                        $\mu' = \dbind(\mu, m \mapsto f(\pi_{R}(m)) \otimes
                        \dunit(\pi_{C}(m)))$.
                    \end{enumerate}
                  \end{lemma}

			We recall the definition of validity in \SYSTEM.

			\cpslvalid*

			Now, we are ready to prove soundness of \SYSTEM.

			\cpslsound*

			\begin{proof}
				By induction on the derivation. Throughout, we write $\mu : \DR{\Var}$ for the
				input and $f : \R{\emptyset} \to \DR{\Var}$ for the lifted input, and we
				assume that $f$ satisfies the pre-condition of the conclusion.

				\begin{description}[leftmargin=*]
					\item[Case: \textsc{Assn}.]
By restriction (\cref{thm:restriction}), there exists $k_1 \sqsubseteq f$ such that $\FV(e) \subseteq
S\FV(P) \subseteq \range(k_1) \subseteq \FV(P)$; let $K = \range(k_1)$. Since
$f$ has empty domain, we have $f = k_1 \odot k_2$ for some $k_2 : \R{K}
\to \DR{\Var}$. Let $f' =f_{\denot{\Assn{x}{e}}\mu}$ be the lifted
output. By \cref{lem:assn-sem-ker} and associativity, we have:
\begin{align*}
	f' &= f
	\odot (m \mapsto \dunit(m^{\Var \setminus \{ x \}}))
	\odot ((m_1 \mapsto \dunit(m_1 \cup (x \mapsto \denot{e}(m_1))))
	\oplus (m_2 \mapsto \dunit(m_2))) \\
				&= \underbrace{k_1 \odot k_2 \odot (m \mapsto \dunit(m^{\Var \setminus \{ x \}}))}_{j}
				\odot (\underbrace{m_1 \mapsto \dunit(m_1 \cup (x \mapsto \denot{e}(m_1)))}_{j_1}
				\oplus \underbrace{m_2 \mapsto \dunit(m_2)}_{j_2})
\end{align*}
where $m : \R{\Var}$, $m_1 : \R{\FV(e)}$, and $m_2 : \R{\Var \setminus
\FV(e) \setminus \{ x \} }$. Note that even though the components of $j$ do
not preserve input to output, $j$ itself does preserve input to output;
$j_1$ and $j_2$ also evidently have this property. Now since $k
\sqsubseteq j$ and $k_1 \models P$, we have $j \models P$. Since $j_1
\sqsubseteq j_1 \oplus j_2$ and $j_1 \models \kp{\FV(e)}{x = e}$, we have
$j_1 \oplus j_2 \models \kp{\FV(e)}{x = e}$ as well. Thus, we conclude $f'
\models P \depand \kp{\FV(e)}{x = e}$.
					\item[Case: \textsc{Samp}.]
By restriction (\cref{thm:restriction}), there exists $k_1 \sqsubseteq f$ such that $\FV(d) \subseteq
S\FV(P) \subseteq \range(k_1) \subseteq \FV(P)$; let $K = \range(k_1)$. Since
$f$ has empty domain, we have $f = k_1 \odot k_2$ for some $k_2 : \R{K}
\to \DR{\Var}$. Let $f' = f_{\denot{\Assn{x}{e}}\mu}$ be the lifted
output. By \cref{lem:samp-sem-ker} and associativity, we have:
\begin{align*}
	f' &= f
	\odot (m \mapsto \dunit(m^{\Var \setminus \{ x \}}))
	\odot ((\denot{d} \odot (v \mapsto [x \mapsto v])) \oplus (m_1 \mapsto \dunit(m_1)) \oplus (m_2 \mapsto \dunit(m_2))) \\
				&= \underbrace{k_1 \odot k_2 \odot (m \mapsto \dunit(m^{\Var \setminus \{ x \}}))}_{j}
				\odot (\underbrace{(\denot{d} \odot (v \mapsto [x \mapsto v])) \oplus (m_1 \mapsto \dunit(m_1))}_{j_1}
				\oplus \underbrace{m_2 \mapsto \dunit(m_2)}_{j_2})
\end{align*}
where $m : \R{\Var}$, $\denot{d} : \R{\FV(d)} \to \DR{\Val}$, $m_1 :
\R{\FV(d)}$, and $m_2 : \R{\Var \setminus \FV(d) \setminus \{ x \} }$. Note
that even though the components of $j$ do not preserve input to output,
$j$ itself does preserve input to output; $j_1$ and $j_2$ also evidently
have this property. Now since $k \sqsubseteq j$ and $k_1 \models P$, we
have $j \models P$. Since $j_1 \sqsubseteq j_1 \oplus j_2$ and $j_1
\models \kp{\FV(d)}{x \sim d}$, we have $j_1 \oplus j_2 \models
\kp{\FV(d)}{x \sim d}$ as well. Thus, we conclude $f' \models P \depand
\kp{\FV(d)}{x \sim d}$.
					\item[Case: \textsc{Skip}.]
Trivial.
					\item[Case: \textsc{Seqn}.]
Trivial.
					\item[Case: \textsc{DCond}.]
Since all assertions are in \RLOGIC, we have $\FFV(P) \subseteq
\SFV\kp{\Exact{\emptyset}}{\Dist[b]} = \{ b \}$. Since $f \models
\kp{\Exact{\emptyset}}{\Dist[b]}$, there exists $k_1, k_2$ such that $k_1
\odot k_2 = f$, with $k_1 \models
\kp{\Exact{\emptyset}}{\Dist[b]}$ and $k_2 \models P$.

By restriction (\cref{thm:restriction}), there exists $j_1$ such that $j_1 \sqsubseteq k_1$ and
\begin{align*}
	\dom(j_1) &\subseteq \FFV\kp{\Exact{\emptyset}}{\Dist[b]} = \emptyset \\
	\{ b \} &= \SFV\kp{\Exact{\emptyset}}{\Dist[b]}
	\subseteq \range(j_1)
	\subseteq \FV\kp{\Exact{\emptyset}}{\Dist[b]}
	= \{ b \} .
\end{align*}

By restriction (\cref{thm:restriction}), there exists $j_2$ such that $j_2
\sqsubseteq k_2$ and $j_2 \models P$, and $\dom(j_2) \subseteq \FFV(P)
\subseteq \SFV\kp{\Exact{\emptyset}}{\Dist[b]} = \{ b \}$. Since
$\dom(k_2) = \range(k_1) \supseteq \{ b \}$, we may assume without loss of
generality that $j_2 \models P$, $j_2 \sqsubseteq k_2$, and $\dom(j_2) =
\{ b \}$.  Thus $j_1 \odot j_2$ is defined, and so $j_1 \odot j_2
\sqsubseteq k_1 \odot k_2 \sqsubseteq f$ by \cref{odotdownwards}.

By \cref{lemma:condfirst}, there exists $j : \R{\range(j_2)} \to \DR{\Var}$
such that $j_1 \odot (j_2 \odot j) = (j_1 \odot j_2) \odot j = f$. Since
$j_2 \sqsubseteq j_2 \odot j$, we have $j_2 \odot j \models P$. Thus, we
may assume without loss of generality that $\range(j_2) = \Var$ and $j_1
\odot j_2 = f = \overline{\mu}$.

Let $l_{\ktt}, l_{\kff} : \R{\emptyset} \to \DR{b}$ be
defined by $l_{\ktt}(\empmem) = \dunit{[b = \ktt]}$ and
$l_{\kff}(\empmem) = \dunit{[b = \kff]}$; evidently, $l_{\ktt}
\models \kp{\Exact{\emptyset}}{b = \ktt}$ and $l_{\kff} \models
\kp{\Exact{\emptyset}}{b = \kff}$. Now, we have:
\begin{align*}
	f_{\dcond{\mu}{\denot{b = \ktt}}} &= l_{\ktt} \odot j_2
	\\
	f_{\dcond{\mu}{\denot{b = \kff}}} &= l_{\kff} \odot j_2
\end{align*}
where each equality holds if the left side is defined. Regardless of
whether the conditional distributions are defined, we always have:
\begin{align*}
	l_{\ktt} \odot j_2 &\models \kp{\Exact{\emptyset}}{b = \ktt} \depand P
	\\
	l_{\kff} \odot j_2 &\models \kp{\Exact{\emptyset}}{b = \kff} \depand P .
\end{align*}
Since both of these kernels have empty domain, we have $l_{\ktt} \odot j_2
= \overline{\nu_{\ktt}}$ and $l_{\kff} \odot j_2 = \overline{\nu_{\kff}}$
for two distributions $\nu_{\ktt}, \nu_{\kff} \in \DR{\Var}$. By
induction, we have:
\begin{align*}
	f_{\denot{c}\nu_{\ktt}} &\models \kp{\Exact{\emptyset}}{b = \ktt} \depand \kp{b : b = \ktt}{Q_1}
	\\
	f_{\denot{c}\nu_{\kff}} &\models \kp{\Exact{\emptyset}}{b = \kff} \depand \kp{b : b = \kff}{Q_2} .
\end{align*}
By similar reasoning as for the pre-conditions, there exists $k_1', k_2' :
\R{b} \to \DR{\Var}$ such that $k_1' \models \kp{b : b = \ktt}{Q_1}$ and $k_2'
\models \kp{b : b = \kff}{Q_2}$,
and:
\begin{align*}
	f_{\denot{c}\nu_{\ktt}} &= l_{\ktt} \odot k_1'
	\qquad \qquad
	f_{\denot{c}\nu_{\kff}} = l_{\kff} \odot k_2' .
\end{align*}
Let $k' : \R{b} \to \DR{\Var}$ be the composite kernel defined as follows:
\[
	k'([b \mapsto v]) \triangleq
	\begin{cases}
		k_1'([b \mapsto \ktt]) &: v = \ktt \\
		k_2'([b \mapsto \kff]) &: v = \kff
	\end{cases} .
\]
By assumption, $k' \models (\kp{b : b = \ktt}{Q_1} \land \kp{b : b = \kff}{Q_2})$.
Now, let $p \triangleq \mu(\denot{b = \ktt})$ be the probability of taking
the first branch. Then we can conclude:
\begin{align*}
	f_{\denot{\Cond{b}{c}{c'}}\mu}
		&= f_{\dconv{p}
			{\denot{c} (\dcond{\mu}{\denot{b = \ktt}})}
		{\denot{c'} (\dcond{\mu}{\denot{b = \ktt}})}}
          \\
        &= f_{\dconv{p}
          {\denot{c}\nu_{\ktt}}
          {\denot{c}\nu_{\kff}}}
          \\
        &= \kconv{p}
          {f_{\denot{c}\nu_{\ktt}}}
          {f_{\denot{c}\nu_{\kff}}}
          \\
        &= \kconv{p} {(l_{\ktt} \odot k_1')} {(l_{\kff} \odot k_2')}
          \\
        &= \kconv{p} {(l_{\ktt} \odot k')} {(l_{\kff} \odot k')}
          \\
        &= (\kconv{p} {l_{\ktt}} {l_{\kff}}) \odot k'
          \\
        &\models \kp{\Exact{\emptyset}}{\Dist[b]} \depand (\kp{b : b = \ktt}{Q_1} \land \kp{b : b = \kff}{Q_2}) .
      \end{align*}
      Above, $\kconv{p}{k_1}{k_2}$ lifts the convex combination operator from
      distributions to kernels from $\R{\emptyset}$. We show the last equality
      in more detail. For any $r \in \R{\Var}$:
      \begin{align*}
        &\kconv{p} {(l_{\ktt} \odot k')} {(l_{\kff} \odot k')}(\empmem)(r)\\
        &= p \cdot (l_{\ktt} \odot k')(\empmem)(r) + (1 - p) \cdot (l_{\kff} \odot k')(\empmem)(r) \\
        &= p \cdot (l_{\ktt} \odot k')(\empmem)(r) + (1 - p) \cdot (l_{\kff} \odot k')(\empmem)(r) \\
        &= p \cdot l_{\ktt}(\empmem)(b \mapsto \ktt) \cdot k'(b \mapsto \ktt)(r)
        + (1 - p) \cdot l_{\kff}(\empmem)(b \mapsto \kff) \cdot k'(b \mapsto \kff)(r) \\
        &= ((\kconv{p}{l_{\ktt}}{l_{\kff}}) \odot k')(\empmem)(r) .
      \end{align*}
      where the penultimate equality holds because $l_{\ktt}$ and $l_{\kff}$ are
      deterministic.
    \item[Case: \textsc{Weak}.]
      Trivial.
    \item[Case: \textsc{Frame}.]
      The proof for this case follows the argument for \textsc{Frame} rule in
      PSL, with a few minor changes.

      There exists $k_1, k_2$ such that $k_1 \oplus k_2 \sqsubseteq f$, and $k_1
      \models P$ and $k_2 \models R$; let $S_1 \triangleq \range(k_1)$, and note
      that $\RV(c) \subseteq S_1$ by the last side-condition.  By restriction
      (\cref{thm:restriction}), there exists $k_2' \sqsubseteq k_2$ such that
      $k_2' \models R$ and $\range(k_2') \subseteq \FV(R)$; let $S_2 \triangleq
      \range(k_2')$. Since $k_1$ and $k_2$ have empty domains, $S_1$ and $S_2$
      must be disjoint. Let $S_3 = \Var \setminus S_2 \setminus S_1$. Since
      $\WV(c)$ is disjoint from $S_2$ by the first side-condition, we have $\WV(c)
      \subseteq S_1 \cup S_3$.

      Let $f' = f_{\denot{c}\mu}$ be the lifted output. By induction, we
      have $f' \models Q$; by restriction (\cref{thm:restriction}), there
      exists $k_1' \sqsubseteq f'$ such that $\range(k_1') \subseteq \FV(Q)$ and
      $k_1' \models Q$. By the third side condition, $\RV(c) \subseteq \SFV(P)
      \subseteq S_1$.

      By soundness of $\RV$ and $\WV$ (\cref{lem:fv}), all variables in $\WV(c)$
      must be written before they are read and there is a function $F : \R{S_1}
      \to \DR{\WV(c) \cup S_1}$ such that:
      \[
        \pi_{\WV(c) \cup S_1} \denot{c} \mu = \dbind(\mu, m \mapsto F(m^{S_1})) .
      \]
      Since $S_2 \subseteq \FV(R)$, variables in $S_2$ are not in $\MV(c)$ by the
      first side-condition, and $S_2$ is disjoint from $\WV(c) \cup S_1$.
      By soundness of $\MV$, we have:
      \[
        \pi_{\WV(c) \cup S_1 \cup S_2} \denot{c}\mu
        = \dbind(\pi_{\WV(c) \cup S_1 \cup S_2} \mu, F \oplus \dunit)
      \]
      where $\dunit : \R{\WV(c) \cup S_2} \to \DR{\WV(c) \cup S_2}$.

      Since $S_1$ and $S_2$ are independent in $\mu$, we know that $S_1 \cup
      \WV(c)$ and $S_2$ are independent in $\denot{c}\mu$. Hence:
      \[
       f_{\pi_{S_1 \cup \WV(c)}\denot{c}\mu}
        \oplus f_{\pi_{S_2}\denot{c}\mu}
        \sqsubseteq f' .
      \]
      By induction, $f' \models Q$. Furthermore, $\FV(Q) \subseteq \SFV(P) \cup
      \WV(c) \subseteq S_1 \cup \WV(c)$ by the second side-condition.  By
      restriction (\cref{thm:restriction}), $f_{\pi_{S_1 \cup
      \WV(c)}\denot{c}\mu} \models Q$.  Furthermore, $\pi_{S_2}\denot{c}\mu =
      \pi_{S_2} \mu$, so $\pi_{S_2}\denot{c}\mu \models R$ as well. Thus, $f'
      \models Q \sepand R$ as desired.
  \end{description}
\end{proof}

\subsection{Section~\ref{sec:cpsl_ex}, proving CI: omitted proofs}
\label{app:extraaxioms}

\RLOGICAx*
\begin{proof} We prove them one by one.
	\begin{description}[leftmargin=*]
		\item [\ref{ax:dep2sep}]
			We want to show that when $(P \depand Q) \depand R$, $P \depand (Q \sepand R)$ are both formula in $\RLOGIC$,
			$f \models  (P \depand Q ) \depand R$ implies $f \models P \depand (Q \sepand R)$.

		By proof system of \LOGIC, $f \models  (P \depand Q ) \depand R$ implies that
		$f \models  P \depand \big(Q \depand R\big)$.
		While $ P \depand \big(Q \depand R\big)$ may not satisfy the restriction property, that is okay because
		we will only used conditions guaranteed by the fact that
		$(P \depand Q) \depand R$, $P \depand (Q \sepand R) \in \Frestrict$.
	In particular, we rely on $P, Q, R$ each satisfies restriction, and
	\(\FFV(Q \sepand R) \subseteq \SFV(P),\)
	which implies that
	\begin{align}
		\FFV(R) \subseteq \FFV(Q \sepand R) \subseteq \SFV(P) \label{FFVRinSFVP}
	\end{align}

	$f \models  P \depand \big(Q \depand R\big)$ implies there exists $f_p, f_q, f_r$ such that $f_p \models P$, $f_q \models Q$, and  $f_r \models R$, and $f_p \odot (f_q \odot f_r) = f$.

By restriction property~\cref{thm:restriction}, $f_q \models Q$ implies that there exists $f_q' \sqsubseteq f_q$ such that $\SFV(Q) \subseteq \range(f_q') \subseteq \FV(Q)$ and $\dom(f_q') \subseteq\FFV(Q)$. $f_q' \sqsubseteq f_q$ so there exists $v, T$ such that $f_q = (f_q' \oplus_k \unit_T) \odot v$.

	Similarly, $f_r \models R$, by~\cref{thm:restriction}, there exists $f_r' \sqsubseteq f_r$ such that $\SFV(R) \subseteq \range(f_r') \subseteq \FV(R)$ and $\dom(f_r') \subseteq\FFV(R)$.	$f_r' \sqsubseteq f_r$ so there exists $u, S$ such that $f_r = (f_r' \oplus_k \unit_S) \odot u$.

Now,	we claim that $\FFV(R) \subseteq \dom(f_q' \oplus \unit_{T})$:

By~\cref{thm:restriction}
	$f_p \models P$ implies that there exists $f_p' \sqsubseteq f_p$ such that $\SFV(P) \subseteq \range(f_p') \subseteq \FV(P)$,
	$\dom(f_p') \subseteq F\FV(P),\, \text{and } f_p' \models P. $
	Thus, 	$\SFV(P) \subseteq \range(f_p) = \dom(f_q).$

	Recall that $\FFV(R) \subseteq \SFV(P)$, so $\FFV(R) \subseteq \dom{f_q} = \dom{f_q' \oplus \unit_T}$.

	As a corollary, we have $\dom(f_r') \subseteq \FFV(R) \subseteq \dom(f_q' \oplus \unit_T) \subseteq \dom(v)$, and $\dom(f_r') \subseteq \FFV(R) \subseteq \dom(f_q' \oplus \unit_T)$.
	Then,
	\begin{align*}
		f_q \odot f_r &=  \big( (f_q' \oplus \unit_T) \odot v \big) \odot \big( (f_r' \oplus \unit_S) \odot u \big) \\
				&=  (f_q' \oplus \unit_T) \odot \big( v \odot (f_r' \oplus \unit_S) \big) \odot u \tag{By standard associativity of $\odot$}\\
				&=  (f_q' \oplus \unit_T) \odot ( f_r' \oplus v  ) \odot u
				\tag{By~\cref{odot2oplus} and $\dom(f_r') \subseteq \dom(v)$} \\
				&=  (f_q' \oplus \unit_T) \odot ( (f_r' \odot \unit_{\range(f_r')}) \oplus ( \unit_{\dom(v)} \odot v) \odot u\\
				&=  (f_q' \oplus \unit_T) \odot ( f_r' \oplus \unit_{\dom(v)}) \odot (\unit_{\range(f_r')} \oplus v) \odot u \tag{$\heartsuit$} \\
				&= ((f_q' \oplus \unit_T) \oplus f_r') \odot ( v \oplus \unit_{\range(f_r')}) \odot u
				\tag{$\dagger$} \\
				&= ((f_q' \oplus \unit_T) \odot v) \oplus (f_r' \odot \unit_{\range(f_r')}) \odot u \tag{$\heartsuit$} \\
				&= f_q \oplus f_r
	\end{align*}
	where $\dagger$ follows from ~\cref{odot2oplus}, $\dom(f_r') \subseteq \dom(f_q' \oplus \unit_T)$ and exact commutativity,
	$\heartsuit$ follows from ~\cref{revexeq} and~\cref{exch_revex}.

	Thus,
	\(f_q \odot f_r \models Q \sepand R \).
	And by satisfaction rules,
	\begin{align*}
		f 	\models P \depand (Q \sepand R)
	\end{align*}
\item [\ref{ax:dep2sep2}]
	We want to show that under the special condition $\FFV(Q) = \emptyset$,
	$f \models P \depand Q$ implies that $f \models P \sepand Q$.

					If $f \models P \depand Q$, then there exists $f_p, f_q$ such that $f_p \odot f_q = f$ and $f_p \models P$, $f_q \models Q$.

By restriction property~\cref{thm:restriction}, 	$f_q \models Q$ implies that there exists $f_q' \sqsubseteq f_q$ such that $\SFV(Q) \subseteq \range(f_q') \subseteq \FV(Q)$ and $\dom(f_q') \subseteq\FFV(Q)$. $f_q' \sqsubseteq f_q$ so there exists $v, T$ such that $f_q = (f_q' \oplus_k \unit_T) \odot v$.

Since $\dom(f_q') \subseteq\FFV(Q) $ and $\FFV(Q) = \emptyset$,
it must $\dom(f_q') = \emptyset$, and thus no matter what the domain of $f_p$ is,
$\dom(f_q') \subseteq \dom(f_p)$.
Thus,
\begin{align*}
	f_p \odot f_q &= f_p \odot (f_q' \oplus \unit_T) \odot v \\
			&= (f_p \oplus f_q') \oplus v
			\tag{By~\cref{odot2oplus} and $\dom(f_q') \subseteq \dom(f_p)$}
\end{align*}
Thus, $f_p \oplus f_q' \sqsubseteq f_p \odot f_q = f$.
By satisfaction rules, $f_p \models P$ and $f_q' \models Q$ implies that
$f_p \oplus f_q' \models P \sepand Q$. Thus, by persistence,
$f \models P \sepand Q$

\item [\ref{ax:padding}]
	We want to show that when $ P \depand Q$, $P \depand (Q \sepand \pair{S}{S})$
	are both in $\Frestrict$, $f \models P \depand Q$ implies $f \models P \depand (Q \sepand \pair{S}{S})$.

One key guarantee we rely on from the grammar of $\Frestrict$ is that
	\[\FFV(Q) \cup S = \FFV (Q \sepand \pair{S}{S})  \subseteq \SFV(P). \]

				When $f \models P \depand Q$, there exists $f_p, f_q$ such that $f_p \odot f_q = f$ and $f_p \models P$, $f_q \models Q$,

	By~\cref{thm:restriction}, $f_p \models P$ implies that there exists $f_p' \sqsubseteq f_p$ such that $\SFV(P) \subseteq \range(f_p') \subseteq \FV(P)$,
	$\dom(f_p') \subseteq F\FV(P),\, \text{and } f_p' \models P. $
	By the fact that $f_p \odot f_q$ is defined,
	and that the definition of preorder in our concrete models, $f_p' \sqsubseteq f_p$ implies
	\[ \dom(f_q) = \range(f_p) \supseteq \range(f_p') \supseteq \SFV(P) \supseteq S \]
	Since $f_q$ preserves input,  $S \subseteq \dom(f_q)$ implies that $f_q = f_q \oplus \unit_{S}$, and thus \(f_p \odot f_q = f_p \odot (f_q \oplus \unit_{S}) \).

	Note that $\unit_S \models \kp{\Exact{S}}{\Dist[S]}$, and $f_q \models Q$.
	Thus, $f_q \oplus \unit_S \models Q \sepand \kp{S}{\Dist[S]}$.
	Since $f_p \models P$, it follows that
	\[f_p \odot (f_q \oplus \unit_S) \models P \depand (Q \sepand \kp{S}{\Dist[S]})\]
	Since	$f = f_p \odot f_q = f_p \odot (f_q \oplus \unit_{S})$,
	\[f \models P \depand (Q \sepand \kp{S}{\Dist[S]})\]
		\item [\ref{ax:exch}]
			We want to show that
			when $	(P \sepand Q) \depand (R \sepand S) $ and $(P \depand R) \sepand (Q \depand S)$ are both formula in $\Frestrict$,
			 $f \models 	(P \sepand Q) \depand (R \sepand S) $ implies
$f \models (P \sepand R) \sepand (Q \sepand S)$.

The key properties that being in $\Frestrict$ guarantees us are that
\begin{align*}
	&\FFV(R) \subseteq \SFV(P) \qquad \FFV(S) \subseteq \SFV(Q) \\
	&\FFV(R \sepand S) = \FFV(R) \cup \FFV(S) \subseteq \SFV(P \sepand Q) = \SFV(P) \cup \SFV(Q)
\end{align*}
			If $f \models (P \sepand Q) \depand (R \sepand S) $, then there exists $f_1, f_2$ such that $f_1 \odot f_2 = f$, $f_1 \models P \sepand Q$, $f_2 \models R \sepand S$. That is, there exist $u_1, v_1$ such that $u_1 \oplus v_1 \sqsubseteq f_1$, $u_1 \models P$, and $v_1 \models Q$; there exist $u_2, v_2$ such that $u_2 \oplus v_2 \sqsubseteq f_2$, $u_2 \models R$, $v_2 \models S$.

		By~\cref{thm:restriction},
		\begin{itemize}
			\item $u_1 \models P$ implies there exists $u_1' \sqsubseteq u_1$ such that $\SFV(P) \subseteq \range(u_1') \subseteq \FV(P)$, $\dom(u_1') \subseteq \FFV(P)$, and $u_1' \models P. $

			\item $v_1 \models Q$ implies there exists $v_1' \sqsubseteq v_1$ such that $\SFV(Q) \subseteq \range(v_1') \subseteq \FV(Q)$, $\dom(v_1') \subseteq \FFV(Q)$, and $v_1' \models Q. $

			\item $u_2 \models R$ implies there exists $u_2' \sqsubseteq u_2$ such that $\SFV(R) \subseteq \range(u_2') \subseteq \FV(R)$, $\dom(u_2') \subseteq \FFV(R)$, and $u_2' \models R. $

			\item $v_2 \models S$ implies there exists $v_2' \sqsubseteq v_2$ such that $\SFV(S) \subseteq \range(v_2') \subseteq \FV(S)$, $\dom(v_2') \subseteq \FFV(S)$, and $v_2' \models S. $
		\end{itemize}
			By Downwards closure property of $\oplus$, $u_2' \oplus v_2'$ is defined and $u_2' \oplus v_2' \sqsubseteq u_2 \oplus v_2 \sqsubseteq f_2$.
		Say that $f_1 =  (u_1 \oplus v_1 \oplus \unit_{S_1}) \odot h_1$, $f_2 =  (u_2' \oplus v_2' \oplus \unit_{S_2}) \odot h_2$.
	Also,
		\begin{align*}
			\dom(u_2' \oplus v_2') &= \dom(u_2') \cup \dom(v_2') \subseteq \FFV(R) \cup \FFV(S)  \subseteq \SFV(P) \cup \FFV(Q) \\
&\subseteq \range(u_1') \cup \range(v_1') \subseteq \range(u_1) \cup \range(v_1) = \range(u_1 \oplus v_1)
		\end{align*}

		Then
		\begin{align*}
			f_1 \odot f_2 &=(u_1 \oplus v_1 \oplus \unit_{S_1}) \odot h_1 \odot (u_2' \oplus v_2' \oplus \unit_{S_2}) \odot h_2 \\
					&= (u_1 \oplus v_1 \oplus \unit_{S_1}) \odot ((u_2' \oplus v_2' ) \oplus h_1) \odot h_2
					\tag{$\heartsuit$} \\
					&= (u_1 \oplus v_1 \oplus \unit_{S_1}) \odot ((u_2' \oplus v_2') \odot \unit_{\range(u_2' \oplus v_2')}) \oplus (\unit_{\dom(h_1)} \odot h_1) \odot h_2 \\
					&= (u_1 \oplus v_1 \oplus \unit_{S_1}) \odot (u_2' \oplus v_2' \oplus \unit_{\dom(h_1)})
					 \odot (\unit_{\range(u_2' \oplus v_2')} \oplus h_1) \odot h_2 \tag{$\dagger$} \\
					&= (u_1 \oplus v_1 \oplus \unit_{S_1}) \odot (u_2' \oplus v_2' \oplus \unit_{\range(u_1 \oplus v_1)} \oplus \unit_{S_1})
					\odot ( \unit_{\range(u_2' \oplus v_2')} \oplus h_1) \odot h_2  \\
					&= \big( ((u_1 \oplus v_1) \odot (u_2' \oplus v_2' \oplus \unit_{\range(u_1 \oplus v_1)} )) \oplus \unit_{S_1}\big) \odot (\unit_{\range(u_2' \oplus v_2')} \oplus h_1) \odot h_2 \tag{$\dagger$} \\
					&= \big((u_1 \odot (u_2' \oplus \unit_{\range(u_1)})) \oplus (v_1 \odot (v_2' \oplus \unit_{\range(v_1)})) \oplus \unit_{S_1} \big) \\
					&\odot  (\unit_{\range(u_2' \oplus v_2')} \oplus h_1) \odot h_2  \tag{$\dagger$ and exact commutativity, associativity}
		\end{align*}
		where		$\heartsuit$ follows from~\cref{odot2oplus}, $\dom(u_2' \oplus v_2') \subseteq \range(u_1 \oplus v_1) \subseteq \dom(h_1)$,
		and  $\dagger$ follows from~\cref{revexeq} and~\cref{exch_revex}.

		Thus, $ (u_1 \odot (u_2' \oplus \unit_{\range(u_1)})) \oplus (v_1 \odot (v_2' \oplus \unit_{\range(v_1)})) \sqsubseteq f_1 \odot f_2$.
 Recall that $u_2' \models R$. By persistence, $u_2' \oplus \unit_{\range(u_1)} \models R$. Similarly, $v_2' \models S$, so by persistence, $v_2' \oplus \unit_{\range(v_1)} \models S$. 	Therefore,
		\[ (u_1 \odot (u_2' \oplus \unit_{\range(u_1)})) \oplus (v_1 \odot (v_2' \oplus \unit_{\range(v_1)})) \models (P \depand R) \sepand (Q \depand S)\]

		Then, by persistence, $f \models  (P \depand R) \sepand (Q \depand S)$.

	\end{description}
\end{proof}

\APAx*
\begin{proof}
	We prove it one by one.
	\begin{description}[leftmargin=*]
		\item[\ref{ax:revpar}]
			Given any $f \models \kp{S}{\Dist[A] \sepand \Dist[B]}$, by satisfaction rules and semantic of atomic propositions, there exists
	$f' \sqsubseteq f$ such that for all $m \in M_d$ such that $m \models_d \Exact{S}$, $f'(m) \models_r \Dist[A] \sepand \Dist[B] $.

	Since $f'(m)$ is defined and  $f'(m) \models_r \Dist[A] \sepand \Dist[B] $,
	it follows that $\dom(f') = S$ and $\range(f') \supseteq S \cup A \cup B$.
	Thus, we can define $f_1 = \pi_{S \cup A} f'$, $f_2 = \pi_{S \cup B} f'$.
	Note that  $f_1 \models \pair{S}{A}$, $f_2 \models \pair{S}{B}$.
	Also, because $A \cap B \subseteq S$,
	\[\range(f_1) \cap \range(f_2) = ( S \cup A) \cap (S \cup B) = S, \]
and thus $f_1 \oplus f_2$ is defined.
We now want to show that $f_1 \oplus f_2 \sqsubseteq f$.

Note	$f'(m) \models_r \Dist[A] \sepand \Dist[B]$ implies that
	there exists $\mu_1, \mu_2$ such that $\mu_1 \oplus_r \mu_2 \sqsubseteq f'(m)$, and  $dom(\mu_1) \supseteq A$, $dom(\mu_2) \supseteq B$.
	Since $f'$ preserves input on its domain $S$, $\pi_S f'(m) = \unit(m)$, so
	$(\mu_1 \oplus_r \unit(m)) \oplus_r (\mu_2 \oplus_r \unit(m)) \sqsubseteq
	f'(m) \oplus_r \unit(m) \oplus_r \unit(m)= f'(m)$ too.
	Let $\mu_1' = \pi_{A \cup S} (\mu_1 \oplus_r \unit(m)) $ and $\mu_2' = \pi_{B \cup S} (\mu_2 \oplus_r \unit(m)) $. Then due to Downwards closure in $M_d$,
	$\mu_1' \oplus_r \mu_2'$ will also be defined, and
	\[\mu_1' \oplus_r \mu_2' \sqsubseteq (\mu_1 \oplus_r \unit(m)) \oplus_r (\mu_2 \oplus_r \unit(m)) \sqsubseteq f'(m) ,\]
	which implies that $\mu_1' \oplus_r \mu_2'= \pi_{S \cup A \cup B} f'(m)$.
In the range model, this means that
$\mu_1' = \pi_{S \cup A} f'(m)$, $\mu_2' = \pi_{S \cup B} f'(m)$.

Then
	for any $m' \in \R{S}$, any $r \in \R{A \cup B \cup S}$,
	\begin{align*}
		(\pi_{S \cup A \cup B} f') (m')(r) &= (\pi_{S \cup A \cup B} f'(m')) (r) = \mu_1' \oplus_r \mu_2'(r) =\mu_1'(r^{S \cup A}) \cdot \mu_2'(r^{S \cup B}) \\
		(f_1 \oplus f_2) (m') (r)& = f_1(m')(r^{S \cup A}) \cdot f_2(m')(r^{S \cup B})\\&=  (\pi_{S \cup A} f') (m') (r^{S \cup A}) \cdot (\pi_{S \cup B} f' (m') (r^{S \cup B}) \\
																											&=  \mu_1' (r^{S \cup A}) \cdot \mu_2'(r^{S \cup B}) \
	\end{align*}

Thus, $f_1 \oplus f_2 = 		\pi_{S \cup A \cup B} f'$, which implies that $f_1 \oplus f_2 \sqsubseteq f$.
By their types,  $f_1 \oplus f_2  \models \pair{S}{A} \sepand \pair{S}{B}$.

By persistence, $f \models \pair{S}{A} \sepand \pair{S}{B}$.

		\item[\ref{ax:unionran}] Obvious from the semantics of atomic proposition and the range logic.
		\item[\ref{ax:atomseq}]
				Given any $f \models \pair{A}{B} \depand \pair{B}{C}$, by satisfaction rules and semantic of atomic propositions, there exists
		\begin{itemize}
			\item $f_1, f_2$ such that $f_1 \odot f_2 = f$;
			\item $f_1' \sqsubseteq f_1$ such that for any $m \in M_d$ such that $m \models_d \Exact{A}$, $f_1' (m) \models_r \Dist[B]$.
			\item $f_2' \sqsubseteq f_2$ such that for any $m \in M_d$ such that $m \models_d \Exact{B}$, $f_2' (m) \models_r \Dist[C]$.
		\end{itemize}
		Note that $f_1' (m) \models_r \Dist[B]$ implies that $B \subseteq \range(f'_1)$, so
		$ \pi_{B} f_1'$ is defined. Let $f_1 '' = \pi_{B} f_1'$.

		Note that for any $m \in M_d$ such that $m \models_d \Exact{A}$, $f_1'' (m) \models_r \Dist[B]$ too, so $f'' \models \pair{A}{B}$ too. Also, by transitivity, $f_1'' \sqsubseteq f_1' \sqsubseteq f_1$.

		Say $f_1 = (f_1'' \oplus \eta_{S_1}) \odot v_1$, $f_2 = (f_2' \oplus \eta_{S_2}) \odot v_2$, then since $\range(f_1'') = B = \dom(f_2')$,
		\begin{align*}
			f_1 \odot f_2 &= (f_1 '' \oplus \eta_{S_1}) \odot v_1 \odot (f_2' \oplus \eta_{S_2}) \odot v_2 \\
					&= (f_1 '' \oplus \eta_{S_1}) \odot ( f_2' \oplus v_1) \odot v_2 \tag{By~\cref{odot2oplus} and $\dom(f_2') = B = \range(f_1'') \subseteq \dom(v_1)$} \\
					&= (f_1 '' \oplus \eta_{S_1}) \odot (f_2' \oplus \eta_{\dom(v_1)}) \odot (v_1 \oplus \eta_{\range(f_1)}) \odot v_2 \tag{By~\cref{oplus2odot}}\\
					&= (f_1 '' \oplus \eta_{S_1}) \odot (f_2' \oplus \eta_{S}) \odot (v_1 \oplus \eta_{\range(f_1)}) \odot v_2 \\
					&= ((f_1 '' \odot f_2') \oplus \eta_{S_1}) \odot (v_1 \oplus \eta_{\range(f_1)}) \odot v_2
		\end{align*}
		So $f_1'' \odot f_2' \sqsubseteq f_1 \odot f_2 = f$.

		$f_1''\colon \R{A} \to \DR{B}$, $f_2' \colon \R{B} \to \DR{\range(f_2')}A$, so $f_1'' \odot f_2' \colon \R{A} \to \DR{\range(f_2')}$.
		Since $\range(f_2') \supseteq C$, it follows that $f_1'' \odot f_2' \models \pair{A}{C}$, and thus $f \models \pair{A}{C}$ by persistence.
		\item[\ref{ax:repeatdom}]
			If $f \models \pair{A}{B}$, then there must exists $f' \sqsubseteq f$ such that for all $m \in M_d$ such that $m \models \Exact{A}$, $f'(m) \models_r \Dist[B]$.

			Given any witness $f'$, $f' = \unit_{\R{A}} \odot f'$, and also $f' \models_r \pair{A}{B}$.

			Note that $\unit_{\R{A}} \models_r \pair{A}{A}$,
			so $f' = \unit_{\R{A}} \odot f' \models \pair{A}{A} \depand \pair{A}{B}$.

		\item[\ref{ax:repeatran}]
			Analogous as the~\ref{ax:repeatdom} case, except that now using the fact $f' = f' \odot \unit_{\R{B}}$ for any $f' \colon \R{A} \to \DR{B}$.
	\end{description}
\end{proof}

            \subsection{Common properties of models \MD and \MP}
            \label{app:common}
            We define a more general class of models, parametric on a monad $\T$, which encompasses both our concrete models $\MP$ and $\MD$. We will call them $\T$-models and use their properties to simplify proofs of certain  properties of $\MD$ and $\MP$.
\begin{definition}[$\T$-models]\label{def:tmodel}
	We say that $(M, \sqsubseteq,\oplus,\odot, M)$ is a $\T$-model if it  satisfies the following conditions.
            \begin{enumerate}
              \item $M$ consists of all maps of the type $\m{S} \rightarrow \TR{S \cup U}$, where $S,U $ are finite subsets of $\Var$.
              \item All $m\in M$ preserve the input $m\colon  \R{S} \rightarrow \TR{S \cup U}$ is in $M$ only if $\pi_{S} m = \unit_{S}$;
              \item $\odot$ is defined to be the Kleisli composition associated with $\T$;
              \item $\oplus$ is deterministic and partial: $f \oplus g$ is defined when $\range(f) \cap \range(g)  =\dom(f) \cap \dom(g)$;
              \item $\oplus$ satisfies standard associativity: when both  $(f \oplus g) \oplus h$ and $f \oplus (g \oplus h)$ are defined, $(f \oplus g) \oplus h = f \oplus (g \oplus h)$;
		\item When $f \oplus g$ are $g \oplus f$ are both defined, $f \oplus g = g \oplus f$.
              \item For any $f\colon  \R{A} \rightarrow \TR{A \cup X} \in M$, and any $S \subseteq A$,
                \begin{align}
					f \oplus \unit_{S} = f 	\tag{Padding equality}\label{paddingdom}
                .\end{align}
              \item When both $(f_1 \oplus f_2) \odot (f_3 \oplus f_4)$ and $(f_1 \odot f_3) \oplus (f_2 \odot f_4)$ are defined,
                \begin{align}
                  \label{revexeq}
                  (f_1 \oplus f_2) \odot (f_3 \oplus f_4) &=(f_1 \odot f_3) \oplus (f_2 \odot f_4) \tag{Exchange equality}
                \end{align}
			\item $M$ is closed under $\oplus$ and $\odot$;
              \item For $ f, g \in M$, $f \sqsubseteq g$ if and only if there exist $v \in M$ and some finite set $S$ such that,
                \begin{align}
					g = (f \oplus \unit_{S}) \odot v \label{M:preorder}
                \end{align}
            \end{enumerate}
					\end{definition}

Below, we prove properties $\T$-models, which would be common properties of $\MD$ and $\MP$. Two main results are  that all $\T$-models are \LOGIC frames (\Cref{Mframeaxioms}).

\begin{lemma}[Standard associativity of $\oplus$]
	\label{standardassoc}
	For any $f_1, f_2, f_3 \in M$, $(f_1 \oplus f_2) \oplus f_3$ is defined if and only if $f_1 \oplus (f_2 \oplus f_3)$ is defined and they are equal.
\end{lemma}
\begin{proof}
	$(f_1 \oplus f_2) \oplus f_3$ is defined if and only if $\rg{1} \cap \rg{2} = \dm{1} \cap \dm{2}$ and $(\rg{1} \cup \rg{2}) \cap \rg{3} = (\dm{1} \cup \dm{2}) \cap \dm{3}$.

	$f_1 \oplus (f_2 \oplus f_3)$ is defined if and only if
	$\rg{2} \cap \rg{3} =  \dm{2} \cap \dm{3}$ and $\rg{1} \cap (\rg{2} \cup \rg{3}) = \dm{1} \cap (\dm{2} \cup \dm{3})$.
	Thus, to show that $(f_1 \oplus f_2) \oplus f_3$ is defined
	if and only if $f_1 \oplus (f_2 \oplus f_3)$ is defined,
	it suffices to show that
	\begin{align}
		\rg{1} \cap \rg{2} &= \dm{1} \cap \dm{2}  \label{concrete:cond1}\\
		(\rg{1} \cup \rg{2}) \cap \rg{3} &= (\dm{1} \cup \dm{2}) \cap \dm{3} \label{concrete:cond2}
	\end{align}
	if and only if
	\begin{align}
		\rg{2} \cap \rg{3} &= \dm{2} \cap \dm{3} \label{concrete:cond3} \\
		\rg{1} \cap (\rg{2} \cup \rg{3}) &= \dm{1} \cap (\dm{2} \cup \dm{3}) \label{concrete:cond4}
	\end{align}

		We show that \cref{concrete:cond3} and \cref{concrete:cond4} follows from
		\cref{concrete:cond1} and \cref{concrete:cond2}:

	Recall that $\dm{1} \subseteq \rg{1}$, $\dm{2} \subseteq \rg{2}$, $\dm{3} \subseteq \rg{3}$, so
	\begin{itemize}
		\item \cref{concrete:cond3} follows from $\dm{2} \cap \dm{3} \subseteq \rg{2} \cap \rg{3}$ and $\dm{2} \cap \dm{3} \supseteq \rg{2} \cap \rg{3}$,
			which holds because
			\begin{align}
				\rg{2} \cap \rg{3} &= \rg{2} \cap (\rg{2} \cap \rg{3}) \subseteq \rg{2} \cap ((\rg{1} \cup \rg{2}) \cap \rg{3}) \notag \\
					&= \rg{2} \cap ((\dm{1} \cup \dm{2}) \cap \dm{3}) = \rg{2} \cap (\dm{1} \cap \dm{3}) \tag{By \cref{concrete:cond2}} \\
					&\subseteq (\rg{2} \cap \dm{1}) \cap \dm{3} \subseteq (\rg{2} \cap \rg{1}) \cap \dm{3} \tag{By $\dm{1} \subseteq \rg{1}$}\\
					&= (\dm{2} \cap \dm{1}) \cap \dm{3} \subseteq \dm{2} \cap \dm{3} \tag{By \cref{concrete:cond1}}
			\end{align}

		\item \cref{concrete:cond4} follows from $(\dm{1} \cup \dm{2}) \cap \dm{3} \subseteq (\rg{1} \cup \rg{2} ) \cap \rg{3}$ and
			$(\dm{1} \cup \dm{2}) \cap \dm{3} \supseteq (\rg{1} \cup \rg{2} ) \cap \rg{3}$, which holds because
			\begin{align*}
				\rg{1} \cap (\rg{2} \cup \rg{3}) &= (\rg{1} \cap \rg{2}) \cup (\rg{1} \cap \rg{3}) \subseteq (\rg{1} \cap \rg{2}) \cup (\rg{1} \cap (\rg{1} \cup \rg{2})\cap \rg{3}) \\
	&= (\dm{1} \cap \dm{2}) \cup (\rg{1} \cap (\dm{1} \cup \dm{2}) \cap \dm{3}) \tag{By \cref{concrete:cond1} and \cref{concrete:cond2}} \\
	&= (\dm{1} \cap \dm{2}) \cup ((\rg{1} \cap \dm{1} \cap \dm{3}) \cup (\rg{1} \cap \dm{2} \cap \dm{3}))  \\
	&\subseteq (\dm{1} \cap \dm{2}) \cup ((\dm{1} \cap \dm{3}) \cup (\rg{1} \cap \rg{2} \cap \dm{3})) \tag{By $\dm{2} \subseteq \rg{2}$} \\
	&\subseteq (\dm{1} \cap \dm{2}) \cup ((\dm{1} \cap \dm{3}) \cup (\dm{1} \cap \dm{2} \cap \dm{3})) \tag{By \cref{concrete:cond1}} \\
	&\subseteq (\dm{1} \cap \dm{2}) \cup (\dm{1} \cap \dm{3}) = \dm{1} \cap (\dm{2} \cup \dm{3})
			\end{align*}
	\end{itemize}

	We show that \cref{concrete:cond1} and \cref{concrete:cond2} follows from
		\cref{concrete:cond3} and \cref{concrete:cond4}:

	\begin{itemize}
		\item \cref{concrete:cond1} follows from $\dm{1} \cap \dm{2} \subseteq \rg{1} \cap \rg{2}$ and $\dm{1} \cap \dm{2} \supseteq \rg{1} \cap \rg{2}$,
			which holds because
			\begin{align*}
				\rg{1} \cap \rg{2} &= \rg{1} \cap (\rg{2} \cup \rg{3}) \cap \rg{2} = \dm{1} \cap (\dm{2} \cup \dm{3}) \cap \rg{2} \tag{By \cref{concrete:cond3}} \\
					&= \dm{1} \cap ((\dm{2} \cap \rg{2}) \cup (\dm{3} \cap \rg{2})) = \dm{1} \cap (\dm{2} \cup (\dm{3} \cap \rg{2})) \\
					&\subseteq \dm{1} \cap (\dm{2} \cup (\rg{1} \cap \rg{2})) \tag{By $\dm{2} \subseteq \rg{1}$} \\
					& = \dm{1} \cap (\dm{2} \cup (\dm{1} \cap \dm{2})) \tag{By \cref{concrete:cond3}} \\
					& = \dm{1} \cap \dm{2} \end{align*}
				\item \cref{concrete:cond2} follows from $(\dm{1} \cup \dm{2}) \cap \dm{3} \subseteq (\rg{1} \cup \rg{2}) \cap \rg{3}$ and
					$(\dm{1} \cup \dm{2}) \cap \dm{3} \supseteq (\rg{1} \cup \rg{2}) \cap \rg{3}$, which holds because
					\begin{align*}
(\rg{1} \cup \rg{2}) \cap \rg{3} &= (\rg{1} \cap \rg{3}) \cup (\rg{2} \cap \rg{3}) \\
			&= (\rg{1} \cap (\rg{2} \cup \rg{3}) \cap \rg{3}) \cup (\rg{2} \cap \rg{3}) \\
			&= (\dm{1} \cap (\dm{2} \cup \dm{3}) \cap \rg{3}) \cup (\dm{2} \cap \dm{3}) \tag{By \cref{concrete:cond4}}\\
			&= (\dm{1} \cap ((\dm{2} \cap \rg{3}) \cup (\dm{3} \cap \rg{3}))) \cup (\dm{2} \cap \dm{3}) \\
			&\subseteq (\dm{1} \cap ((\rg{2} \cap \rg{3}) \cup \dm{3})) \cup (\dm{2} \cap \dm{3}) \tag{By $\dm{2} \subseteq \rg{2}$, $\dm{3} \subseteq \rg{3}$} \\
			&=  (\dm{1} \cap ((\dm{2} \cap \dm{3}) \cup \dm{3})) \cup (\dm{2} \cap \dm{3}) \tag{By \cref{concrete:cond3}} \\
			&=  (\dm{1} \cap \dm{3}) \cup (\dm{2} \cap \dm{3}) = (\dm{1} \cup \dm{2}) \cap \dm{3}       \end{align*}
				\end{itemize}

				Thus, \cref{concrete:cond1} and \cref{concrete:cond2} hold if and only if \cref{concrete:cond3} and \cref{concrete:cond4} hold. Therefore, $(f_1 \oplus f_2) \oplus f_3$ is defined if and only if $f_1 \oplus (f_2 \oplus f_3)$ is defined and by \Cref{def:tmodel}(5) they are equal.
\end{proof}

\begin{lemma} [Reflexivity and transitivity of order]
\label{Md:order}
For any $\T$-model $M$,	the order $\sqsubseteq$ defined in $M$ is transitive and reflexive.
\end{lemma}
\begin{proof}
Let $x \colon \R{A} \rightarrow \TR{X} \in M$, $S = \emptyset$,  $v = \unit_{X}$.  Then
	\begin{align*}
	(x \oplus \unit_{S}) \odot v &= 	(x \oplus \unit_{\emptyset}) \odot \unit_{X}\\
																														&= x \odot \unit_{X} \tag{By~\cref{paddingdom}}\\
																														&= x \tag{By~\cref{def:tmodel}(3)}
	\end{align*}
Thus, by \Cref{M:preorder} we have $x \sqsubseteq x$, and the order is reflexive.

For any $x, y , z \in M$, if $x \sqsubseteq y$ and $y \sqsubseteq z$, then by definition of $\sqsubseteq$, there exist $S_1$ and  $v_1$ such that
\(	y = (x \oplus \unit_{S_1}) \odot v_1 \), and there exist $S_2$ and  $v_2$ such that \(z = (y \oplus \unit_{S_2}) \odot v_2\).

We can now calculate:
	\begin{align*}
		z   &= (y \oplus \unit_{S_2}) \odot v_2\\
		&= (((x \oplus \unit_{S_1}) \odot v_1) \oplus \unit_{S_2}) \odot v_2 \\
		&= (((x \oplus \unit_{S_1}) \odot v_1) \oplus (\unit_{S_2} \odot \unit_{S_2})) \odot v_2 \\
		&= (x \oplus \unit_{S_1} \oplus \unit_{S_2}) \odot (v_1 \oplus \unit_{S_2}) \odot v_2 \tag{By~\ref{revexeq} and~\cref{exch_revex}} \\
		&= (x \oplus \unit_{S_1 \cup S_2}) \odot ((v_1 \oplus \unit_{S_2}) \odot v_2)
		\end{align*}
 $\M$ is closed under $\oplus$, $\odot$, so $(v_1 \oplus \unit_{S_2}) \odot v_2 \in \M$. Thus, we can instantiate \Cref{M:preorder} with $S=S_1 \cup S_2$ and $v=(v_1 \oplus \unit_{S_2}) \odot v_2$ obtaining $x \sqsubseteq z$. So the order is transitive.
\end{proof}

            \begin{proposition}
              \label{exch_revex}
		For any $\T$-model M, states \(f_1, f_2, f_3, f_4 \) in $M$,
               $(f_1 \odot f_3) \oplus (f_2 \odot f_4)$ is defined
               implies $(f_1 \oplus f_2) \odot (f_3 \oplus f_4)$ is also defined.
               The converse does not always hold, but if $f_1 \odot f_3$
               and $f_2 \odot f_4$ are defined, then $(f_1 \oplus f_2) \odot (f_3
               \oplus f_4)$ is defined implies $(f_1 \odot f_3) \oplus (f_2 \odot
               f_4)$ is defined too.
            \end{proposition}
\begin{proof}
	We prove each direction individually:
	\begin{itemize}
		\item Given $(f_1 \odot f_3) \oplus (f_2 \odot f_4)$ is defined, it must that $\rg{1} = \dm{3},\, \rg{2} = \dm{4},$ and $\rg{3} \cap \rg{4} = \dm{1} \cap \dm{2}$.
			Thus, $\rg{1} \cap \rg{2} = \dm{3} \cap \dm{4} \subseteq \rg{3} \cap \rg{4} = \dm{1} \cap \dm{2}$, ensuring that $f_1 \oplus f_2$ is defined;
			\newline
			$\rg{3} \cap \rg{4} = \dm{1} \cap \dm{2} \subseteq \rg{1} \cap \rg{2} = \dm{3} \cap \dm{4}$, ensuring that $f_3 \oplus f_4$ is defined;
			\newline
			$\range(f_1 \oplus f_2) = \rg{1} \cup \rg{2} = \dm{3} \cup \dm{4} = \dom(f_3 \oplus f_4)$, ensuring $(f_1 \oplus f_2) \odot (f_3 \oplus f_4)$ is defined.

		\item Given $f_1 \odot f_3$ and $f_2 \odot f_4$ are defined, $(f_1 \odot f_3) \oplus (f_2 \odot f_4)$ is defined if $\rg{3} \cap \rg{4} = \dm{1} \cap \dm{2}$. When $(f_1 \oplus f_2) \odot (f_3 \oplus f_4)$ is defined,
			\begin{align*}
				\rg{3} \cap \rg{4} &= \dm{3} \cap \dm{4} \tag{Because $f_3 \oplus f_4$ is defined}\\
					&= \rg{1} \cap \rg{2} \tag{Because $f_1 \odot f_3$ and $f_2 \odot f_4$ are defined}\\
					&= \dm{1} \cap \dm{2} \tag{Because $f_1 \oplus f_2$ is defined}
			\end{align*}
			So $(f_1 \odot f_3) \oplus (f_2 \odot f_4)$ is also defined. \qedhere
	\end{itemize}
\end{proof}

\begin{lemma}[ $\odot$ elimination]
\label{odot2oplus}
For any $\T$-model M, and $f, g \in M$, if $f \odot (g \oplus \unit_X) $ is defined and $\dom(g) \subseteq \dom(f)$, then
$f \odot (g \oplus \unit_X) = g \oplus f$.
\end{lemma}
\begin{proof}Let $f\colon \m{S} \to \TR{S\cup T}$ and $g\colon \m{U} \to \TR{U\cup V}$ be in $M$.
	When $U \subseteq S$,
\begin{align*}
&f \odot (g \oplus \unit_X)\\
  &= (f \oplus \unit_{U}) \odot (g \oplus \unit_X \oplus \unit_{S\cup T}) \tag{By~\ref{paddingdom}} \\
  &= (\unit_{U} \oplus f) \odot (g \oplus \unit_X \oplus \unit_{S\cup T}) \tag{By commutativity} \\
		&= (\unit_{U} \oplus f) \odot (g \oplus \unit_{S\cup T}) \tag{$\dagger$}\\
 &= (\unit_{U} \odot g) \oplus (f \odot  \unit_{S\cup T})\tag{By~\cref{exch_revex} and~\ref{revexeq}}\\
 &= g \oplus f \qedhere
\end{align*}
where $\dagger$ follows from $X \subseteq S \cup T$, which holds as
$f \odot (g \oplus \unit_X)$ defined implies $S \cup T = X \cup U$.
\end{proof}

\begin{lemma}[Converting $\oplus$ to $\odot$]
	\label{oplus2odot}
For any $\T$-model M, let $f\colon \m{S} \to \TR{S\cup T}$ and $g\colon \m{U} \to \TR{U\cup V}$ be in $M$. If $f \oplus g$ is defined, then $f \oplus g = (f \oplus \unit_{U}) \odot (\unit_{S\cup T} \oplus g)$.
\end{lemma}
\begin{proof}
 \begin{align*}
 f \oplus g &= (f \odot \unit_{S\cup T}) \oplus (\unit_{U} \odot g) \\
 		&=(f \oplus \unit_{U}) \odot (\unit_{S\cup T} \oplus g) \tag{By~\cref{exch_revex} and~\ref{revexeq}}
 \end{align*}
\end{proof}

\begin{lemma}[Quasi-Downwards-closure of $\odot$]
	\label{odotdownwards}
	For any $\T$-model M, and \(f, g, h, i \in \M\),
	if $f \sqsubseteq h$, $g \sqsubseteq i$, and $f \odot g$, $h \odot i$ are all defined, then $f \odot g \sqsubseteq h \odot i$.
\end{lemma}
\begin{proof}
	Since $f \sqsubseteq h$, $g \sqsubseteq i$, there must exist sets $S_1, S_2$ and $v_1, v_2 \in M$ such that $h = (f \oplus \unit_{S_1}) \odot v_1$,  $i = (g \oplus \unit_{S_2}) \odot v_2$.
	$f \odot g$ is defined, so $\dom(g) = \range(f) \subseteq \range(f \oplus \unit_{S_1}) = \dom(v_1)$.
	Thus,
	\begin{align*}
	h\odot i
	&= (f \oplus \unit_{S_1}) \odot v_1 \odot  (g \oplus \unit_{S_2}) \odot v_2 \\
	&=(f \oplus \unit_{S_1}) \odot (g \oplus v_1) \odot v_2 \tag{By~\cref{odot2oplus} and $\dom(g) \subseteq \dom(v_1)$}\\
	&=(f \oplus \unit_{S_1}) \odot (g \oplus \unit_{\dom(v_1)}) \odot (\unit_{\range(g)} \oplus v_1) \odot v_2 \tag{By~\cref{oplus2odot}}\\
	&=(f \oplus \unit_{S_1}) \odot (g \oplus \unit_{ S_1}) \odot (\unit_{\range(g)} \oplus v_1) \odot v_2 \tag{$\dagger$}\\
	&=((f  \odot g) \oplus (\unit_{S_1} \odot \unit_{ S_1})) \odot (\unit_{\range(g)} \oplus v_1) \odot v_2 \tag{$\heartsuit$}\\
	&=((f  \odot g) \oplus \unit_{ S_1}) \odot (\unit_{\range(g)} \oplus v_1) \odot v_2
	\end{align*}
	where $\dagger$ follows from  $\dom(g) = \range(f)$ and~\cref{paddingdom},
	and $\heartsuit$ follows from~\cref{exch_revex} and~\ref{revexeq}.

	Therefore, $f  \odot g \sqsubseteq h \odot i$.
\end{proof}

\begin{lemma}
\label{Mframeaxioms}
Any $\T$-model $M$ is in $\LOGIC$.
\end{lemma}
\begin{proof} The axioms that we need to check are the follows.
	\begin{description}
	\item [$\oplus$ Down-Closed]
	\label{downwards_oplus}
	We want to show that for any $x',x,y',y \in M$, if $x' \sqsubseteq x$ and $y'\sqsubseteq y$ and $x \oplus y = z$, then $x' \oplus y'$ is defined, and $x'\oplus y' = z' \sqsubseteq z$.

	Since $x' \sqsubseteq x$ and $y' \sqsubseteq y$, there exist sets $S_1, S_2$, and $v_1, v_2 \in M$ such that
	$x = (x'\oplus \unit_{S_1} ) \odot v_1$, and
	$y = (y' \oplus \unit_{S_2} ) \odot v_2$.
	Thus,
	\begin{align*}
		x \oplus y &= ((x' \oplus \unit_{S_1} ) \odot v_1) \oplus ((y' \oplus \unit_{S_2} ) \odot v_2)\\
		               &=  \big((x' \oplus \unit_{S_1}) \oplus (y' \oplus \unit_{S_2}) \big) \odot ( v_1 \oplus v_2) \tag{By~\cref{exch_revex} and~\ref{revexeq}}\\
			      &=  \big((x' \oplus y') \oplus (\unit_{S_1} \oplus \unit_{S_2})   \big) \odot ( v_1 \oplus v_2) \tag{By commutativity and associativity}\\
			       &=  \big((x' \oplus y') \oplus (\unit_{S_1\cup S_2})   \big) \odot ( v_1 \oplus v_2)
	\end{align*}
	This derivation proved that $x' \oplus y'$ is defined,
	and $x' \oplus y' \sqsubseteq x \oplus y = z$.

\item [($\odot$ Up-Closed)]
	We want to show that for any $z',z,x,y \in M$, if $z = x \odot y$ and  $z' \sqsupseteq z$, then there exists $x', y'$ such that $x' \sqsupseteq x$, $y' \sqsupseteq y$, and $z' = x' \odot y'$.

	Since $z' \sqsupseteq z$, there exist set $S$, and $v \in M$ such that
	$z' = (z \oplus \unit_{S} ) \odot v$. Thus,
	\begin{align*}
		z' &=  (z \oplus \unit_{S} ) \odot v \\
				&=  ((x \odot y) \oplus \unit_{S} ) \odot v \\
				&=  ((x \odot y) \oplus (\unit_{S} \odot \unit_{S})) \odot v \\
				&=  ((x \oplus \unit_{S}) \odot (y \oplus \unit_{S})) \odot v \tag{By~\cref{exch_revex} and~\ref{revexeq}} \\
				&=  (x \oplus \unit_{S}) \odot ((y \oplus \unit_{S}) \odot v) \tag{By standard associativity of $\odot$}
	\end{align*}
	Thus, for $x' = x \oplus \unit_S$ and $y' = (y \oplus \unit_S) \odot v$,
	$z' = x' \odot y'$.
\item [ ($\oplus$ Commutativity) ]
	We want to show that $z = x \oplus y$ implies that $z = y \oplus x$.
By definition of $T$-models:
first, $x \oplus y$ is defined iff $\range(x) \cap \range(y) = \dom(x) \cap \dom(y)$ iff $y \oplus x$ is defined;
second, when $x \oplus y$ and $y \oplus x$ are both defined, they are equal. Thus, $\oplus$ commutativity frame condition is satisfied.
\item [ ($\oplus$ Associativity) ]
			Since $\oplus$ is deterministic and partial,the associativity of $\oplus$ frame axiom reduces to ~\cref{standardassoc}.
\item [($\oplus$ Unit existence)] We want to show that for any $x \in M$, there exists $e \in E$ such that $x = e \oplus x$.
	For any $x: \R{A} \rightarrow \DR{B}$,  $x \oplus \unit_{\R{\emptyset}}$ is defined because $B \cap \emptyset = \emptyset = A \cap \emptyset$,
	and by~\cref{paddingdom}, $(x \oplus \unit_{\R{\emptyset}}) = x$.
	Also, $\unit_{\R{\emptyset}} \in E = M$.
	So $e = \unit_{\R{\emptyset}}$ serves as the unit under $\oplus$ for any $x$.
\item [($\oplus$ Unit Coherence)]
        We want to show that for any $y \in M$, $e \in E = M$,
								if $x = y \oplus e$,
        then $x \sqsupseteq y$.
				\begin{align*}
					x = y \oplus e &= (y \odot \unit_{\range(y)}) \oplus (\unit_{\dom(e)} \odot e) \\
																				&= (y \oplus \unit_{\dom(e)}) \odot (\unit_{\range(y)} \oplus e) \tag{By~\cref{revexeq}}\\
																				&= (y \oplus \unit_{\dom(e)}) \odot (e \oplus \unit_{\range(y)}) \tag{$\oplus$ Commutativity}
				\end{align*}
				Thus, $x \sqsupseteq y$.
\item[ ($\odot$ Associativity)]
	Since $\odot$ is deterministic and partial, the associativity of $\odot$ frame axiom reduces to the standard associativity.
	Kleisli composition satisfies standard associativity, so $\odot$ also satisfies standard associativity.

\item [($\odot$ Unit $\text{Existence}_{\text{L and R}}$)] Since $\odot$ is the Kleisli composition, for any morphism $x: \R{A} \rightarrow \DR{B}$, $\unit_{\R{A}}$ is the left unit, and $\unit_{\R{B}}$ is the right unit. For all $S$, $\unit_{\R{S}} \in M = E$. Thus, for any $x \in M$, there exists $e \in E$ such that $e \odot x = x$, and there exists $e' \in E$ such that $x \odot e' = x$.
\item [($\odot$ $\text{Coherence}_{R}$)]
For any $y \in M, e \in E = M$ such that $x = y \odot e$, we want to show that $x \sqsupseteq y$. We just proved that $(y \oplus \unit_{\R{\emptyset}})= y$ for any $y$, so $x = y \odot e = (y \oplus \unit_{\R{\emptyset}}) \odot e$, and $x \sqsubseteq y$ as desired.
			\item [(Unit closure)] We want to show that for any $e \in E$ and $e' \sqsupseteq e$, $e' \in E$. This is evident because $E = M$ and $M$ is closed under $\oplus$ and $\odot$.
			\item [(Reverse exchange)]
				Given $x = y \oplus z$ and $y = y_1 \odot y_2$, $z = z_1 \odot z_2$, we want to show that there exists $u = y_1 \oplus z_1$, $v = y_2 \oplus z_2$, and $x = u \odot v$.

			After substitution, we get $(y_1 \odot y_2) \oplus (z_1 \odot z_2) = y \oplus z = x$.
By~\ref{revexeq} and~\cref{exch_revex},  when $(y_1 \odot y_2) \oplus (z_1 \odot z_2)$ is defined, $ (y_1 \oplus z_1) \odot ( y_2 \odot  z_2) $ is also defined, and
			\(
			(y_1 \odot y_2) \oplus (z_1 \odot z_2) = (y_1 \oplus z_1) \odot ( y_2 \oplus  z_2)
			\).
Thus $ (y_1 \oplus z_1) \odot ( y_2 \oplus  z_2)  = y \oplus z = x $, and thus $u = y_1 \oplus z_1$, $v = y_2 \oplus z_2$ completes the proof. \qedhere
\end{description}
\end{proof}

\begin{lemma}	[Classical flavor in intuitionistic model]
	\label{findexact}
	For any $\T$-model M such that Disintegration holds (see~\cref{lemma:condfirst} and~\cref{Mp:lemma:condfirst}), and $f \in M$,
		\[
		f \models  \pair{\emptyset}{Z} \depand (\pair{Z}{X} \sepand \pair{Z}{Y})
		\]
if and only if there exist $g, h,i \in M$, such that $g\colon \R\emptyset \rightarrow \TR{Z}$, $h\colon \R{Z} \rightarrow \TR{Z \cup X}$, $i\colon \R{Z} \rightarrow \TR{Z \cup Y}$, and $g \odot (h \oplus i) \sqsubseteq f$.
\end{lemma}
\begin{proof}
The backwards direction trivially follows from persistence. We detail the proof for the forward direction here. Suppose
$f \models \pair{\emptyset}{Z} \depand (\pair{Z}{X} \sepand \pair{Z}{Y})$. Then, there exist $f_1, f_2, f_3, f_4$ such that $f_1 \odot f_2 = f$, $f_3 \oplus f_4 \sqsubseteq f_2$, $f_1 \models \pair{\emptyset}{Z}$, $f_3 \models \pair{Z}{X}$ and $f_4 \models \pair{Z}{Y}$.
\begin{itemize}
\item	$f_1 \models \pair{\emptyset}{Z}$ implies that there exists $f_1'' \sqsubseteq f_1$ such that $\dom(f_1'') = \emptyset$, and $\range(f_1'') \supseteq Z$.
Let $f_1' = \pi_{Z} f_1''$. Note that $f_1'\colon \R\emptyset \rightarrow \TR{Z}$ and  $f_1' \sqsubseteq f_1'' \sqsubseteq f_1$.
Hence, there exists some set $S_1$ and $v_1 \in M$ such that
$f_1 = (f_1' \oplus \unit_{S_1}) \odot v_1$.
\item $f_3 \models \pair{Z}{X}$ implies that there exists $f_3'' \sqsubseteq f_3$ such that $\dom(f_3'') = Z$, and $\range(f_3'') \supseteq X$. Define $f_3' = \pi_{Z \cup X} f_3''$. Then $f_3' \sqsubseteq f_3'' \sqsubseteq f_3$, and $f_3'\colon \R{Z} \rightarrow \TR{X\cup Z}$.
\item $f_4 \models \pair{Z}{Y}$ implies that there exists $f_4'' \sqsubseteq f_4$ such that $\dom(f_4'') = Z$, and $\range(f_4'') \supseteq Y$.
	Define $f_4' = \pi_{Z \cup Y} f_4''$ and note that $f_4'\colon \R{Z} \rightarrow \TR{Y\cup Z}$.
\item
	By Downwards closure of $\oplus$ (\Cref{downwards_oplus}), having $f_3 \oplus f_4$ defined implies that $f_3' \oplus f_4'$ is also defined and $f_3' \oplus f_4' \sqsubseteq f_3 \oplus f_4 \sqsubseteq f_2$.
	Thus, there exists some $v_2 \in M$ and finite set $S_2$ such that  $f_2 = (f_3' \oplus f_4' \oplus \unit_{S_2}) \odot v_2$.
\end{itemize}
Using these observations, we can now calculate and show that $f_1'  \odot (f_3' \oplus f_4'  \oplus  \unit_{Z}) \sqsubseteq f_1 \oplus f_2$:
{\small	\begin{align*}
		&f_1 \odot f_2 \\
		&= (f_1' \oplus \unit_{S_1}) \odot v_1 \odot (f_3' \oplus f_4' \oplus \unit_{S_2}) \odot v_2 \\
		&= (f_1' \oplus \unit_{S_1}) \odot \big(  f_3' \oplus f_4' \oplus  v_1 \big) \odot v_2 \tag{By~\cref{odot2oplus} and $\dom(f_3' \oplus f_4') = Z \subseteq \range(f_1' \oplus \unit_{S_1})$ }\\
		&= (f_1' \oplus \unit_{S_1}) \odot \big(  (f_3' \oplus f_4'  \oplus  \unit_{\dom(v_1)}) \odot (\unit_{X \cup Y\cup Z} \oplus v_1) \big) \odot v_2 \tag{By~\cref{oplus2odot}}\\
		&= (f_1' \oplus \unit_{S_1}) \odot (f_3' \oplus f_4'  \oplus  \unit_{Z} \oplus \unit_{S_1}) \odot (\unit_{X \cup Y\cup Z} \oplus v_1) \odot v_2\tag{By $\dom(v_1) = Z \cup S_1$} \\
		&= \big( (f_1'  \odot (f_3' \oplus f_4'  \oplus  \unit_{Z})) \oplus (\unit_{S_1} \odot \unit_{S_1}) \big) \odot (\unit_{X \cup Y\cup Z} \oplus v_1) \odot v_2\tag{By~\cref{revexeq} and \cref{exch_revex}}\\
		&= \big( (f_1'  \odot (f_3' \oplus f_4'  \oplus  \unit_{Z})) \oplus \unit_{S_1} \big) \odot (\unit_{X \cup Y\cup Z} \oplus v_1) \odot v_2 \\
		&= \big( (f_1'  \odot (f_3' \oplus f_4')) \oplus \unit_{S_1} \big) \odot (\unit_{X \cup Y\cup Z} \oplus v_1) \odot v_2 \tag{Because $f_3', f_4'$ preserves input on $Z$}
	\end{align*}
	}
To finish, take $g\! =\! f_1'\colon \R\emptyset \rightarrow \TR{Z}$, $h\! =\!  f_3' \colon \R{Z} \rightarrow \TR{Z \cup X}$, $i \! =\!  f_4' \colon \R{Z} \rightarrow \TR{Z \cup Y}$, and note that $g\odot (h\oplus i) = f_1'  \odot (f_3' \oplus f_4') \sqsubseteq f_1 \oplus f_2 \sqsubseteq f$.
\end{proof}

\begin{lemma}[Uniqueness]
\label{uniqueness}
For any $\T$-model $M$,  $f,g \colon \R{X} \to \TR{X\cup Y}$ in $M$, and arbitrary $h\in M$, if $f \sqsubseteq h$ and $g \sqsubseteq h$, then $f=g$.
\end{lemma}
		\begin{proof}
			$f \sqsubseteq h$ implies that there exists $v_1, {S_1}$ such that $(f \oplus \unit_{S_1}) \odot v_1 = h$;
			$g \sqsubseteq h$ implies that there exists $v_2, {S_2}$ such that $(g \oplus \unit_{S_2}) \odot v_2 = h$. Take $h\colon \m{W} \rightarrow \TR{Z\cup W}$,
			and then
			\begin{align*}
				f \oplus \unit_{S_1} &= \pi_{\range(f \oplus \unit_{S_1})} h =  \pi_{X \cup Y \cup \dom(h)} h \\
				g \oplus \unit_{S_2} &= \pi_{\range(g \oplus \unit_{S_2})} h =  \pi_{X \cup Y \cup \dom(h)} h
			\end{align*}
			Thus, $ f \oplus \unit_{S_1} = g \oplus \unit_{S_2}$. Now, suppose $f \neq g$. This would imply $f \oplus \unit_{S_1} \neq g \oplus \unit_{S_2}$ which is a contradiction. Thus, $f=g$.
		\end{proof}
\fi
\end{document}